\documentclass[a4paper]{article}
\usepackage[utf8]{inputenc}
\usepackage[T1]{fontenc}
\usepackage{textcomp, amsmath, amssymb}
\usepackage{latexsym}
\usepackage{amsthm}
\usepackage[german,english]{babel}
\usepackage{graphicx}
\usepackage{mathrsfs}
\usepackage[a4paper]{geometry}
\usepackage[usenames,dvipsnames,cmyk]{xcolor}
\usepackage[margin=1cm,font={small}]{caption}
\usepackage{setspace}
\usepackage{empheq}
\definecolor{mygrey}{cmyk}{0, 0, 0, 0.07}          % graue Farbe definieren
\newcommand*\widefbox[1]{%                          % graue Box
\colorbox{mygrey}{\hspace{1em}#1\hspace{1em}}} 
\newcommand*\greybox[1]{%                           % graue schmale Box
\colorbox{mygrey}{\hspace{0.5em}#1\hspace{0.5em}}} 
\newcommand*\widebox[1]{%                           % Kasten um Formel
\fbox{\hspace{1em}#1\hspace{1em}}} 
\newcommand*\smallerbox[1]{%                       % schmalerer Kasten um Formel
\fbox{\hspace{0.5em}#1\hspace{0.5em}}}

\usepackage[numbers,sort&compress]{natbib}    
\newtheoremstyle{plain2}% name 
{}% Space above 
{}% Space below 
{\itshape}% Body font 
{}% Indent amount
{{\sffamily}\bfseries}% Theorem head font 
{.}% Punctuation after theorem head 
{\newline}% Space after theorem head
{}% Theorem head spec (can be left empty, meaning `normal'

\newtheoremstyle{definition2}% name 
{}% Space above 
{}% Space below 
{\itshape}% Body font 
{}% Indent amount
{{\sffamily}\bfseries}% Theorem head font 
{.}% Punctuation after theorem head 
{\newline}% Space after theorem head
{}% Theorem head spec (can be left empty, meaning `normal'

\newtheoremstyle{remark2}% name 
{}% Space above 
{}% Space below 
{\slshape}% Body font 
{}% Indent amount
{{\sffamily}\bfseries}% Theorem head font 
{.}% Punctuation after theorem head 
{\newline}% Space after theorem head
{}% Theorem head spec (can be left empty, meaning `normal'

\swapnumbers
\numberwithin{equation}{section}  %section
\theoremstyle{plain2}
\newtheorem{thm}{Theorem}[section] % section
\newtheorem{lem}[thm]{Lemma}
\newtheorem{prop}[thm]{Proposition}
\newtheorem{cor}[thm]{Corollary}

\theoremstyle{definition2}
\newtheorem{defn}[thm]{Definition}

\theoremstyle{remark2}
\newtheorem{rem}[thm]{Remark}

\theoremstyle{plain}

\newcommand{\p}[0]{\partial}
\DeclareMathOperator{\grad}{\nabla}
\newcommand{\gradx}[0]{\grad_x}
\newcommand{\grady}[0]{\grad_y}
\let \div \relax
\DeclareMathOperator{\div}{div}
\newcommand{\divx}[0]{\div_x }
\newcommand{\divy}[0]{\div_y}
\newcommand{\ud}{\ \mathrm{d}}

\newcommand{\weak}{\longrightharpoonup\,}
\newcommand{\two}{\overset{\text{\tiny 2}}{\longrightharpoonup\,}}

\newcommand{\longrightharpoonup}{\mathrel{\relbar\joinrel\rightharpoonup}}

\renewcommand{\O}{\Omega}
\newcommand{\nx}[2]{\left\lVert#1\right\rVert_\text{\tiny{$#2$}}}
\newcommand{\nnormalx}[2]{\lVert#1\rVert_\text{\tiny{$#2$}}}

\newcommand{\eps}{\varepsilon}
\newcommand{\xe}{\frac{x}{\eps}}

\newcommand{\Oe}{\Omega^\varepsilon}
\newcommand{\ue}{u^\eps}
\newcommand{\pe}{p^\eps}
\newcommand{\R}{\mathbb{R}}
\newcommand{\N}{\mathbb{N}}
\newcommand{\Z}{\mathbb{Z}}
\newcommand{\C}{\mathcal{C}}
\newcommand\norm[1]{\lVert#1\rVert}
\DeclareMathOperator{\Id}{Id}

\DeclareMathOperator{\tr}{tr}
\DeclareMathOperator{\range}{im}
\DeclareMathOperator{\curl}{curl}  % Curl von Vektor
\DeclareMathOperator{\Curl}{Curl} % Curl von Skalar
\DeclareMathOperator{\curlt}{\widetilde{curl}} % Transformierter Curl von Vektor
\DeclareMathOperator{\Curlt}{\widetilde{Curl}} % Transformierter Curl von Skalar

\newcommand{\ZBL}{Z_{\mathrm{BL}}}

\newcommand{\ueff}{u^\mathrm{eff}}
\newcommand{\peff}{p^\mathrm{eff}}
\newcommand{\Ulim}{\mathcal{U}^\mathrm{lim}}
\newcommand{\Plim}{\mathcal{P}^\mathrm{lim}}
\newcommand{\Cbl}{C^\mathrm{bl}}
\newcommand{\Cblo}{C^\mathrm{bl}_\omega}
\newcommand{\bbl}{\beta^{\mathrm{bl}}}
\newcommand{\gbl}{\gamma^{\mathrm{bl}}}
\newcommand{\mubl}{\mu^{\mathrm{bl}}}
\newcommand{\lbl}{\lambda^{\mathrm{bl}}}
\newcommand{\kbl}{\kappa^{\mathrm{bl}}}
\newcommand{\Cblg}{C^\text{bl}_\gamma}
\newcommand{\Cblmu}{C^\text{bl}_\mu}
\newcommand{\Cbll}{C^\text{bl}_\lambda}
\newcommand{\Cblk}{C^\text{bl}_\kappa}
\newcommand{\obl}{\omega^{\mathrm{bl}}}
\newcommand{\bble}{\beta^{\mathrm{bl},\eps}}
\newcommand{\oble}{\omega^{\mathrm{bl},\eps}}
\newcommand{\Cqb}{C^Q_\beta}
\newcommand{\Qeb}{Q^\eps_\beta}
\newcommand{\Cqg}{C^Q_\gamma}
\newcommand{\Qeg}{Q^\eps_\gamma}
\newcommand{\Cql}{C^Q_\lambda}
\newcommand{\Qel}{Q^\eps_\lambda}
\newcommand{\bblo}{\beta^{\mathrm{bl},1}}
\newcommand{\bblz}{\beta^{\mathrm{bl},2}}
\newcommand{\oblo}{\omega^{\mathrm{bl},1}}
\newcommand{\oblz}{\omega^{\mathrm{bl},2}}
\newcommand{\Cbli}{C^{\mathrm{bl},i}}
\newcommand{\Cbloi}{C^{\mathrm{bl},i}_\omega}

  \vfuzz \hfuzz 
\title{On the Beavers-Joseph-Saffman boundary condition for curved interfaces}
\author{Sören Dobberschütz\footnote{Nano-Science Center, University of Copenhagen, Universitetsparken 5, 2100 København, Denmark.
\mbox{E-Mail}:~\texttt{sdobber@nano.ku.dk}}}

\begin{document}
\maketitle

\begin{abstract}
\emph{This document is an extended version of the results presented in S. Dobberschütz: Effective Behaviour of a Free Fluid in Contact with a Flow in a Curved Porous Medium, SIAM Journal on Applied Mathematics, 2015.}

The appropriate boundary condition between an unconfined incompressible viscous fluid and a porous medium is given by the law of Beavers and Joseph. The latter has been justified both experimentally and mathematically, using the method of periodic homogenisation. However, all results so far deal only with the case of a planar boundary. In this work, we consider the case of a curved, macroscopically periodic boundary. By using a coordinate transformation, we obtain a description of the flow in a domain with a planar boundary, for which we derive the effective behaviour: The effective velocity is continuous in normal direction. Tangential to the interface, a slip occurs. Additionally, a pressure jump occurs. The magnitude of the slip velocity as well as the jump in pressure can be determined with the help of a generalised boundary layer function. The results indicate the validity of a generalised law of Beavers and Joseph, where the geometry of the interface has an influence on the slip and jump constants.
\end{abstract}

%!TEX encoding = UTF-8 Unicode
%!TEX root = Report.tex

\tableofcontents

\section{Introduction}

A now classical result in the theory of homogenization states that, starting with the Stokes or Navier-Stokes equation, the effective fluid flow in a porous medium is given by Darcy's law (see the works of Tartar in \cite{sanpal}, Allaire in \cite{hor_homog} and Mikeli\'{c} \cite{mi_nsgrain}). When dealing with porous bodies inside another fluid, the boundary condition coupling the free fluid flow and the Darcy flow at the porous-liquid interface is of great interest.
%
%The boundary condition coupling a free fluid flow and a Darcy flow at the interface of a porous medium is of great interest in in the modelling and mathematical analysis of porous media. 
However, due to the different nature of the governing equations, the derivation of a `natural' boundary condition is difficult: While the equation for the Darcy velocity consists of a second order equation for the pressure and a first order equation for the velocity, the system of equations governing the free fluid velocity (e.g.\ the Stokes or Navier-Stokes equation) is of second order for the velocity and of first order for the pressure. \\
For an incompressible fluid, the flow in the direction normal to the interface has to be continuous due to mass conservation. However, additional conditions at the interface are not clearly available.
%

%We want to develop a new approach for the derrivation of the behaviour of a fluid flow at the boundary of a porous medium. 

From a mechanical point of view, Beavers and Joseph \cite{bejo_bc} concluded by practical experiments that a jump in the effective velocity appears in tangential direction. This jump is given by
\begin{equation}
\label{eq:jumpcond}
 \alpha {K^{-\frac{1}{2}}} (v_F - v_D)\cdot \tau   =    (\grad v_F \nu) \cdot \tau,
\end{equation}
where $v_F$ denotes the velocity of the free fluid, $v_D$ denotes the effective Darcy velocity in the porous medium and $K$ is the permeability of the porous medium. The factor $\alpha$ is the so-called \emph{slip coefficient} which has to be determined experimentally. Moreover, $\nu$ and $\tau$ are the unit normal and tangential vector with respect to the interface separating the porous medium and the free fluid. The Darcy velocity in the absence of outer forces for given fluid viscosity $\mu$ is given by
\[
v_D=-\frac{1}{\mu} K \grad p,
\]
where $p$ denotes the pressure. Note that the condition mentioned above gives a relation between the velocity of the free fluid at the interface and the \emph{effective} velocity inside the porous medium -- it does not impose a condition on the actual fluid velocity inside the porous medium at the interface. Figure \ref{fig:jumpbc} contains a schematic illustration.

%%----
%\begin{figure}[htb]
%\vspace{0.5cm}
%\centering
%\includegraphics[height=0.497\textwidth]{Illustrationen/Sprungrandbedingung}
%\caption[Illustration of the fluid velocity around a porous-liquid interface.]{Schematic illustration of the velocity profile for a horizontal flow in a domain consisting of an impermeable upper boundary (with no-slip condition), a free fluid part and a porous region. $v_F$ denotes the velocity in the free fluid domain, whereas $v_D$ is the effective Darcy velocity. The quantity $\Delta v = v_F|_\Sigma-v_D$ corresponds to the jump across the interface as discussed in Equation \eqref{eq:jumpcond}.}
%\label{fig:jumpbc}
%\end{figure}
%%----

Later, Saffman used a statistical model to derive the boundary condition of Beavers and Joseph. In \cite{sa_bcpor}, he argued that $v_D\cdot \tau$ is of lesser order than the other terms and arrived at a jump given by
\begin{equation}
v_F\cdot \tau = \frac{1}{\alpha} K^\frac{1}{2} (\grad v_F \nu) \cdot \tau + \mathscr{O}(K). \label{eq:bjjump}
\end{equation}
% neu: Ochoa Tapia, Whitaker
Other boundary conditions were proposed as well: Ochoa-Tapia and Whitaker for example used the REV-method to obtain that the velocity and pressure as well as the normal stress are continuous over the porous-liquid interface, but a jump appears in the tangential stress in the form
\[
\bigl(\grad \langle v_D \rangle \nu - \grad \langle v_F \rangle\nu \bigr) \cdot \tau = \beta K^{-\frac{1}{2}} \langle v_D \rangle \cdot \tau  .
\]
Here $\langle v_F \rangle$ denotes the averaged free fluid velocity, which is given by a Stokes equation, and $\langle v_D \rangle$ is the averaged velocity in the porous medium, which in this case fulfills a Darcy law with Brinkman correction,
\[
\langle v_D \rangle = -\frac{1}{\mu} K \bigl(\grad \langle p \rangle - \mu_B\Delta \langle v_D \rangle\bigr).
\]
$\mu_B$ is a known constant, and the dimensionless factor $\beta$ has to be determined experimentally. For details see \cite{ochtapwhi_momtransf1} and \cite{ochtapwhi_momtransf2}.

However, a rigorous mathematical derivation of the effective fluid behavior at the boundary was not available until J\"ager and Mikeli\'{c} applied the theory of homogenization to the problem. \\
In \cite{jami_bc-fluidpor} they developed a mathematical boundary layer together with several corrector terms, which allowed them to justify a jump boundary condition. The main tool was the construction of several `boundary layer functions': These functions have a given value at the interface and decay exponentially outside it. They can be used to correct the influence of spurious terms at the boundary, stemming from the contributions of other functions to the fluid velocity and pressure. 

In \cite{jami_ibc-bjs}, this theory was applied to give a mathematical proof of the Saffman modification of the boundary condition of Beavers and Joseph (see also Section 4 of the Chapter  ``Homogenization Theory and Applications to Filtration through Porous Media''  in \cite{esfami_filtr} for a more comprehensible, simplified version of the proofs), yielding the condition
\[
\eps\  (\grad v_F \nu) \cdot \tau = \alpha v_F \cdot \tau +\mathscr{O}(\eps^2)
\]
where $\alpha=-\frac{1}{\eps C_D}$ can be calculated explicitely. The constant $C_D$ stems from a boundary layer problem for the Stokes equation, cf.\ \cite{jami_ibc-bjs}. 
%where the slip coefficient $\alpha$ can explicitly be calculated using the well known cell problem for the Stokes equation (see for example \cite{hor_homog}).
Numerical simulations of the boundary layer functions can be found in  \cite{jamine_lamvisc} .

These results suffer from several drawbacks: First, only a planar boundary in the form of a line or a plane is considered (this also applies to the results of Beavers, Joseph and Saffman). Therefore, the effect of a possible curvature of the interface is not known. Second, the external force on the fluid, appearing as a right hand side in the Navier-Stokes equation, had to be $0$. This issue was adressed in the recent paper \cite{mik_effpress}, together with the derivation of the next corrector for the pressure.

%----
\begin{figure}[t]
\vspace{0.5cm}
\centering
\includegraphics[height=0.5\textwidth]{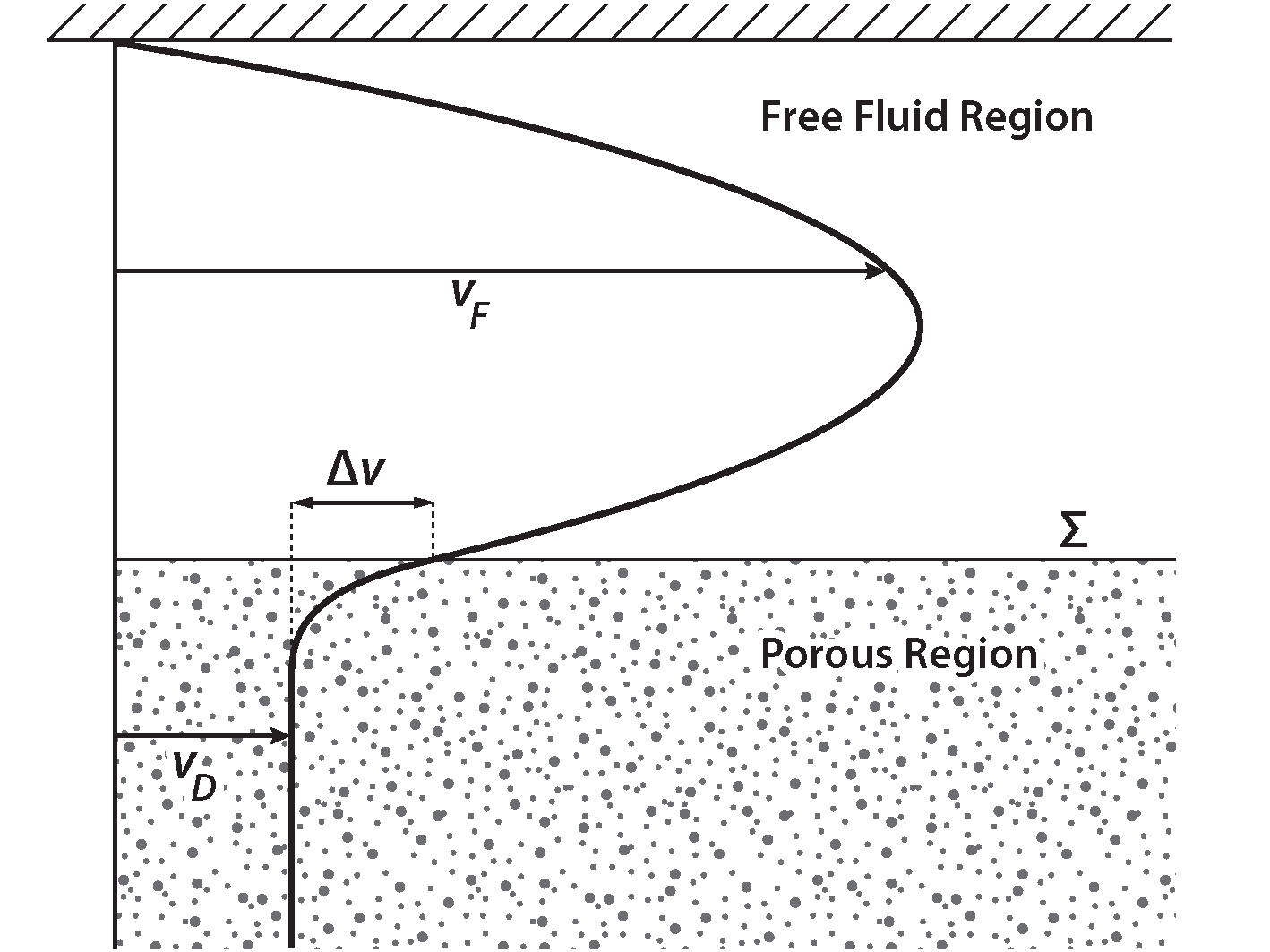}  %0.497
\caption[Illustration of the fluid velocity around a porous-liquid interface.]{Schematic illustration of the velocity profile for a horizontal flow in a domain consisting of an impermeable upper boundary (with no-slip condition), a free fluid part and a porous region. $v_F$ denotes the velocity in the free fluid domain, whereas $v_D$ is the effective Darcy velocity. The quantity $\Delta v = v_F|_\Sigma-v_D$ corresponds to the jump across the interface as discussed in Equation \eqref{eq:jumpcond}.}
\label{fig:jumpbc}
\end{figure}
%----

Generalizations of the boundary layers in \cite{jami_bc-fluidpor} were developed by  Neuss-Radu in \cite{radu_decay}. However, applications only treat reaction-diffusion systems without flow, and explicit results can only be obtained in the case of a layered medium, see \cite{radu_bbehav}.\\
The main problem which makes the treatment of general settings infeasible is the loss of exponential decay of the boundary layer functions (cf.\ Section \ref{sec:boundary}): With the generalized definition, Neuss-Radu was able to show in \cite{radu_decay} that an exponential stabilization is not possible in a general setting. However, all available tools for the treatment of these problems depend on this type of decay\footnote{Maria Neuss-Radu, private communications.}.

%

%

%In this work, a new generalization of the concept of boundary layers is proposed. It works by a transformation of the domain of interest, keeping the exponential decay.
In this work, we use the approach developed in \cite{do_dipl} and \cite{do_trapr} for providing a generalization of the law of Beavers and Joseph for curved interfaces. We closely follow \cite{mik_effpress} for investigating the effective behaviour of a free fluid in contact with flow in a curved porous medium.   
The main idea is to transform a reference geometry with a straight interface to a domain with a curved interface. It is assumed that the porous part in the reference geometry consists of a periodic array of  scaled reference cells and that the flow in the transformed geometry is governed by the stationary Stokes equation. Therefore, one obtains a set of transformed differential equations in the reference configuration. Boundary layer functions for these equations are constructed such that -- due to the straight transformed boundary -- their exponential decay can be assured. 
The difference to \cite{radu_decay} and \cite{radu_bbehav} is that in theses works, a periodic geometry was intersected by a curved interface. In this work, a periodic geometry with planar interface is transformed to give the geometry in which the fluid flow takes place. These results have been announced in \cite{do_stda} and \cite{sd_effflow}. In comparison to the latter paper, we will present the transformed equations in a more general formulation, hopefully facilitating generalizations.

\section{The Problem on the Microscale}

\subsection{Description of the Flow using Coordinate Transformations}

In this section we describe the main geometrical setting which is used throughout this work. % For illustrations, the reader is referred to Figure \ref{fig:gebietloecher} (on page \pageref{fig:gebietloecher}), Figure \ref{fig:referencecell} (on page \pageref{fig:referencecell}) and Figure \ref{fig:boundarylayer} (page \pageref{fig:boundarylayer}).
Let $L,K,h>0$. Then $\O:=(0,L)\times(-K,h)$ is a rectangular domain in $\R^2$ (later corresponding to the transformed domain) with parts $\O_1:=(0,L)\times(0,h)$ (later the transformed free fluid domain), $\O_2:=(0,L)\times(-K,0)$ (the transformed porous medium) and $\Sigma=(0,L)\times \{0 \}$ (later the transformed interface).

%-------
\begin{figure}[t]
\centering
\includegraphics[width=0.8\textwidth]{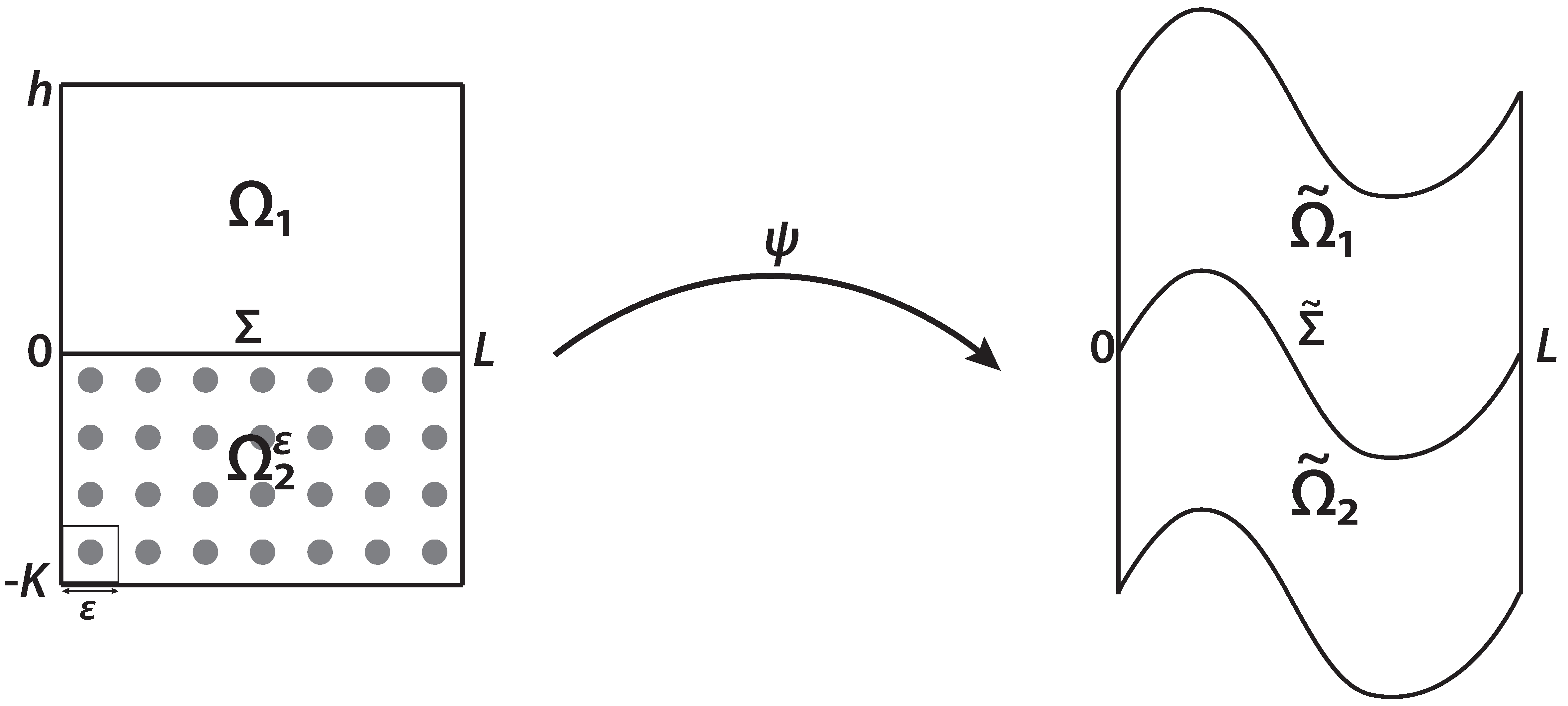}
\caption{Illustration of the coordinate transformation and the periodic domain.}
\label{fig:gebietloecher}
\end{figure}
%-------

Let ${g}\in \mathcal{C}^\infty(\R)$ be a given function such that ${g}(x+L)={g}(x)$ for all $x\in\R$. We consider $g$ to describe a periodic curved structure in our domain of interest. Define the coordinate transformation 
\begin{gather*}
\psi:{\O}\longrightarrow \tilde{\O}\\
\begin{pmatrix}
z_1\\
z_2
\end{pmatrix}
\longmapsto
\begin{pmatrix}
x_1\\
x_2
\end{pmatrix}
=
\begin{pmatrix}
z_1\\
z_2 + {g}(z_1)
\end{pmatrix}.
\end{gather*}
such that $\tilde{\O}=\psi(\O)$, $\tilde{\O}_1:=\psi(\O_1)$, $\tilde{\O}_2:=\psi(\O_2)$ and $\tilde{\Sigma}:= \psi(\Sigma)=\{ (x,g(x))| x\in (0,L) \}$. We are interested in the behavior of a fluid flowing through the curved channel $\tilde{\O}$, where $\tilde{\O}_1$ represents a domain with a free fluid flow, and $\tilde{\O}_2$ is a porous medium. We are especially interested in the behavior of the fluid at the curved boundary $\tilde{\Sigma}$. See Figure~\ref{fig:gebietloecher} for an illustration.

Let there be given a solid inclusion $\tilde{\O}_S \subset\subset \tilde{\O}_2$. We will later use a sequence of such inclusions to create a porous medium via homogenization theory. For a given volume force $\tilde{f}\in L^2(\tilde{\O})$, we assume that a mathematical description of the fluid is given by the steady state Stokes equation with no slip condition on the boundary of the solid inclusion and on the outer walls
\begin{subequations}
\label{eq:stokes}
\begin{alignat}{2}
-\mu\Delta_x \tilde{u}(x) +  \gradx \tilde{p}(x)&= \tilde{f}(x)&\quad &\text{ in $\tilde{\O}\backslash \overline{\tilde{\O}_S}$} \\
\divx(\tilde{u}(x))&=0 &&\text{ in $\tilde{\O}\backslash \overline{\tilde{\O}_S}$}\\
\tilde{u}(x)&=0 && \text { on $\p \tilde{\O}_S \cup \p \tilde{\O}\backslash( \{ x_1=0  \} \cup \{ x_1=L \} ) $} \\
\tilde{u}, \tilde{p} & \text{ are $L$-periodic in $x_1$}&&
\end{alignat}
\end{subequations}
Here $\mu>0$ denotes the dynamic viscosity; we will set $\mu=1$ in the sequel. We are looking for a velocity field $\tilde{u} \in H^1(\tilde{\O})^2$ and a pressure $ \tilde{p} \in L^2(\tilde{\O})/\R$. The Stokes equation is an approximation of the full Navier-Stokes equation which is valid for low Reynolds number flows. Using the transformation rules for the differential operators (see Appendix~\ref{sec:coordtrans}), we obtain the following equation for the transformed quantities ${u}(z)=\tilde{u}(\psi(z))$, ${p}(z)=\tilde{p}(\psi(z))$ and $f(z)=\tilde{f}(\psi(z))$ in the rectangular domain $\O$:
\begin{subequations}
\label{eq:trst}
\begin{empheq}[box=\widebox]{align}
-\div_z(F^{-1}(z)F^{-T}(z) &\grad_z {u}(z)) + F^{-T}(z) \grad_z {p}(z)= {f}(z) \text{ in ${\O\backslash \overline{\O_S}}$} \\
\div_z(F^{-1}(z){u}(z))=0 &\text{ in ${\O\backslash\overline{\O_S}}$}\\
{u}(z)=0 & \text{ on $\p {\O_S} \cup   \p\Omega\backslash  ( \{ z_1=0  \} \cup \{ z_1=L \}   )  $} \\
{u}, {p}& \text{ are $L$-periodic in $z_1$}
\end{empheq}
\end{subequations}
Here $\O_S:= \psi^{-1}(\tilde{\O}_S)$ is the transformed solid inclusion, and $F$ is defined as the Jacobian matrix of $\psi$ given by
\begin{equation}
\label{eq:F}
F(z)=
\begin{bmatrix}
1 & 0 \\
{g}'(z_1) & 1
\end{bmatrix}.
\end{equation}
Since $\det F =1$, $\psi$ is a volume preserving $\mathcal{C}^\infty$-coordinate transformation. In this connection, please note that we define the gradient of a vector field column-wise, i.e.\ $\grad u$ is the transpose of the Jacobian matrix of $u$. Defining $\grad u $ row-wise leads to slightly different transformed equations. Now, the crucial assumption is  that $\O_S$ is given as an $\eps$-periodic structure. This is described in the next subsection.

%-------
\begin{figure}[t]
\centering
\includegraphics[width=0.8\textwidth]{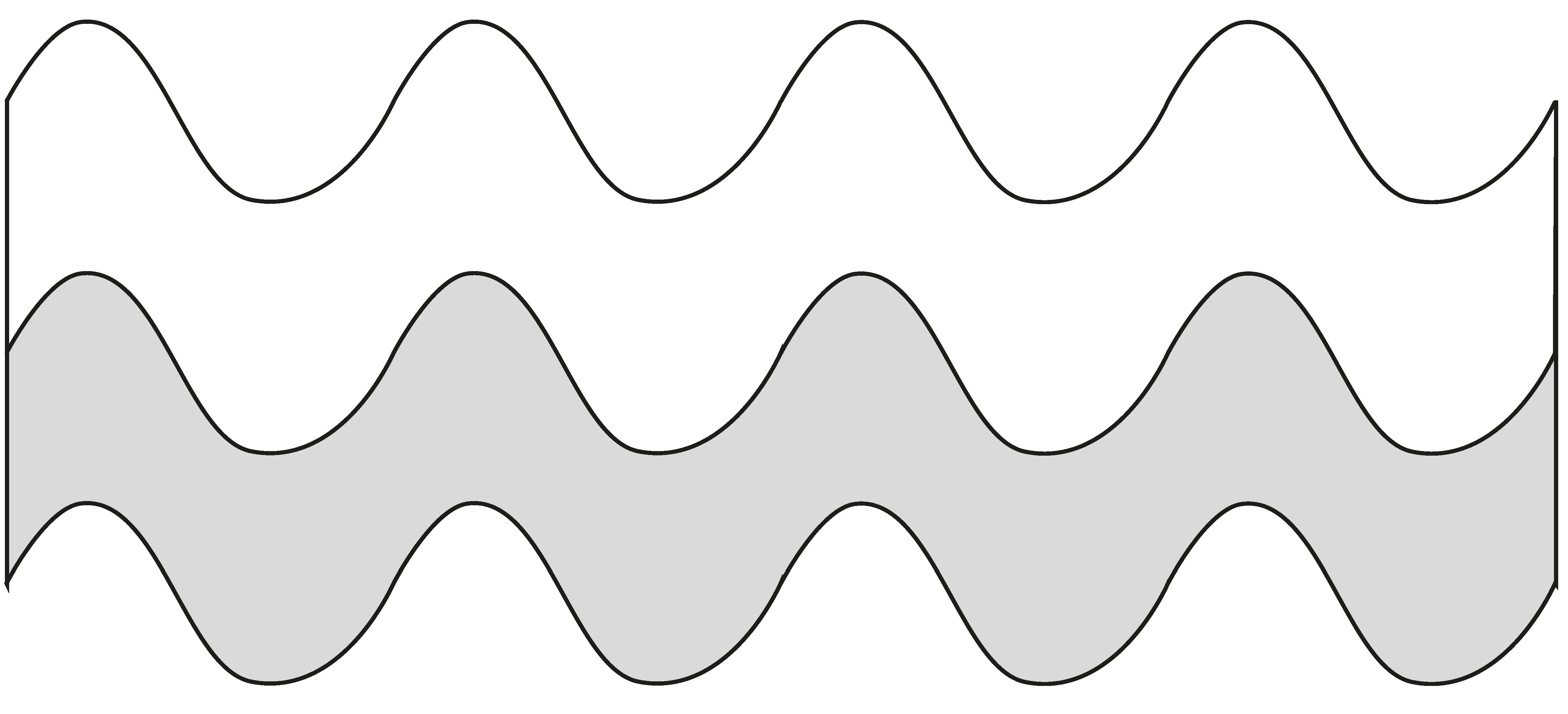}
\caption{Example of a periodic domain to which our methods can be applied. The upper part corresponds to the free fluid domain, and the lower (gray) part to a porous medium.}
\label{fig:domain}
\end{figure}
%-------

\subsection{The $\eps$-periodic Problem}
\label{sec:auxgeom}

%--
\begin{figure}[t]
\centering
\includegraphics[width=0.37\textwidth]{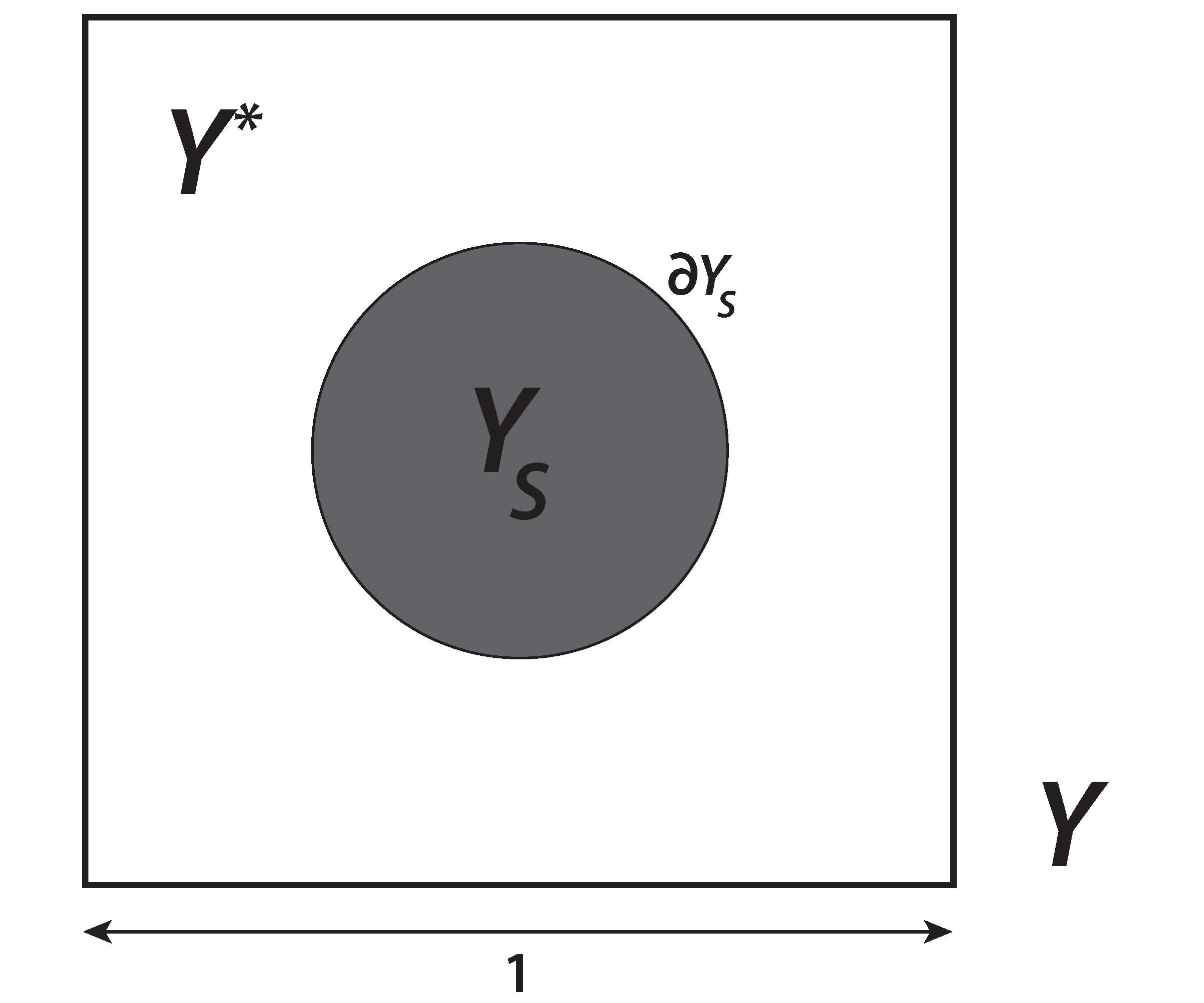}
\caption[The reference cell.]{The reference cell, consisting of the solid part $Y_S$ with boundary $\p Y_S$, and the fluid part $Y^*$.}
\label{fig:referencecell}
\end{figure}
%--

We assume an $\eps$-periodic geometry in $\O_2$: Define a reference cell as
\[
Y:=[0,1]^2,
\]
containing a connected open set $Y_S$ (corresponding to the solid part of the cell). Its boundary $\p Y_S$ is assumed to be of class $\mathcal{C}^\infty$ with $\p Y_S \cap \p Y = \emptyset$. Let \[
Y^*:=Y\backslash \overline{Y_S}
\]
be the fluid part of the reference cell. 
For given $\eps>0$ such that $\frac{L}{\eps}\in\N$, let $\chi$  be the characteristic function of $Y^*$, extended by periodicity to the whole $\R^2$. Set $\chi^\eps(x):=\chi(\xe)$ and define the fluid part of the porous medium as
\[
\Oe_2 = \{  x\in \O_2\ |\ \chi^\eps(x)=1 \}.
\]
The fluid domain is then given by
\[
\Oe= \O_1 \cup \Sigma \cup \Oe_2.
\]

In order to be able to obtain the effective fluid behavior near $\Sigma$, we have to define a number of so-called boundary layer problems, see Section \ref{sec:boundary}. To this end, we introduce the following setting: We consider the domain $[0,1]\times \R$ subdivided as follows:
\[
Z^+= [0,1] \times (0,\infty) 
\]
corresponds to the free fluid region, whereas the union of translated reference cells
\[
Z^- =\bigcup_{k=1}^\infty \{  Y^* - \binom{0}{k} \} \backslash S
\]
is considered to be the void space in the porous part. Here
\[
S = [0,1] \times \{ 0 \} 
\]
denotes the interface between $Z^+$ and $Z^-$. Finally, let 
\[
Z= Z^+ \cup Z^-
\]
and
\[
\ZBL = Z^+ \cup S \cup  Z^- 
\]
be the fluid domain without and with interface. For illustrations, the reader is referred to Figure~\ref{fig:gebietloecher}, %(on page \pageref{fig:gebietloecher}) 
Figure \ref{fig:referencecell} %(on page \pageref{fig:referencecell}) 
and Figure \ref{fig:boundarylayer}.% (page \pageref{fig:boundarylayer}).

%-------
\begin{figure}[htb]
\centering
\includegraphics[width=0.1775\textwidth]{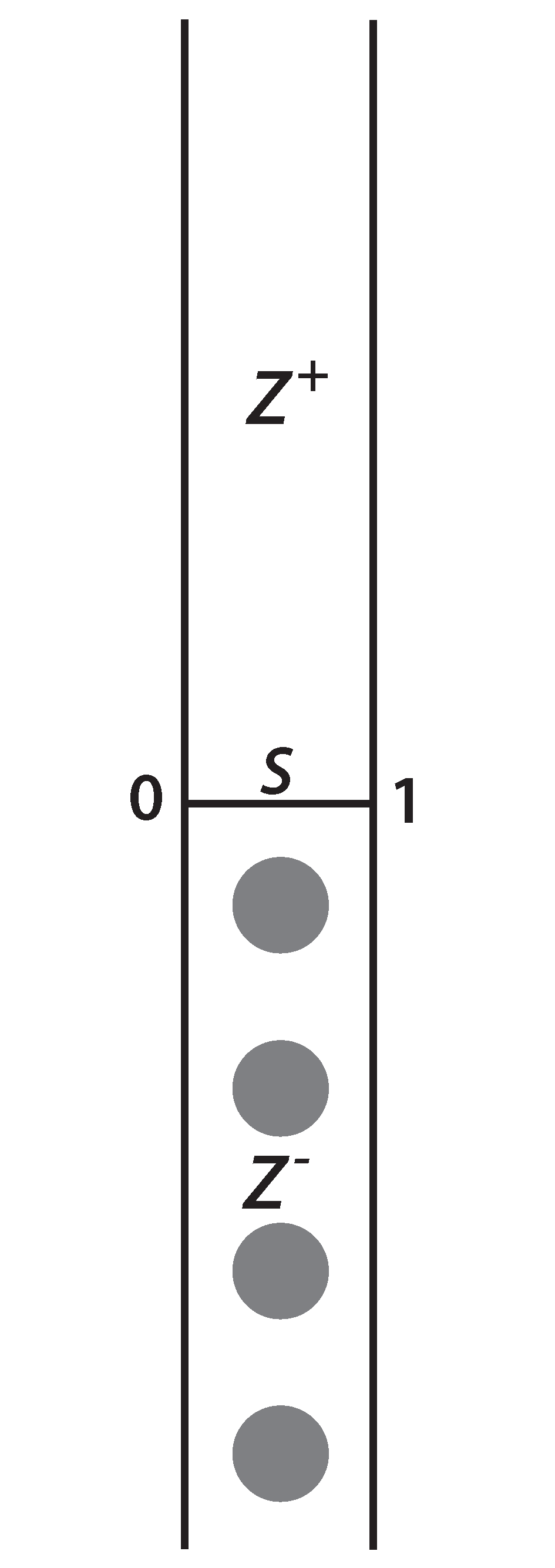}
\caption{The boundary layer strip $\ZBL$.}
\label{fig:boundarylayer}
\end{figure}
%-------

Using the constructions given above, we are interested in the limit behavior of the problem: For $\eps>0$, find a velocity $\ue$ and a pressure $\pe$ such that
\begin{subequations}
\label{eq:trsteps}
\begin{empheq}[box=\widefbox]{align}
-\div(F^{-1}F^{-T} \grad \ue) + F^{-T} \grad \pe= {f} &\text{ in $\Oe$} \label{seq:trstepsmain} \\
\div(F^{-1}\ue)=0 &\text{ in $\Oe$}\\
\ue=0 & \text{ on $\p \Oe \backslash ( \{ x_1=0  \} \cup \{ x_1=L \}   )  $} \\
\ue, \pe & \text{ are $L$-periodic in $x_1$}
\end{empheq}
\end{subequations}
Note that here and in the sequel,  we use the variable $x$ to designate points in $\O$, while $y$ is used for the strip $Z$ or the reference cell $Y$. There exists a solution $(\ue,\pe)\in H^1(\Oe)^2 \times L^2(\Oe)/\R$ of the problem above, see Appendix~\ref{sec:trst}.

\section{The effective behavior of a fluid at a curved porous medium}

The effective velocity field in the free fluid domain $\O_1$ is given by: Find a velocity $u^\text{eff}$ and a pressure $p^\text{eff}$ such that
\begin{subequations}
\label{eq:effo1}
\begin{alignat}{2}
-\div(F^{-1} F^{-T} \grad u^\text{eff}) + F^{-T} \grad p^\text{eff} &= f &\quad& \text{in } \O_1 \\
\div(F^{-1} u^\text{eff})&=0&&\text{in } \O_1 \\
\int_{\O_1} p^\text{eff} \ud x &=0\\
u^\text{eff}&=0 && \text{on } (0,L)\times \{ h\} \\
u^\text{eff}&, p^\text{eff} && \text{are $L$-periodic in $x_1$}\\
u^\text{eff}&= -\eps C^\text{bl} && \text{on } \Sigma \label{seq:bjcond}
\end{alignat}
\end{subequations}
$C^\text{bl}$ is the decay function of the boundary layer function $\bbl$ defined in Section \ref{sec:mainblfunc}. It holds $\Cbl\cdot F^{-T}e_2=0$ (see \eqref{eq:expstabvelocity}). In transformed tangential direction, a slip velocity $-\eps \Cbl\cdot F e_1$ occurs. Its magnitude can be calculated using the result in Lemma \ref{lem:expstabveloc} as
\begin{equation}
\label{eq:bj}
u^\text{eff}(x)\cdot F(x)e_1 = -\eps \Cbl(x)\cdot F(x) e_1 = - \eps\int_0^1 \beta^\mathrm{bl}(x,y_1,+0) \cdot F(x) e_1 \ud y_1.
\end{equation}
The effective Darcy pressure in $\O_2$ is given by
\begin{subequations}
\label{eq:effo2}
\begin{alignat}{2}
\div(F^{-1} A (f-F^{-T}\grad \tilde{p}^\text{eff}  )) &=0& \quad & \text{in } \O_2 \\
A (f-F^{-T}\grad \tilde{p}^\text{eff}  )\cdot F^{-T}e_2 &=0 && \text{on } (0,L) \times \{ -K\} \\
\tilde{p}^\text{eff} &= {p}^\text{eff} + C^\text{bl}_\omega && \text{on } \Sigma \\
&\tilde{p}^\text{eff}& & \text{ is $L$-periodic in $x_1$}
\end{alignat}
\end{subequations}
$A$ is given in \eqref{eq:A}, and $\Cblo$ is the pressure stabilization function defined in Section \ref{sec:mainblfunc}, given by
\begin{equation*}
%\label{eq:Cblo}
\Cblo(x_1) = \int_0^1 \obl(x_1,y_1,+0) \ud y_1.
\end{equation*}
Define the effective mass flow rates in transformed tangential direction as
\begin{gather*}
M^\mathrm{eff}:= \int_{\O_1} u^\mathrm{eff} \cdot Fe_1, \ud x \qquad
M^\eps:=\int_{\O_1} \ue \cdot Fe_1 \ud x.
\end{gather*}
We obtain the following estimates:
\begin{thm}
Let $f\in \C^\infty(\bar{\O})$ be $L$-periodic in the first variable. For $(\ue,\pe)$ as defined in \eqref{eq:trsteps} and $(u^\mathrm{eff}, p^\mathrm{eff})$ defined in \eqref{eq:effo1} the estimates
\begin{gather*}
\nx{\ue - u^\mathrm{eff}}{L^2(\O_1)^2} + |M^\eps-M^\mathrm{eff}| \leq C\eps^\frac{3}{2} \\
\nx{\ue-u^\mathrm{eff}}{H^\frac{1}{2}(\O_1)^2} + \nx{\pe-p^\mathrm{eff}}{L^1(\O_1)} + \nx{\grad(\ue - u^\mathrm{eff})}{L^1(\O_1)^4} \leq C \eps \\
\nx{|x_2|^\frac{1}{2}\grad(\ue - u^\mathrm{eff})  }{L^2(\O_1)^4} + \nx{|x_2|^\frac{1}{2}(\pe - p^\mathrm{eff})}{L^2(\O_1)} \leq C \eps
\end{gather*}
hold.
\end{thm}

In the porous medium $\O_2$, we arrive at the following results:
\begin{thm}
\label{thm:pormed}
For the effective pressure in the porous medium defined by \eqref{eq:effo2}, we have for all $\delta >0$
\begin{gather}
\frac{1}{\eps^2} \ue - A(f- F^{-T}\grad \tilde{p}^\mathrm{eff}) \weak 0 \text{ in } L^2((0,L)\times (-H, -\delta)) \notag \\
\pe - \tilde{p}^\mathrm{eff} \weak 0 \text{ in } L^2(\O_2)\notag \\
\nx{\pe - \tilde{p}^\mathrm{eff}}{H^{-\frac{1}{2}}(\Sigma)} \leq C \eps^\frac{1}{2}\label{eq:sigmaestim1}
\end{gather}
\end{thm}
\begin{thm}
\label{thm:sigma}
At the interface $\Sigma$, it holds
\begin{gather}
\frac{1}{\eps}(\ue - u^\mathrm{eff}) \weak 0 \quad \text{in } L^2(\Sigma) \label{eq:sigmaestim2} \\
\nx{\ue - u^\mathrm{eff}}{H^{-\frac{1}{2}}(\Sigma)} \leq C \eps^\frac{4}{3}. \notag
\end{gather}
\end{thm}

%WHAT ABOUT OTHER MAIN RESULTS?

These results show that the following behavior of a free fluid in contact with a flow in a curved porous medium can be expected for low Reynolds number flows:
\begin{itemize}
\item In the free fluid domain, the velocity and pressure are given by the Stokes equation.
\item In the porous medium, the flow is pressure driven and given by Darcy's law.
\item At the interface, a slip-condition occurs. The velocity is given with the help of the decay function $C^\mathrm{bl}(x)$ of an auxiliary boundary layer problem. This depends on parameters from the geometry of the interface. 
\item In transformed normal direction $C^\mathrm{bl}(x)\cdot F^{-T}(x)e_2=0$ holds, which is an approximation of continuity of the velocity in that direction. In tangential direction, a jump between the velocities and pressures occurs.  
Looking at the form of the boundary condition, we see that the generalized boundary condition of Beavers and Joseph has to incorporate effects stemming from the geometry of the fluid-porous interface. The estimates \eqref{eq:sigmaestim1} and \eqref{eq:sigmaestim2} corresponds to a generalized rigorous version of \eqref{eq:bjjump}. %Note, however, that contrary to the works \cite{jami_ibc-bjs} and \cite{mik_effpress}, we are not able to  write the function $C^\mathrm{bl}(x)$ as a product containing a factor of $\nabla u^0$. The latter velocity appears instead in the equations for $\beta^{\mathrm{bl}}$, see \eqref{seq:bljump}.
In \cite{jami_ibc-bjs} and \cite{mik_effpress}, the jump boundary condition is given by $\ueff_1 + \eps \Cbl_1 \frac{\p \ueff_1}{\p x_2}=0$, whereas in our case it reads $u^\text{eff} +\eps C^\text{bl}=0$. The missing velocity gradient appears in the equations for $\beta^{\mathrm{bl}}$, see Section~\ref{sec:mainblfunc}. 
\end{itemize}

\section{Derivation of the general law of Beavers and Joseph}
%\label{sec:boundary}

In this section we carry out the steps that are necessary to derive the generalized boundary condition of Beavers and Joseph \eqref{eq:bj} as well as the theorems given above. We will successively correct the velocity $\ue$ and the pressure $\pe$ given by \eqref{eq:trsteps} with the help of auxiliary functions. This will give us insight into the effective behavior, while at the same time introducing problems with $\div(F^{-1}\ue)$. Therefore we have to correct this term as well.

\subsection{Correction of the velocity}

\subsubsection{Elimination of the forces}

We start by eliminating the right hand side of \eqref{seq:trstepsmain} in $\O_1$. Let $u^0, \pi^0$ be a solution of
\begin{subequations}
\label{eq:pi0}
\begin{empheq}[box=\widebox]{align}
-\div(F^{-1}F^{-T} \grad u^0) + F^{-T} \grad \pi^0= {f} &\text{ in $\O_1$} \\
\div(F^{-1} u^0)=0 &\text{ in $\O_1$}\\
u^0=0 & \text{ on $\p \O_1 \backslash ( \{ x_1=0  \} \cup \{ x_1=L \}   )  $} \\
u^0, \pi^0 & \text{ are $L$-periodic in $x_1$}
\end{empheq}
\end{subequations}
There exists a unique solution $u^0\in H^1(\O_1)^2$, $\pi^0\in L^2(\O_1/\R)$ by the results for the transfomed Stokes equation. By regularity results (see e.g.\ \cite{temam}), this solution is smooth for smooth $f$. We extend the velocity $u^0$ by $0$ in $\O_2$ and the pressure in a smooth manner to a pressure $\tilde{\pi}^0$ defined in all of $\O$. As it will turn out, $\tilde{\pi}^0$ will be given by the Darcy pressure in $\O_2$. Details of this extension procedure will be given below (see Section~\ref{sec:proof}).

To obtain estimates for $\ue-u^0$, we need the following Lemma from \cite{jami_bc-fluidpor}, which is a variant of the Poincar\'e inequality:
\begin{lem}
\label{lem:poincare}
Let $\phi\in \{ g\in H^1(\Oe_2) | g=0 \text{ on } \p \Oe_2 \backslash \p \O_2 \}$. Then
\begin{gather*}
\nx{\phi}{L^2(\Sigma)} \leq C \sqrt{\eps} \nx{\grad \phi}{L^2(\Oe_2)^2} \qquad \text{as well as} \qquad
\nx{\phi}{L^2(\Oe_2)} \leq C \eps\nx{\grad \phi}{L^2(\Oe_2)^2} .
\end{gather*}
\end{lem}
Define the space $W^\eps=\{g\in H^1(\Oe)^2\ |\ g=0 \text{ on } \p \Oe \backslash(\{ x_1=0 \} \cup \{x_1=L \}), g \text{ is $L$-periodic in $x_1$} \}$; then by subtraction the weak formulations, one easily sees that $\ue-u^0$ satisfies the variational equation
\begin{equation}
\begin{gathered}
\int_{\Oe} F^{-T}(\grad\ue - \grad u^0):F^{-T}\grad \phi \ud x - \int_{\Oe} (\pe - \pi^0) \div(F^{-1}\phi) \ud x \\ = \int_{\Sigma} F^{-1} F^{-T} \grad u^0 e_2 \cdot \phi \ud \sigma_x 
 - \int_\Sigma F^{-1}[\tilde{\pi}^0]_\Sigma \phi \cdot e_2 \ud \sigma_x + \int_{\Oe_2}(f-F^{-T}  \grad \tilde{\pi}^0)\phi \ud x
\end{gathered} \label{eq:velocweak}
\end{equation}
for all $\phi \in W^\eps$. Here $[q]_\Sigma:=q|_{\O_1}- q|_{\O_2}$ denotes the jump of the function $q$ across the boundary~$\Sigma$.
\begin{prop}
\label{prop:basicestim}
We have the estimate
\begin{gather*}
\sqrt{\eps}\nx{\grad \ue - \grad u^0}{L^2(\Oe)^4} + \frac{1}{\sqrt{\eps}} \nx{\ue}{L^2(\Oe_2)^2} + \nx{\ue}{L^2(\Sigma)} \leq C \eps.
\end{gather*}
\end{prop}
\begin{proof}
Choose $\phi= \ue - u^0$ in \eqref{eq:velocweak}. Since for $u,v\in H^1(\Oe)$ the form $a(u,v):=\int_{\Oe} F^{-T}\grad u:F^{-T} \grad v \ud x$ is bounded and coercive, we obtain
\begin{align*}
\nx{\grad \ue - \grad u^0}{L^2(\Oe)^4}^2 &\leq C \int_{\Oe} |F^{-T}(\grad \ue - \grad u^0)|^2 \ud x\\
&\leq \biggl|\int_{\Sigma} F^{-1} F^{-T} \grad u^0 e_2 \cdot (\ue - u^0) \ud \sigma_x \biggr|
 +\biggl| \int_\Sigma F^{-1}[\tilde{\pi}^0]_\Sigma (\ue - u^0) \cdot e_2 \ud \sigma_x\biggr| \\
 & \qquad +\biggl| \int_{\Oe_2}(f-F^{-T}  \grad \tilde{\pi}^0)(\ue - u^0) \ud x\biggr|.
\end{align*}
Since $\grad u^0$ and $[\tilde{\pi}^0]_\Sigma$ are smooth, the first two terms on the right hand side are bounded by 
\begin{gather*}
C \nx{\ue - u^0}{L^2(\Sigma)}\leq C\sqrt{\eps} \nx{\grad \ue - \grad u^0}{L^2(\Oe_2)^4} \leq C(\delta) \eps + C \delta \nx{\grad \ue - \grad u^0}{L^2(\Oe_2)^4}^2.
\end{gather*}
The function $(f-F^{-T}  \grad \tilde{\pi}^0)$ is smooth as well by elliptic regularity results for the Darcy pressure, so the last term is bounded by
\begin{gather*}
C \nx{\ue-u^0}{L^2(\Oe_2)^2} \leq C \eps \nx{\grad \ue - \grad u^0}{L^2(\Oe_2)^4} \leq C(\delta) \eps^2 + C  \delta \nx{\grad \ue - \grad u^0}{L^2(\Oe_2)^4}^2.
\end{gather*}
Choosing $\delta$ small enough we arrive at $\nx{\grad \ue - \grad u^0}{L^2(\Oe)^4}^2 \leq C \eps$. Moreover, we have
\begin{gather*}
\frac{1}{\eps}\nx{\ue}{L^2(\Oe_2)^2}^2 \leq C \eps\nx{\grad \ue}{L^2(\Oe_2)^4}^2 = C \eps \nx{\grad \ue- \grad u^0}{L^2(\Oe_2)^4}^2 \leq C\eps^2
\end{gather*}
as well as
\begin{gather*}
\nx{\ue}{L^2(\Sigma)^2}^2 \leq C \eps\nx{\grad \ue- \grad u_0}{L^2(\Oe_2)^4}^2 \leq C \eps^2.
\end{gather*}
This finishes the proof.
\end{proof}

In order to derive estimates for the pressure and $\ue-u^0$ in $\O_1$, we use the theory of very weak solutions for the transformed Stokes equations, see Appendix~\ref{sec:vws}.
\begin{lem}
It holds
\[
\sqrt{\eps} \nx{\pe - \pi^0}{L^2(\O_1)} + \nx{\ue - u^0}{L^2(\O_1)^2} \leq C\eps.
\]
\end{lem}
\begin{proof}
$\ue-u^0$ is a very weak solution with $G_1=G_2=0$, $\xi=\ue|_{\Sigma_T}$, $\Sigma_T= \Sigma \cup \{ x_2=h\}$. Using the estimate \eqref{eq:vwsestim} from the theory of very weak solutions of the transformed Stokes equation, we obtain
\[
\nx{\ue - u^0}{L^2(\O_1)^2} \leq C \nx{\xi}{L^2(\Sigma_T)} = C \nx{\ue}{L^2(\Sigma)} \leq C\eps
\]
by the previous Proposition. Using the Ne\v{c}as inequality yields
\begin{gather*}
\sqrt{\eps} \nx{\pe - \pi^0}{L^2(\O_1)} \leq C \sqrt{\eps}( \nx{\ue - u^0}{L^2(\O_1)^2} + \nx{\grad \ue - \grad u^0}{L^2(\O_1)^4}    ) \leq C \eps.\qedhere
\end{gather*}
\end{proof}

\subsubsection{Continuity of the traces}
\label{sec:mainblfunc}

Looking at the right hand side of equation~\eqref{eq:velocweak} and the proof of Proposition~\ref{prop:basicestim}, one can see that the expression $\int_{\Oe_2}(f-F^{-T}  \grad \tilde{\pi}^0)\phi \ud x$ allows estimates for $\nx{\grad \ue - \grad u^0}{L^2(\Oe)^4}$ on the order of $\eps$, whereas the other two integrals only allow for estimates on the order of $\eps^\frac{1}{2}$. Later we will choose $[\tilde{\pi}^0]_\Sigma$ to allow for better estimates. This leaves us with the expression $\int_{\Sigma} F^{-1} F^{-T} \grad u^0 e_2 \cdot \phi \ud \sigma_x$, which we are going to correct next: Construct the following parameter dependent boundary layer functions $(\bbl, \obl)$ satisfying
\begin{empheq}[box=\smallerbox]{align*}
-\divy(F^{-1}(x)F^{-T}(x) \grady \bbl(x,y)) + F^{-T}(x)\grady\obl(x,y)=0 & \quad \text{ in $\O\times Z$} \\
\divy(F^{-1}(x)\bbl(x,y))=0 & \quad\text{ in $\O\times Z$} \\
[ \bbl(x,y) ]_S= 0 & \quad\text{ on $\O\times S$} \\
[(F^{-1}(x)F^{-T}(x)\grady \bbl(x,y)-F^{-1}(x)\obl(x,y))e_2]_S&\\ =F^{-1}(x)F^{-T}(x)\grad u^0(x)e_2  & \quad\text{ on $\O\times S$} \\
\bbl(x,y)=0 &\quad\text{ on } \O \times {\textstyle\bigcup_{k=1}^\infty \{  \p Y_S - \binom{0}{k} \} }\\
\bbl(x,\cdot), \obl(x,\cdot) \text{ are $1$-periodic in $y_1$}
\end{empheq}
and define $\bble(x)= \eps\bbl(x, \frac{x}{\eps}) $ as well as $\oble(x)=\obl(x,\frac{x}{\eps})$.

By the theory of the boundary layer functions (see Appendix~\ref{sec:boundary}), there exists decay functions $\Cbl: \R^2 \longrightarrow\R^2$, $\Cblo: \R^2 \longrightarrow \R$ such that
\begin{gather*}
\frac{1}{\eps}\nx{\bble-\eps\Cbl(x)H(x_2)}{L^q(\O)^2} + \nx{\oble-\Cblo(x)H(x_2)}{L^q(\O)} + \nx{\grad_y \bble}{L^q(\O)^4} \\
+ \nx{\gradx(\bbl(x,\xe)   - H(x_2) \Cbl(x)  )}{L^q(\O)^4} + \nx{\divx(F^{-1}F^{-T}  \grady \bbl(x,\xe) )}{L^q(\O)^2}
 \leq C \eps^\frac{1}{q}
\end{gather*}
Here $H$ denotes the Heaviside function. 
This correction introduces problems due to the stabilization towards $\Cbl$. Therefore we define the following counterflow:
\begin{empheq}[box=\widebox]{align*}
-\div(F^{-1}(x)F^{-T}(x) \grad u^{\sigma}(x) ) + F^{-T}(x) \grad \pi^{\sigma}(x) =0 &\quad \text{ in $\O_1$} \\
\div(F^{-1}(x)  u^{\sigma}(x)) =0 & \quad\text{ in $\O_1$} \\
 u^{\sigma}(x)= 0 &\quad\text { on $(0,L) \times \{ h\}$}\\
 u^{\sigma}(x)= \Cbl(x) &\quad\text { on $\Sigma$}\\
u^{\sigma}, \pi^{\sigma}\text{ are $L$-periodic in $x_1$}
\end{empheq}
Since $\int_\Sigma \Cbl(x)\cdot F^{-T}(x)e_2 \ud x=0$ (see Lemma \ref{lem:constexpl}), there exists a unique velocity $u^\sigma$ and a pressure $\pi^\sigma$, unique up to constants. Consider $u^0 - ( \bble  - \eps \Cbl H(x_2)) + \eps u^\sigma H(x_2)$ as a first macroscopic approximation. Define the error between this approximation and the microscopic velocity (similarly for the pressure) as
\begin{gather*}
\mathcal{U}^\eps= \ue - u^0 + ( \bble  - \eps \Cbl H(x_2)) + \eps u^\sigma H(x_2), \\
\mathcal{P}^\eps=\pe - \pi^0 H(x_2) - \tilde{\pi}^0 H(-x_2) + ( \oble - \Cblo H(x_2) ) + \eps \pi^\sigma H(x_2).
\end{gather*}
Then we have
\begin{align*}
\int_{\Oe} & F^{-T}\grad \mathcal{U}^\eps : F^{-T} \grad \phi \ud x - \int_{\Oe} \mathcal{P}^\eps \div( F^{-1}\phi) \ud x \\
& =  -\int_\Sigma F^{-1} [\tilde{\pi}^0]_\Sigma e_2 \cdot \phi \ud \sigma  - \int_\Sigma F^{-1} \Cblo e_2 \cdot \phi \ud \sigma+ \int_{\Oe_2} (f-F^{-T}\grad \tilde{\pi}^0) \phi \ud x\\
& \quad + \int_{\Oe} F^{-T} \gradx(\bble - \eps \Cbl H(x_2)) :F^{-T} \grad \phi \ud x - \int_{\Oe} \divx(F^{-1}F^{-T}\grady \bbl(x,\xe)) \phi \ud x \displaybreak[0] \\
& \quad+ \eps \int_\Sigma F^{-1} F^{-T} \grad u^\sigma e_2 \cdot \phi \ud \sigma  + \int_{\Oe} F^{-T} \gradx( \oble - H(x_2) \Cblo )\cdot \phi \ud x + \eps \int_{\Sigma} F^{-1}\pi^\sigma e_2 \cdot \phi \ud \sigma  \displaybreak[0]\\
& \quad + \int_{\{ x_2=h \}}  F^{-1} F^{-T} \gradx(\bble-\eps \Cbl H(x_2))e_2 \cdot \phi  \ud \sigma %+ \int_{\{ x_2=h \}} F^{-1} F^{-T} \grady \bbl(x,\xe) e_2 \cdot \phi  \ud \sigma 
 \\
&\quad   - \int_{\{ x_2=-H\}} F^{-1} F^{-T} \gradx(\bble-\eps \Cbl H(x_2))e_2 \ud \sigma. %  - \int_{\{ x_2=-H\}} F^{-1} F^{-T} \grady \bbl(x,\xe) e_2 \cdot \phi   \ud \sigma 
\end{align*}
We will later eliminate the first two terms on the right hand side by requiring that $[\tilde{\pi}^0]_\Sigma = -\Cblo$ (denoted as $\mathrm{(P)}$, see below). Due to the exponential decay of the boundary layer functions, the last two terms can be chosen arbitrarily small (denoted $\mathrm{(BL)}$). The terms on the right hand side can be estimated by the following orders of $\eps$ with respect to $\grad \mathcal{U}^\eps$ (the signs of the terms are kept in order to facilitate the allocation of the terms, and $\backslash\backslash$ denotes a line break):
\begin{align*}
\int_{\Oe} & F^{-T}\grad \mathcal{U}^\eps : F^{-T} \grad \phi \ud x - \int_{\Oe} \mathcal{P}^\eps \div( F^{-1}\phi) \ud x  = - \mathrm{(P)} - \mathrm{(P)} + \mathscr{O}(\eps) \\
& \quad +\mathscr{O}(\eps^\frac{3}{2}) - \mathscr{O}(\eps^\frac{1}{2}) \quad \backslash\backslash \quad
 +\mathscr{O}(\eps^\frac{3}{2}) + \mathscr{O}(\eps^\frac{1}{2}) + \mathscr{O}(\eps^\frac{3}{2}) \quad \backslash\backslash \quad+ \mathrm{(BL)} %+ \mathrm{(BL)} 
 \quad \backslash\backslash \quad - \mathrm{(BL)}. %- \mathrm{(BL)}
\end{align*}
Now $\grad\mathcal{U}^\eps$ is of order $\eps^\frac{1}{2}$, and $\mathcal{U}^\eps$ is of order $\eps^\frac{3}{2}$ in $\Oe_2$. For the proof of the main result, we need $\grad\mathcal{U}^\eps$ to be on the order of $\eps$, and $\mathcal{U}^\eps$ should be on the order of $\eps^2$. 
We therefore need to correct the velocity again.

\subsubsection{Second correction of the velocity}

Define the following boundary layer function:
\begin{empheq}[box=\smallerbox]{align*}
-\divy(F^{-1}(x)F^{-T}(x) & \grady \gbl(x,y)) + F^{-T}(x)\grady\mubl(x,y) \\ &  =\divx( F^{-1}(x)F^{-T}(x) \grady \bbl(x,y) )  \quad \text{ in $\O\times Z$} \\
\divy(F^{-1}(x)\gbl(x,y))&=0  \quad\text{ in $\O\times Z$} \\
[ \gbl(x,y) ]_S&= 0  \quad\text{ on $\O\times S$} \\
[(F^{-1}(x)(F^{-T}(x)\grady \gbl(x,y)-\mubl(x,y)))e_2]_S& =0 \quad\text{ on $\O\times S$} \\
\gbl(x,y)&=0 \quad\text{ on } \O \times {\textstyle\bigcup_{k=1}^\infty \{  \p Y_S - \binom{0}{k} \} }\\
\gbl(x,\cdot), \mubl(x,\cdot) \text{ are $1$-periodic in $y_1$}
\end{empheq}
The decay functions are denoted by $\Cblg, \Cblmu$. Define the corresponding counterflow
\begin{empheq}[box=\widebox]{align*}
-\div(F^{-1}(x)F^{-T}(x) \grad b(x) ) + F^{-T}(x) \grad q(x) =0 &\quad \text{ in $\O_1$} \\
\div(F^{-1}(x)  b(x)) =0 & \quad\text{ in $\O_1$} \\
 b(x)= 0 &\quad\text { on $(0,L) \times \{ h\}$}\\
 b(x)= \Cblg(x) &\quad\text { on $\Sigma$}\\
b, q\text{ are $L$-periodic in $x_1$}
\end{empheq}

\subsubsection{Correction of the pressure}

For the correction of the pressure, define the following boundary layer function
\begin{empheq}[box=\smallerbox]{align*}
-\divy(F^{-1}(x)F^{-T}(x)  \grady \lbl(x,y)) + F^{-T}(x)&\grady\kbl(x,y) \\   = F^{-T}(x) \gradx( \obl(x,y) - H(x_2) \Cblo(x) ) &\qquad \quad \! \text{ in $\O\times Z$} \\
\divy(F^{-1}(x)\lbl(x,y))&=0  \quad\text{ in $\O\times Z$} \\
[ \lbl(x,y) ]_S&= 0  \quad\text{ on $\O\times S$} \\
[(F^{-1}(x)(F^{-T}(x)\grady \lbl(x,y)-\kbl(x,y)))e_2]_S& =0 \quad\text{ on $\O\times S$} \\
\lbl(x,y)&=0 \quad\text{ on } \O \times {\textstyle\bigcup_{k=1}^\infty \{  \p Y_S - \binom{0}{k} \} }\\
\lbl(x,\cdot), \kbl(x,\cdot) \text{ are $1$-periodic in $y_1$}
\end{empheq}
The decay functions are denoted by $\Cbll, \Cblk$. Define the corresponding counterflow
\begin{empheq}[box=\widebox]{align*}
-\div(F^{-1}(x)F^{-T}(x) \grad d(x) ) + F^{-T}(x) \grad l(x) =0 &\quad \text{ in $\O_1$} \\
\div(F^{-1}(x)  d(x)) =0 & \quad\text{ in $\O_1$} \\
 d(x)= 0 &\quad\text { on $(0,L) \times \{ h\}$}\\
 d(x)= \Cbll(x) &\quad\text { on $\Sigma$}\\
d, l\text{ are $L$-periodic in $x_1$}
\end{empheq}
We extend our approximation to $u^\eps$ and $p^\eps$ by adjusting $\mathcal{U}^\eps$ and $\mathcal{P}^\eps$. Define
\begin{gather*}
\tilde{\mathcal{U}}^\eps = \mathcal{U}^\eps + (\eps^2\gbl(x,\xe) - \eps^2H(x_2)\Cblg) - (\eps^2 \lbl(x,\xe) - \eps^2H(x_2) \Cbll) + \eps^2H(x_2) b - \eps^2H(x_2)d \\
\tilde{\mathcal{P}}^\eps = \mathcal{P}^\eps + (\eps\mubl(x,\xe) - \eps H(x_2)\Cblmu) - (\eps  \kbl(x,\xe) - \eps H(x_2) \Cblk) + \eps^2 H(x_2) q - \eps^2 H(x_2) l 
\end{gather*}
The weak formulation is now given by
\begin{align*}
\int_{\Oe} & F^{-T}\grad \tilde{\mathcal{U}}^\eps:F^{-T}\grad \phi \ud x - \int_{\Oe} \tilde{\mathcal{P}}^\eps \div(F^{-1}\phi) \ud x \\
& =  -\int_\Sigma F^{-1} [\tilde{\pi}^0]_\Sigma e_2 \cdot \phi \ud \sigma  - \int_\Sigma F^{-1} \Cblo e_2 \cdot \phi \ud \sigma+ \int_{\Oe_2} (f-F^{-T}\grad \tilde{\pi}^0) \phi \ud x\\
& \quad + \int_{\Oe} F^{-T} \gradx(\bble - \eps \Cbl H(x_2)) :F^{-T} \grad \phi \ud x + \eps \int_{\Sigma} F^{-1}\pi^\sigma e_2 \cdot \phi \ud \sigma \displaybreak[0] \\
& \quad+ \eps \int_\Sigma F^{-1} F^{-T} \grad u^\sigma e_2 \cdot \phi \ud \sigma \displaybreak[0] \\
& \quad + \int_{\Oe}\eps^2 F^{-T} \gradx( \gbl(x,\xe) - H(x_2) \Cblg  ) : F^{-T} \grad \phi\ud x - \int_{\Oe} \eps \divx( F^{-1}F^{-T} \grady \gbl(x,\xe) ) \phi \ud x \\
& \quad + \int_{\Oe}\eps^2 F^{-T} \gradx( \lbl(x,\xe) - H(x_2) \Cbll  ) : F^{-T} \grad \phi\ud x - \int_{\Oe} \eps \divx( F^{-1}F^{-T} \grady \lbl(x,\xe) ) \phi \ud x \\
& \quad  - \eps^2\int_{\Sigma} F^{-1}F^{-T} \grad( b-d ) e_2 \cdot \phi \ud \sigma + \int_{\Oe} \eps F^{-T} \gradx( \mubl(x,\xe)  - H(x_2) \Cblmu ) \phi \ud x\\
& \quad + \int_{\Sigma} \eps F^{-1} \Cblmu e_2 \cdot \phi \ud \sigma - \int_{\Oe} \eps F^{-T} \gradx( \kbl(x,\xe)  - H(x_2) \Cblk ) \phi \ud x \\
& \quad- \int_{\Sigma} \eps F^{-1} \Cblk e_2 \cdot \phi \ud \sigma + \int_\Sigma \eps^2 F^{-1}(q-l) e_2 \cdot \phi \ud x + \mathrm{(B)}.
\end{align*}
Here $\mathrm{(B)}$ designates terms stemming from the boundary layer functions at the outer boundary of $\O$. Again, due to their exponential decay, they can be chosen arbitrarily small. Similar to the calculations above, we arrive at the following estimates of $\grad \tilde{\mathcal{U}}^\eps$ in terms of $\eps$:
\begin{align*}
\int_{\Oe} &  F^{-T}\grad \tilde{\mathcal{U}}^\eps:F^{-T}\grad \phi \ud x  - \int_{\Oe} \tilde{\mathcal{P}}^\eps \div(F^{-1}\phi) \ud x  =- \mathrm{(P)} - \mathrm{(P)} + \mathscr{O}(\eps) \\
& \quad +\mathscr{O}(\eps^\frac{3}{2}) + \mathscr{O}(\eps^\frac{3}{2})
 \quad \backslash\backslash \quad +\mathscr{O}(\eps^\frac{3}{2}) 
 \quad \backslash\backslash \quad+\mathscr{O}(\eps^\frac{5}{2}) - \mathscr{O}(\eps^\frac{3}{2})   
 \quad \backslash\backslash \quad + \mathscr{O}(\eps^\frac{5}{2})  - \mathscr{O}(\eps^\frac{3}{2}) \\
& \quad - \mathscr{O}(\eps^\frac{5}{2}) + \mathscr{O}(\eps^\frac{3}{2})  
\quad \backslash\backslash \quad + \mathscr{O}(\eps^\frac{3}{2})  - \mathscr{O}(\eps^\frac{3}{2}) 
\quad \backslash\backslash \quad- \mathscr{O}(\eps^\frac{3}{2})  + \mathscr{O}(\eps^\frac{5}{2})  + \mathrm{(BL)}.
\end{align*}
Thus, one arrives at an approximation $\grad \tilde{\mathcal{U}}^\eps \approx \mathscr{O}(\eps)$, $\tilde{\mathcal{U}}^\eps \approx\mathscr{O}(\eps^2) $ in $\Oe_2$ and $\tilde{\mathcal{U}}^\eps |_{\Sigma} \approx \mathscr{O}(\eps^\frac{3}{2})$. In order to prove these estimates rigorously, we need to correct the divergence of $\tilde{\mathcal{U}}^\eps$ next.

\subsection{Correction of the divergence}

\subsubsection{Compressibility due to the boundary layer functions}

A calculation shows that
\begin{gather*}
\div(F^{-1}\tilde{\mathcal{U}}^\eps)(x) = \eps \divx[ F^{-1}( \bbl(x,\xe) - H(x_2)\Cbl(x) )   ] %
 + \eps^2 \divx[ F^{-1}( \gbl(x,\xe) - H(x_2) \Cblg(x) ) ]\\ - \eps^2\divx[ F^{-1}(\lbl(x,\xe) - H(x_2)\Cbll(x))  ]
\end{gather*}
% OLD VERSION
We construct the following convergence of the divergence: Let $Q_\beta$ satisfy
\begin{empheq}[box=\smallerbox]{alignat*=2}
\divy(F^{-1}(x)Q_\beta(x,y))&= \divx\Bigl(F^{-1}(x) \bigl[ \bbl(x,y) -H(y_2)\Cbl(x) \bigr]\Bigr) & \quad&\text{ in $\Omega\times Z$} \\
Q_\beta(x,y)&=0  &&\text{ on } \Omega\times {\textstyle\bigcup_{k=1}^\infty \{   Y_S - \binom{0}{k} \} }\\
[Q_\beta]_S (x,y)& =   \Cqb(x)     &&\text{ on $\Omega\times S$} \\
Q_\beta(x,y) &\text{ is $1$-periodic in $y_1$}
\end{empheq}
where
\[
\Cqb(x)=F(x) \int_{\ZBL} \divx\Bigl( F^{-1}(x)(\bbl(x,y) - H(y_2 )\Cbl(x) ) \Bigr)e_2 \ud y  .
\]
Set $\Qeb(x):= \eps^2 Q_\beta(x,\xe)$. Similarly, we define $Q_\gamma$ and $Q_\lambda$
\begin{empheq}[box=\smallerbox]{alignat*=2}
\divy(F^{-1}(x)Q_\gamma(x,y))&= \divx\Bigl(F^{-1}(x) \bigl[ \gbl(x,y) -H(y_2)\Cblg(x) \bigr]\Bigr) & \quad&\text{ in $\Omega\times Z$} \\
Q_\gamma(x,y)&=0  &&\text{ on } \Omega\times {\textstyle\bigcup_{k=1}^\infty \{  Y_S - \binom{0}{k} \} }\\
[Q_\gamma]_S (x,y)& =   \Cqg(x)     &&\text{ on $\Omega\times S$} \\
Q_\gamma(x,y) &\text{ is $1$-periodic in $y_1$}
\end{empheq}
as well as
\begin{empheq}[box=\smallerbox]{alignat*=2}
\divy(F^{-1}(x)Q_\lambda(x,y))&= \divx\Bigl(F^{-1}(x) \bigl[ \lbl(x,y) -H(y_2)\Cbll(x) \bigr]\Bigr) & \quad&\text{ in $\Omega\times Z$} \\
Q_\lambda(x,y)&=0  &&\text{ on } \Omega\times {\textstyle\bigcup_{k=1}^\infty \{   Y_S - \binom{0}{k} \} }\\
[Q_\lambda]_S (x,y)& =   \Cql(x)     &&\text{ on $\Omega\times S$} \\
Q_\lambda(x,y) &\text{ is $1$-periodic in $y_1$}
\end{empheq}
with decay functions
\begin{align*}
\Cqg(x)&=F(x) \int_{\ZBL} \divx\Bigl( F^{-1}(x)(\gbl(x,y) - H(y_2 )\Cblg(x) ) \Bigr)e_2 \ud y  \\
\Cql(x)&=F(x) \int_{\ZBL} \divx\Bigl( F^{-1}(x)(\lbl(x,y) - H(y_2 )\Cbll(x) ) \Bigr)e_2 \ud y .
\end{align*}
Again, define $\Qeg(x)=\eps^3 Q_\gamma(x,\xe)$ and $\Qel(x)=\eps^3 Q_\lambda(x,\xe)$. The tools and techniques for proving existence and exponential decay for these type of problems can be found in Appendix~\ref{sec:divcorr}.

The jump across the boundary of the functions above is corrected with the help of the following counterflows: For $i\in \{ \beta, \gamma, \lambda \}$ define $(u^{Q_i}, \pi^{Q_i})$ as the solution of
\begin{empheq}[box=\widebox]{align*}
-\div(F^{-1}(x)F^{-T}(x) \grad u^{Q_i}(x) ) + F^{-T}(x) \grad \pi^{Q_i}(x) =0 &\quad \text{ in $\O_1$} \\
\div(F^{-1}(x)  u^{Q_i(x)}) =0 & \quad\text{ in $\O_1$} \\
 u^{Q_i}(x)= 0 &\quad\text { on $(0,L) \times \{ h\}$}\\
 u^{Q_i}(x)= C^Q_i(x) &\quad\text { on $\Sigma$}\\
u^{Q_i}, \pi^{Q_i}\text{ are $L$-periodic in $x_1$}
\end{empheq}
These problems have a solution if $\int_\Sigma C^Q_i \cdot F^{-T}e_2 \ud \sigma =0$. For $C^Q_\beta$ we obtain
\begin{align*}
\int_\Sigma C^Q_\beta \cdot F^{-T}e_2 \ud \sigma &= \int_\Sigma \int_{\ZBL} \divx\Bigl( F^{-1}(x)(\bbl(x,y) - H(y_2 )\Cbl(x) ) \Bigr) \ud y \ud \sigma_x \\
& = \int_\Sigma \divx\Bigl( F^{-1}(x) [\int_{\ZBL} \bbl(x,y) - H(y_2 )\Cbl(x) \ud y] \Bigr) \ud \sigma_x \\
& = \int_{\Sigma} \frac{\p}{ \p x_1} \Bigl( F^{-1}(x) [\int_{\ZBL} \bbl(x,y) - H(y_2 )\Cbl(x) \ud y] \Bigr) \cdot e_1 \\
& \qquad + \frac{\p}{ \p x_2} \Bigl( F^{-1}(x) [\int_{\ZBL} \bbl(x,y) - H(y_2 )\Cbl(x) \ud y] \Bigr) \cdot e_2 \ud \sigma_x.
\end{align*}
Now the first term on the right hand side vanishes due to the periodic boundary funcion $g$, which leads to periodic boundary conditions in $x$ for $F$, $\bbl$ and $\Cbl$. The second term vanishes, since the latter functions do not depend on $x_2$. The argument for the remaining counterflows is analogous.

Now define
\begin{align*}
\mathcal{U}^\eps_* &= \tilde{\mathcal{U}}^\eps - \Qeb - \Qeg + \Qel + \eps^2 u^{Q_\beta} + \eps^3 u^{Q_\gamma} - \eps^3 u^{Q_\lambda} \\
\mathcal{P}^\eps_* &= \tilde{\mathcal{P}}^\eps + \eps^2 \pi^{Q_\beta} + \eps^3 \pi^{Q_\gamma} - \eps^3 \pi^{Q_\lambda} ,
\end{align*}
which leads to 
\[
\div(F^{-1}(x) \mathcal{U}^\eps_*(x)) = - \eps^2 \divx(F^{-1}(x)Q_\beta(x,\xe)) - \eps^3 \divx(F^{-1}(x)Q_\gamma(x,\xe)) + \eps^3 \divx(F^{-1}(x)Q_\lambda(x,\xe))
\]

\subsubsection{Compressibility due to the auxiliary functions}

% Old Version
We correct the divergence even further by defining the functions $\phi^{1,\eps}_i$ and $\phi^{2,\eps}_i$ via
\begin{empheq}[box=\smallerbox]{alignat*=2}
\div(F^{-1}(x)\phi^{1,\eps}_i(x))&= - \eps_i \divx(F^{-1}(x) Q_i(x,\xe)) & \quad&\text{ in $\Omega_1$} \\
\phi^{1,\eps}_i(x)  &=0  &&\text{ on } (0,L) \times \{ h\}\\
\phi^{1,\eps}_i(x)& =   \frac{F(x)}{|\Sigma|} \int_{\O_1} \eps_i \divx( F^{-1}(x) Q_i(x,\xe))e_2 \ud x    &&\text{ on $\Sigma$} \\
\phi^{1,\eps}_i &\text{ is $L$-periodic in $x_1$}
\end{empheq}
and
\begin{empheq}[box=\smallerbox]{alignat*=2}
\div(F^{-1}(x)\phi^{2,\eps}_i(x))&= - \eps_i \divx(F^{-1}(x) Q_i(x,\xe)) & \quad&\text{ in $\Omega_2$} \\
\phi^{2,\eps}_i(x)  &=0  &&\text{ on } (0,L) \times \{ -K\}\\
\phi^{2,\eps}_i(x)& =   -\frac{F(x)}{|\Sigma|} \int_{\O_2} \eps_i \divx( F^{-1}(x) Q_i(x,\xe))e_2  \ud x    &&\text{ on $\Sigma$} \\
\phi^{2,\eps}_i &\text{ is $L$-periodic in $x_1$}
\end{empheq}
Here $Q_i$ is extended by $0$ into $\bigcup_{k=1}^\infty \{  Y_S - \binom{0}{k} \}$, $i\in \{ \beta, \gamma, \lambda \}$, and
\[
\eps_i:=\begin{cases}
\eps^2 & \text{for } i= \beta \\
\eps^3 & \text{for } i= \gamma \\
-\eps^3 & \text{for } i= \lambda 
\end{cases}.
\]
It holds $\nnormalx{\phi^{1,\eps}_i}{H^1(\O_1)} \leq C|\eps_i|$, $\nnormalx{\phi^{2,\eps}_i}{H^1(\O_2)} \leq C|\eps_i|$ and $\phi^{1,\eps}_i = -\phi^{2,\eps}_i$ on $\Sigma$. %% ????? WHY ?

%%  NEW VERSION - WRONG
%We correct the divergence even further by defining the functions $\phi^{1,\eps}$ and $\phi^{2,\eps}$ via
%\begin{empheq}[box=\smallerbox]{alignat*=2}
%\div(F^{-1}(x)\phi^{1,\eps}(x))&= D^\eps(x) & \quad&\text{ in $\Omega_1$} \\
%\phi^{1,\eps}(x)  &=0  &&\text{ on } (0,L) \times \{ h\}\\
%\phi^{1,\eps}(x)& =  - \frac{F(x)}{|\Sigma|} \int_{\O_1} D^\eps(x) e_2 \ud x    &&\text{ on $\Sigma$} \\
%\phi^{1,\eps} &\text{ is $L$-periodic in $x_1$}
%\end{empheq}
%and
%\begin{empheq}[box=\smallerbox]{alignat*=2}
%\div(F^{-1}(x)\phi^{2,\eps}(x))&= D^\eps(x)  & \quad&\text{ in $\Omega_2$} \\
%\phi^{2,\eps}(x)  &=0  &&\text{ on } (0,L) \times \{ -H\}\\
%\phi^{2,\eps}(x)& =   \frac{F(x)}{|\Sigma|} \int_{\O_2} D^\eps(x)  e_2  \ud x    &&\text{ on $\Sigma$} \\
%\phi^{2,\eps} &\text{ is $L$-periodic in $x_1$}
%\end{empheq}
%Here $Q_\beta$ is extended by $0$ into $\bigcup_{k=1}^\infty \{  Y_S - \binom{0}{k} \}$, $i\in \{ \beta, \gamma, \lambda \}$. It holds $\nnormalx{\phi^{1,\eps}}{H^1(\O_1)} \leq C\eps^2$, $\nnormalx{\phi^{2,\eps}}{H^1(\O_2)} \leq C\eps^2$ and $\phi^{1,\eps}_i  -\phi^{2,\eps}$ is of the order of $\eps^2$ on $\Sigma$.

We define the final correction by making use of the following restriction operator (see \cite{do_dipl} for details):
\begin{prop}
\label{prop:restop}
Let $\bar{H}=\{ u\in H^1(\O_2)^2: u=0 \text{ on } (0,L)\times \{ -H\}, u\text{ is $L$ periodic in the first variable} \}$ and set $\bar{H}^\eps=\{  u\in \bar{H}: u=0 \text{ on } \p\Oe_2 \backslash \p \O_2 \}$. 
There exists a linear restriction operator
$\mathcal{R}^\eps: \bar{H} \longrightarrow \bar{H}^\eps$ 
such that for $w\in \bar{H}$:
\begin{gather*}
w=0 \text{ on } \p\Oe\backslash\p\O \quad \Longrightarrow \quad \mathcal{R}^\eps w = w|_{\Oe} \\
\div(F^{-1}w)=0 \text{ in } \O \quad \Longrightarrow \quad \div(F^{-1}\mathcal{R}^\eps w)=0 \text{ in } \Oe
\end{gather*}
and
\[
\nx{\mathcal{R}^\eps w}{L^2(\Oe)^2} + \eps\nx{\grad(\mathcal{R}^\eps w)}{L^2(\Oe)^4} \leq C (  \nx{ w}{L^2(\O)^2} + \eps\nx{\grad w}{L^2(\O)^4}  ).
\]
\end{prop}
Using an explicit characterization of $\mathcal{R}^\eps$ (see the Appendix of \cite{sanpal} by Luc Tartar), one arrives at the following identity 
\begin{equation}
\label{eq:restrop}
\div(F^{-1} \mathcal{R}^\eps \phi) = \div (F^{-1}\phi) + \sum_{\substack{k\in \{ h\in \Z^2;\\ \eps(Y_S + k) \subset \O_2 \}}}  \frac{\chi_{\eps(Y^*+k)}}{\eps|Y^*|}  \int_{\eps(Y_S+k)} \div(F^{-1}\phi) \ud x.
\end{equation}

% OLD VERSION
Now define
\begin{align*}
\mathcal{U}^\eps_0&= \tilde{U}^\eps_* - H(x_2)\sum_{i\in \{\beta, \gamma, \lambda\}} \phi^{1,\eps}_i  - H(-x_2) \sum_{i\in \{\beta, \gamma, \lambda\}} \mathcal{R}^\eps\phi^{2,\eps}_i \\
\mathcal{P}^\eps_0& = \tilde{ \mathcal{P}}^\eps_*.
\end{align*}
Since $\divx(F^{-1}(x) Q_i(x,\xe)) =0$ in $\eps(Y_S + k)$, $k\in \Z^2$, equation \eqref{eq:restrop} leads to $\div(F^{-1}\mathcal{U}^\eps_0)=0$, and we can use $\mathcal{U}^\eps_0$ as a test function. 
%% NEW VERSION - WRONG
%Now define
%\begin{align*}
%\mathcal{U}^\eps_0&= \tilde{U}^\eps_* - H(x_2) \phi^{1,\eps}  - H(-x_2) \mathcal{R}^\eps\phi^{2,\eps}\\
%\mathcal{P}^\eps_0& = \tilde{ \mathcal{P}}^\eps_*.
%\end{align*}
%Then $\div(F^{-1}\mathcal{U}^\eps_0)=0$, and we can use $\mathcal{U}^\eps_0$ as a test function.  
%
%
We look at the weak formulation
\[
\int_{\Oe} F^{-T}\grad {\mathcal{U}}^\eps_0:F^{-T}\grad \phi \ud x - \int_{\Oe} {\mathcal{P}}^\eps_0 \div(F^{-1}\phi) \ud x.
\]
By calculating the right hand side of this equality and inserting $\phi=\mathcal{U}^\eps_0$, we obtain (similar to the proof of Proposition~\ref{prop:basicestim}) $\nnormalx{\grad \mathcal{U}^\eps_0}{L^2(\Oe)^4}\leq C\eps$, $\nnormalx{\mathcal{U}^\eps_0}{L^2(\Oe_2)^2} \leq C\eps^2$ and $\nnormalx{\mathcal{U}^\eps_0}{L^2(\Sigma)^2} \leq C \eps^\frac{3}{2}$. Similarly, one arrives at $\nnormalx{\mathcal{P}^\eps_0}{H^{-1}(\Oe)}\leq C \eps$.

\begin{cor}
Let $\tilde{\pi}^0$ be a smooth function satisfying $[\tilde{\pi}^0]_\Sigma=-\Cblo$. Then the estimates
\begin{align*}
\nx{\mathcal{U}^\eps}{H^\frac{1}{2}(\O_1)^2} + \eps^\frac{1}{2}\nx{\mathcal{P}^\eps}{L^2(\O_1)} \leq C \eps^\frac{3}{2}
\end{align*}
hold.
\end{cor}
\begin{proof}
The result follows as in \cite{mik_effpress} by transforming the equations in $\O_1$ back to a Stokes system in $\tilde{\O}_1$.
\end{proof}

\subsection{Some Estimates for the Corrected Velocity and Pressure}

Note that $\mathcal{U}^\eps$ differs from $\mathcal{U}^\eps_0$ by $\mathscr{O}(\eps^2)$ in the $L^2$-norm, and by $\mathscr{O}(\eps)$ in the $H^1$-norm. This yields the estimates
\begin{empheq}[box=\widefbox]{gather*}
\eps \nx{\grad \mathcal{P}^\eps}{H^{-1}(\Oe)^2} + \eps\nx{\grad\mathcal{U}^\eps}{L^2(\Oe)^4} + \nx{\mathcal{U}^\eps}{L^2(\Oe_2)^2} + \eps^\frac{1}{2}\nx{\mathcal{U}^\eps}{L^2(\Sigma)^2} \leq C \eps^2
\end{empheq}
as well as
\begin{empheq}[box=\widefbox]{gather*}
\eps \nx{\mathcal{P}^\eps}{L^2(\O_1)} + \nx{\mathcal{U}^\eps}{H^\frac{1}{2}(\O_1)^2} \leq C \eps^\frac{3}{2}.
\end{empheq}
We now collect some more estimates. They are analogous to what is known in the literature \cite{mik_effpress}, and the proofs follow along the same lines.
\begin{lem}
\[\nx{\ue-\ueff}{L^2(\O_1)^2} \leq C \eps^\frac{3}{2}\]
\end{lem}
\begin{proof}
Note that $\ue- \ueff = \mathcal{U}^\eps-(\bble- \eps H(x_2) \Cbl)$. By the theory of very weak solutions we have
\[
\nx{\ue-\ueff}{L^2(\O_1)^2} \leq C \nx{\mathcal{U}^\eps-(\bble- \eps H(x_2) \Cbl)}{L^2(\Sigma)^2} \leq C \eps^\frac{3}{2}. \qedhere 
\]
\end{proof}

\begin{lem}
\[ \sqrt{\eps}\nx{\grad(\ue-\ueff  )}{L^2(\O_1)^2} +\nx{\grad(\ue-\ueff  )}{L^1(\O_1)^2} \leq C\eps, \]
\end{lem}
\begin{proof}
First note that
\[
\nx{\grad(\ue-\ueff  )}{L^2(\O_1)^2}  \leq \nx{\grad \mathcal{U}^\eps}{L^2(\O_1)^2} + \nx{\grad(   \bble- \eps H(x_2) \Cbl)}{L^2(\O_1)^2} \leq C\eps + C \eps^\frac{1}{2} \leq C \eps^\frac{1}{2} .
\]
The second assertion follows via
\[
\nx{\grad(\ue-\ueff  )}{L^1(\O_1)^2}  \leq \nx{\grad \mathcal{U}^\eps}{L^2(\O_1)^2} + \nx{\grad(   \bble- \eps H(x_2) \Cbl)}{L^1(\O_1)^2} \leq  C \eps. \qedhere
\]
\end{proof}

\begin{lem}
\[\nx{\ue-\ueff}{H^\frac{1}{2}(\O_1)^2} \leq C \eps\]
\end{lem}
\begin{proof}
By interpolation we obtain
\begin{gather*}
\nx{\ue-\ueff}{H^\frac{1}{2}(\O_1)^2} \leq C \nx{\ue-\ueff}{L^2(\O_1)^2}^\frac{1}{2}\nx{\grad(\ue-\ueff)}{L^2(\O_1)^2}^\frac{1}{2} \leq C \eps^{\frac{3}{2}\cdot \frac{1}{2}}\eps^{\frac{1}{2}\cdot \frac{1}{2}} \leq C \eps. \qedhere
\end{gather*}
\end{proof}

\begin{lem}
\[\nx{|x_2|^\frac{1}{2}\grad(\ue-\ueff)}{L^2(\O_1)^2} + \nx{|x_2|^\frac{1}{2}\grad(\pe-\peff)}{L^2(\O_1)} \leq C \eps \]
\end{lem}
\begin{proof}
We have 
\begin{gather*}
\nx{|x_2|^\frac{1}{2}\grad(\ue-\ueff)}{L^2(\O_1)^2} \leq C \eps+ \nx{|x_2|^\frac{1}{2}\grad\mathcal{U}^\eps_0}{L^2(\O_1)^2} + \nx{|x_2|^\frac{1}{2}\grad(   \bble- \eps H(x_2) \Cbl)}{L^2(\O_1)^2} \\
\leq C \eps +C\nx{\mathcal{U}^\eps_0}{L^2(\Sigma)^2} + \nx{|x_2|^\frac{1}{2}}{L^\infty(\O_1)} \nx{\grad(   \bble- \eps H(x_2) \Cbl)}{L^1(\O_1)} 
\leq C \eps.
\end{gather*}
For the second part, use Proposition 16 in \cite{mik_effpress} to obtain
\[
\nx{|x_2|^\frac{1}{2}\grad(\pe-\peff)}{L^2(\O_1)} \leq C \nx{\ue - \ueff}{L^2(\Sigma)} + C\eps \leq C \eps \qedhere
\]
\end{proof}

\begin{lem}
\[ \nx{\pe-\peff}{L^1(\O_1)} \leq C \eps \]
\end{lem}
\begin{proof}
\begin{gather*}
\nx{\pe-\peff}{L^1(\O_1)} \leq \nx{\mathcal{P}^\eps - ( \obl-H(x_2)\Cblo ) - \eps H(x_2)\pi^\sigma }{L^1(\O_1)} \\
\leq \nx{\mathcal{P}^\eps_0}{L^2(\O_1)} + \nx{\obl-H(x_2)\Cblo}{L^1(\O_1)} + C\eps 
\leq C \eps \qedhere
\end{gather*}
\end{proof}

\begin{lem}
\[  |M^\eps - M^\mathrm{eff}| \leq C \eps^\frac{3}{2} \]
\end{lem}
\begin{proof}
\begin{gather*}
|M^\eps - M^\mathrm{eff}| = | \int_{\O_1} (\ue - \ueff)\cdot F e_1 \ud x |\\
\leq C\nx{\ue - \ueff}{L^2(\O_1)} \leq C \eps^\frac{3}{2} \qedhere
\end{gather*}
\end{proof}

\begin{lem}
\begin{gather*}
\frac{1}{\eps}(\ue - u^\mathrm{eff}) \weak 0 \quad \text{in } L^2(\Sigma) \\
\nx{\ue - u^\mathrm{eff}}{H^{-\frac{1}{2}}(\Sigma)} \leq C \eps^\frac{9}{5}.
\end{gather*}
\end{lem}
\begin{proof}
It holds $\ue - u^\mathrm{eff}= \mathcal{U}^\eps - (\bble-\eps \Cbl)$ on $\Sigma$. Since $\nx{\eps^{-1}\mathcal{U}^\eps}{L^2(\Sigma)} \leq C \eps$ and $\bbl(\cdot, \frac{\cdot}{\eps})\weak \int_0^1 \bbl(\cdot,y_1,+0) \ud y_1= \Cbl(\cdot)$ in $L^2(\Sigma)$ (cf.\ \cite{amo_osc} or the remarks in \cite{al_2s2}), we obtain the first result.

Now let $v\in H^\frac{8}{5}_\#(\Sigma)$. To begin with, we have
\begin{align*}
\langle  \bbl(\cdot, \frac{\cdot}{\eps})  &  - \Cbl, v \rangle_{\text{\tiny$H^{-\frac{8}{5}}_\#(\Sigma),H^\frac{8}{5}_\#(\Sigma)$}} = \int_0^1 (\bbl(x_1,+0, \frac{x_1}{\eps},+0) - \Cbl(x_1,+0)  )v(x_1) \ud x_1 \\
& =  \eps \int_0^\frac{1}{\eps} (\bbl(\eps y_1,+0, y_1,+0)  - \Cbl(\eps y_1, +0)  )v(\eps y_1) \ud y_1\\
& = \eps \sum_{i\in I^\eps} \int_{Y'_i} (\bbl(\eps y_1,+0, y_1,+0)  - \Cbl(\eps y_1, +0)  )v(\eps y_1) \ud y_1,
\end{align*}
where $I^\eps \subset \N_0$ is some finite index set depending on $\eps$, and $Y_i$ is an interval of the form $[k,k+1)$ for some $k\in \N_0$. For each $i\in I^\eps$, choose a $y_i^*\in Y_i$. Since $\bbl(y_i^*,+0, \cdot, +0)- \Cbl(\cdot, +0)$ has mean value $0$ over $\Sigma$, we have that $\int_{Y'_i} (\bbl(\eps y_i^*,+0, y_1,+0)  - \Cbl(\eps y_i^*, +0)  )\cdot v(\eps y_i^*) \ud y_1 =0$. Thus the expression above is equal to
\begin{gather*}
\eps \sum_{i\in I^\eps} \int_{Y'_i} (\bbl(\eps y_1,+0, y_1,+0)  - \Cbl(\eps y_1, +0)  )v(\eps y_1) \\ \qquad \qquad \qquad -  (\bbl(\eps y_i^*,+0, y_1,+0)  - \Cbl(\eps y_i^*, +0)  )v(\eps y_i^*)\ud y_1.
\end{gather*}
By Taylor series expansion, there exists a $\xi_i^*\in Y'_i$ such that the above expression is equal to
\begin{gather*}
\eps \sum_{i\in I^\eps} \int_{Y'_i}  \frac{\p}{\p x_1}\Bigl[(\bbl(\eps \xi_i^*,+0, y_1,+0)  - \Cbl(\eps \xi_i^*, +0)  )v(\eps \xi_i^*)\Bigr] \eps(y_1-y_i^*  )\ud y_1\\
= \eps \sum_{i\in I^\eps} \int_{Y'_i}  \frac{\p}{\p x_1}\Bigl[(\bbl(\eps \xi_i^*,+0, y_1,+0)  - \Cbl(\eps \xi_i^*, +0)  )v(\eps \xi_i^*)\Bigr] \eps y_1 \ud y_1
\end{gather*}
by the vanishing mean value of $\bbl-\Cbl$. Now transformation of the integral and rearrangement yields
\begin{gather*}
\eps \sum_{i\in I^\eps} \int_{\eps Y'_i} \Bigl[    \frac{\p}{\p x_1}[\bbl( \xi_i^*,+0, \frac{x_1}{\eps},+0)  - \Cbl( \xi_i^*, +0)  ] v(\xi_i^*)  \\
\qquad \qquad \qquad +    v'( \xi_i^*)\bbl( \xi_i^*,+0, \frac{x_1}{\eps},+0)  - \Cbl( \xi_i^*, +0)\Bigr] x_1 \ud x_1
\end{gather*}
By the embedding $H^\frac{8}{5}(\R) \hookrightarrow \C^1(\R)$ as well as by the boundedness of the terms containing $\bbl-\Cbl$ and $x_1$, the integral can be estimated by $\eps C \nx{v}{H^\frac{8}{5}(\eps Y_i) }$, which upon summation gives a bound of $\eps C \nx{v}{H^\frac{8}{5}(\Sigma)}$ for the whole term. This implies that $\nx{\bble - \eps \Cbl}{H^{-\frac{8}{5}}(\Sigma)} \leq C \eps^2$. Since $L^2(\Sigma)\hookrightarrow H^{-\frac{8}{5}}(\Sigma)$, we have that $\nx{\mathcal{U}^\eps}{H^{-\frac{8}{5}}(\Sigma)}\leq C \eps ^2$. Thus similar to the first part $\nx{\ue - u^\mathrm{eff}}{H^{-\frac{8}{5}}(\Sigma)} \leq C \eps^2$, and by interpolation we obtain
\[
\nx{\ue - u^\mathrm{eff}}{H^{-\frac{1}{2}}(\Sigma)} \leq C \nx{\ue - u^\mathrm{eff}}{L^2(\Sigma)}^\frac{1}{5} \nx{\ue - u^\mathrm{eff}}{H^{-\frac{3}{2}}(\Sigma)}^\frac{4}{5} \leq C \eps^\frac{1}{5} \eps^{2\cdot \frac{4}{5}} = C \eps^\frac{9}{5}. \qedhere
\]
\end{proof}

\subsection{Proof of Theorem~\ref{thm:pormed}}
\label{sec:proof}

We need the following auxiliary constructions: Fix $x\in \O$ and let $(w^j, \pi^j) \in H^1_\#(Y^*)^2 \times L^2(Y^*)/\R$ be a solution of the parameter dependent cell problem
\begin{subequations}
\label{subeq:cellpr}
\begin{empheq}[box=\widebox]{align}
-\divy(F^{-1}(x)F^{-T}(x)\grady w^j(x,y)) %\qquad \notag &&&\\  
+ F^{-T}(x)\grady \pi^j(x,y)&=e_j && \text{in $Y^*$} \\
\divy(F^{-1}(x)w^j(x,y))&=0 &&\text{in $Y^*$}\\
w^j(x,y) &= 0&& \text{in $Y_S$}\\
w^j(x,y), \pi^j(x,y)  \text{ are $Y$-periodic in $y$.}
\end{empheq}
\end{subequations} 
Define the matrix $A$ by
\begin{empheq}[box=\widefbox]{equation}
\label{eq:A}
[A(x)]_{ji} = \int_{Y^*} w^j_i(x,y) \ud y  \quad i,j=1,2.
\end{empheq}
One can show that $A$ is symmetric and positive definite. 
The Darcy pressure is given by the following problem: Find $\tilde{\pi}^0 \in H^1(\O_2)$ such that
\begin{subequations}
%\label{eq:effo2}
\begin{alignat}{2}
\div(F^{-1} A (f-F^{-T}\grad \tilde{\pi}^0  )) &=0& \quad & \text{in } \O_2 \\
A (f-F^{-T}\grad\tilde{\pi}^0 )\cdot F^{-T}e_2 &=0 && \text{on } (0,L) \times \{ -K\} \\
\tilde{\pi}^0 &= \pi^0 + C^\text{bl}_\omega && \text{on } \Sigma \\
&\tilde{\pi}^0& & \text{ is $L$-periodic in $x_1$}
\end{alignat}
\end{subequations}
Here $\pi^0$ is given in \eqref{eq:pi0}. Using the weak formulation for $\tilde{\mathcal{U}}^\eps, \tilde{\mathcal{P}}^\eps$, one arrives at the estimates
\begin{gather*}
\nnormalx{\tilde{\mathcal{U}}^\eps}{L^2(\Oe_2)^2} \leq C \eps^2,\qquad \nnormalx{\grad \tilde{\mathcal{U}}^\eps}{L^2(\Oe_2)^4} \leq C\eps,\qquad \nnormalx{\tilde{\mathcal{P}}^\eps}{L^2(\Oe_2)} \leq C.
\end{gather*}
By the theory of two scale convergence (see \cite{ng_2s1}, \cite{al_2s2}) there exist limits $\Ulim\in L^2(\O_2;H^1_\#(Y^*))$ as well as $\Plim\in L^2(\O_2;L^2(Y^*))$ with
\begin{gather*}
\frac{\tilde{\mathcal{U}}^\eps}{\eps^2} \two \Ulim, \qquad
\frac{\grad\tilde{\mathcal{U}}^\eps}{\eps}  \two \grady \Ulim, \qquad
{\tilde{\mathcal{P}}^\eps}  \two \Plim.
\end{gather*}
Choose a test function $\phi \in \C^\infty_0(\O_2; H^1_\#(Y_F))^2$ with $\divy(F^{-1}(x)\phi)=0$ and test the weak formulation for $\tilde{\mathcal{U}}^\eps$. In the 2-scale limit, one obtains
\begin{equation}
\label{eq:Plim}
\begin{gathered}
\int_{\O_2} \int_{Y_F} F^{-T}(x) \grady \Ulim(x,y) : F^{-T}(x)\grady \phi(x,y) \ud y \ud x \\
- \int_{\O_2} \int_{Y_F} \Plim(x,y) \divx(F^{-1}(x)\phi(x,y) ) \ud y \ud x = \int_{\O_2} \int_{Y_F} (f(x) - F^{-T}(x) \grad \tilde{\pi}^0(x) ) \phi(x,y) \ud y \ud x.
\end{gathered}
\end{equation}
By choosing a test function $\phi \in \C^\infty_0(\O_2; H^1_\#(Y_F))^2$, testing with $\eps\phi(x,\xe)$ and taking the 2-scale limit, one arrives at $\int_{\O_2} \int_{Y_F} \Plim(x,y) \divy(F^{-1}(x)\phi(x,y)) =0$, thus $\Plim$ does not depend on the variable $y$. Upon a separation of scales similar to the derivation of Darcy's law (see \cite{hor_homog} and \cite{do_dipl} for the case of the transformed equations), one obtains the representation
\begin{gather*}
\Ulim(x,y) = \sum_{j=1}^2 w^j(x,y)\Bigl[ f(x) - F^{-T}(x) \grad\Bigl( \tilde{\pi}^0(x) + \Plim(x) \Bigr) \Bigr]_j \\
\int_{Y_f} \Ulim(x,y) \ud y = A(x)\Bigl[  f(x) - F^{-T}(x) \grad\Bigl( \tilde{\pi}^0(x) + \Plim(x) \Bigr)   \Bigr].
 \end{gather*}
Choosing $\tilde{\mathcal{U}}^\eps|_{\O_1}$ as a test function yields $\nnormalx{\tilde{\mathcal{U}}^\eps}{L^2(\O_1)^4} \leq C \eps^\frac{3}{2}$ and thus $\grad \frac{\tilde{\mathcal{U}}^\eps}{\eps} \longrightarrow 0$ in $L^2(\O_1)^2$. Choose a $\phi \in \C^\infty_0(\O; H^1_\#(Y_F))^2$ such that $\divy(F^{-1}(x) \phi) =0$. Similar to the above calculations, testing with $\phi(x,\xe)$ gives in the limit
\begin{gather*}
\int_{\O_2} \int_{Y_F} F^{-T}(x) \grady \Ulim(x,y) : F^{-T}(x) \grady \phi(x,y) \ud y \ud x  \\ 
- \int_{\O}\int_{Y_F} \Plim(x) \divx(F^{-1}(x) \phi(x,y)) \ud y \ud x = \int_{\O_2} \int_{Y_F} ( f(x) -F^{-T}(x)\grad \tilde{\pi}^0(x)  )\phi(x,y) \ud y \ud x.
\end{gather*}
Using \eqref{eq:Plim}, one arrives at $\int_{\O_1} \int_{Y_F} \Plim(x) \divx(F^{-1}(x) \phi(x,y)) \ud y \ud x =0 $. Upon an integration by parts, this is equivalent to
\[
-\int_{\O_1} \int_{Y_F} F^{-T}(x) \grad \Plim(x) \phi(x,y) \ud y \ud x + \int_{\Sigma} \int_{Y_F} \Plim(x) \phi(x,y) \cdot F^{-T}(x) e_2 =0.
\]
Since $\phi$ is arbitrary, we conclude that $F^{-T} \grad \Plim=0$ in $\O_1$ and $\Plim|_\Sigma =0$. Thus $\Plim =0$ and $\int_{Y_f} \Ulim(x,y) \ud y = A(x) (  f(x) - F^{-T}(x) \grad\tilde{\pi}^0(x) )$, which is Darcy's law. 

It remains to compare $\tilde{\pi}^0$ with $\tilde{p}^\mathrm{eff}$. Note that
\[
\nx{\tilde{p}^\mathrm{eff}- \tilde{\pi}^0}{H^{k-1}(\O_2)} \leq C \nx{\peff-\pi^0}{H^{k-1}(\Sigma)} \leq C\nx{\peff-\pi^0}{H^{k}(\O_1)}.
\]
After transformation to $\tilde{\O}_1$, the functions $\ueff, \peff, u^0$ and $p^0$ satisfy a Stokes equation. By standard regularity results for this equation (see \cite{temam}, \cite{mik_effpress}) we obtain $\nx{\peff-\pi^0}{H^{k}(\O_1)} \leq C \eps$ for all $k\in \N$. This gives the first statement of the theorem. Additionally, since
\begin{align*}
\tilde{\mathcal{P}}^\eps|_{\O_2} &= \pe - \tilde{\pi}^0 + \mathscr{O}(\eps^{\frac{1}{2}})\\
& = \pe - \tilde{p}^\mathrm{eff} +  \mathscr{O}(\eps^{\frac{1}{2}}) \weak 0,
\end{align*}
we have that $\pe \weak \tilde{p}^\mathrm{eff}$ in $L^2(\O_2)$, which is the second assertion.

Finally, we estimate $\nx{\pe-\pi^0}{H^{-\frac{1}{2}(\Sigma)}}$. Let $w\in H^1(\Oe_2)$ such that $w=0$ on $ \p\Oe_2 \backslash( \{ x_1=0 \}  \cup \{x_1=L \} \cup \Sigma )$ and $w$ is $L$-periodic in $x_1$. Then by trace estimates \cite{temam}
\begin{gather*}
|\langle \pe - \pi^0, w|_\Sigma \rangle_\text{\tiny{$H^{-\frac{1}{2}}, H^\frac{1}{2}$}}  | \leq |( (\pe-\tilde{\pi}^0)e_2,\grad w  )_\text{\tiny{$L^2(\Oe_2)$}}| + |(\frac{\p}{\p x_2}(  \pe-\tilde{\pi}^0),w)_\text{\tiny{$L^2(\Oe_2)$}}|
\end{gather*}
For the first term on the right hand side, notice that $\pe-\tilde{\pi}^0 = ( \pe - \tilde{p}^\mathrm{eff}) + (\tilde{p}^\mathrm{eff} - \tilde{\pi}^0) = \mathscr{O}(\eps^\frac{1}{2}) + \mathscr{O}(\eps)$. For the second term, observe that
\begin{gather*}
|(\frac{\p}{\p x_2}(  \pe-\tilde{\pi}^0),w)_\text{\tiny{$L^2(\Oe_2)$}} | \leq C \eps\nx{w}{H^1(\Oe_2)}  + \nx{\grad\mathcal{P}^\eps}{H^{-1}(\Oe_2)}\nx{w}{H^1(\Oe_2)} + | (\grad( \oble- H(x_2)\Cblo ),w)_\text{\tiny{$L^2(\Oe_2)$}}| \\
\leq C \eps \nx{w}{H^1(\Oe_2)} + C\eps^{-\frac{1}{2}} \cdot \eps \nx{\grad w}{L^2(\Oe_2)}.
\end{gather*} 
This shows that $\nx{\pe-\pi^0}{H^{-\frac{1}{2}(\Sigma)}} \leq C \sqrt{\eps} $, which implies the last statement of the theorem. \qedhere

\appendix

\section{Coordinate Transformations of Differential Operators}
\label{sec:coordtrans}

We recall the definition of coordinate transformations and some differential operators and investigate their relations: %
Let $\tilde{\O}\subset\R^n$ with $n\in \N$ be a Lipschitz domain; let $\tilde{c}:\tilde{\O}\longrightarrow \R$ be a scalar function, $\tilde{j}:\tilde{\O}\longrightarrow \R^n$ a vector field and $\tilde{M}:\tilde{\O}\longrightarrow \R^{n\times n}$ a matrix function. They are assumed to be sufficiently smooth.

\begin{defn}
\label{def:operators}
The gradient of a vector field is defined as
\[
(\grad\tilde{j})_{ik}=\frac{\p \tilde{j}_k}{\p x_i}, 
\]
for $i,k=1,\dots,n$ (i.e. $\grad \tilde{j}$ is the \emph{transpose} of the Jacobian matrix of $\tilde{j}$); the divergence of a matrix-valued function is defined column-wise, thus 
\[
(\div(\tilde{M}))_k=\sum_{i=1}^n \frac{\p \tilde{M}_{ik}}{\p x_i},
\]
for $k=1,\dots,n$; and the Laplacian of a vector field is given by
\[
\Delta \tilde{j} = \div(\grad \tilde{j}).
\]
For $n=2$ we define the two operators %related to the rotation of a 3-dimensional vector field:
\[
\Curl(\tilde{c})=\begin{pmatrix}
                 -\frac{\p\tilde{c}}{\p x_2} \\
                 \phantom{-}\frac{\p \tilde{c}}{\p x_1}
               \end{pmatrix}
              = \begin{bmatrix}
                 0  & -1 \\
                 1 & \phantom{-}0
                \end{bmatrix}
                \grad \tilde{c}
\]
and
\[
\curl(\tilde{j})=\frac{\p \tilde{j}_2}{\p x_1} -\frac{\p \tilde{j}_1}{\p x_2}.
\]
\end{defn}

%\begin{rem}
The `curl'-operators above are two-dimensional variants of the well-known curl operator describing the rotation of three-dimensional vector fields. 
%For the `curl'-operators, it holds
We have the relations
\[
\curl\grad \tilde{c}=0 \qquad\text{and}\qquad \div\Curl \tilde{c}=0,
\]
and $\curl$ is the formal adjoint of $\Curl$
(see \cite{verfuerth} and \cite{dautlions} for details concerning these operators).
%\end{rem}

\begin{defn}
Let $\O, \tilde{\O}\subset \R^n$ be Lipschitz domains and let $\psi:\O\longrightarrow \tilde{\O}$. We call $\psi$ a \emph{regular orientation-preserving $\mathcal{C}^k$-coordinate transformation} if
\begin{enumerate}
\item $\psi$ is a $\mathcal{C}^k$-diffeomorphism, and
\item There exist $c,C >0$ such that
  \[
    c\leq \det F(z) \leq C \quad \forall z\in\O,
  \]
  where $F$ denotes the Jacobian matrix of $\psi$.
\end{enumerate}
If $\det F \equiv  1$, we call $\psi$ \emph{volume preserving}.
\end{defn}
We will indicate coordinates in $\O$ by $z= (z_1, \dots, z_n)$ and those in $\tilde{\O}$ by $x= (x_1, \dots, x_n)$. 
Define
\begin{gather*}
c(z):=\tilde{c}(\psi(z)) \\
j(z):=\tilde{j}(\psi(z)) \\
M(z):=\tilde{M}(\psi(z)). 
\end{gather*}

\begin{lem}
\label{lem:transfoperators}
Let $\psi:\O\longrightarrow \tilde{\O}$ be a $\mathcal{C}^1$-coordinate transformation. Denote by $F$ the Jacobian matrix of $\psi$, and let $J(z):=\det(F(z))$. With the notations and definitions above it holds
\begin{enumerate}
\item $F^{-T}\grad_z c = \gradx \tilde{c}$.
\item $\div_z(JF^{-1}{j})= (J\circ\psi^{-1}) \divx(\tilde{j})$.
\item $\div_z(JF^{-1}{M})= (J\circ\psi^{-1}) \divx(\tilde{M})$. 
\end{enumerate}
\end{lem}
\begin{proof}
The first assertion is a simple application of the chain rule, whereas the second one is known as the Piola-transformation (see \cite{zeidler4}, Chapter 61. Note that Zeidler defines vectors and gradients row-wise, leading to slightly different formulas.) For the matrix divergence the second statement holds column-wise.
\end{proof}

Application of this lemma yields:

\begin{lem}
\label{lem:transform}
Let $\psi$ be a volume-preserving $\mathcal{C}^1$-coordinate transformation. The operators from Definition~\ref{def:operators} transform according to
\begin{enumerate}
\item $\Delta_x(\tilde{c})  =   \div_z(F^{-1}F^{-T}\grad_z c)$.
\item $\Delta_x (\tilde{j})  =\div_z(F^{-1}F^{-T}\grad_z j)     $.
\item $  \div_x(\tilde{j}) =  \div_z(F^{-1} j)   $.
\item $ \div_x (\tilde{M}) = \div_z(F^{-1}M)   $.
\item $ \Curl_x(\tilde{c}) = \Curlt_z(c)   $, \\
      with
      \[
        \Curlt_z(c)= \begin{bmatrix}
                    0& -1 \\
                    1 & \phantom{-}0
                   \end{bmatrix}
                   F^{-T} \grad_z c.
      \]
\item $  \curl_x(\tilde{j}) = \curlt_z(j) $, \\
      with
       \[
         \curlt_z(j)=\curl_z(F^T j) . 
       \]
\end{enumerate}
\end{lem}
\begin{proof}
For volume-preserving coordinate transformations it holds $J\equiv 1$, thus in that case  by the preceding lemma we have $\div_z(F^{-1}{j})= \divx(\tilde{j})$ and $\div_z(F^{-1}{M})=  \divx(\tilde{M})$, which gives the third and the fourth statement. The first and the second statement follow by the equalities $\Delta_x(\tilde{c})= \div_x(\gradx\tilde{c})$ and $\Delta_x (\tilde{j})=\div_x (\gradx\tilde{j})$ and application of the above results to the right hand sides.

The fifth statement follows along the same lines, whereas the sixth can be obtained by a direct calculation of the effect of the transformation on the defining equation.
%\begin{align*}
%\curl_x(\tilde{j}) & = \frac{\p \tilde{j}_2}{\p x_1} -\frac{\p \tilde{j}_1}{\p x_2}  =  \frac{\p j_2}{\p z_1} -g'(z_1)\frac{\p j_2}{\p z_2} -  \frac{\p j_1}{\p z_2} \\
%& = \curl_z \begin{pmatrix}
%g'(z_1) j_2 + j_1 \\
%j_2
%\end{pmatrix}  = \curl_z(F^T j)
%\end{align*}
\end{proof}

%\begin{rem}
A simple computation shows that $\begin{bmatrix}
                    0& -1 \\
                    1 & \phantom{-}0
                   \end{bmatrix}
                   F^{-T}= F\begin{bmatrix}
                    0& -1 \\
                    1 & \phantom{-}0
                   \end{bmatrix}$; thus it holds
\[
\Curlt_z(c)=F \Curl_z(c).
\]
%\end{rem}

\begin{lem}[Transformed Differential Identities]
\label{lem:transform2}
Let $\psi$ be a volume-preserving $\mathcal{C}^1$-coordinate transformation as above. Then the following identities hold:
\begin{enumerate}
\item $\div_z(F^{-1}c)= F^{-T}\grad_z c$.
\item $\div_z(F^{-1}F^{-T}\grad_z (\div_z (F^{-1}j)) )= \div_z(F^{-1}\div_z(F^{-1}F^{-T}\grad_z j))$. 
\item $\div_z(F^{-1}\Curlt_z(c))=0$.
\item  $\div_z(F^{-1} (c j) ) = c \div_z(F^{-1}j) + F^{-T}\grad c \cdot j$.
\item $\curlt_z(F^{-T}\grad_z c) =0 $.
\item $\curlt_z (\div_z(F^{-1}F^{-T} \grad_z j ))= \div_z(F^{-1}F^{-T} \grad_z  \curlt_z(j) )$.
\item If $\div_z(F^{-1}j)=0$, then 
         \[F^{-T}\grad_z( \curlt_z(j) ) = \begin{bmatrix}
                    0& -1 \\
                    1 & \phantom{-}0
                   \end{bmatrix} \div_z(F^{-1}F^{-T} \grad_z j). \]
\item $F^{-T}\grad_z( \div_z(F^{-1}F^{-T} \grad c) )= \div_z(F^{-1}F^{-T}\grad_z(F^{-T}\grad_z c))$.
\end{enumerate}
\end{lem}
\begin{proof}
To obtain the first statement transform the well-known equation $\divx(\tilde{c}I)=\gradx \tilde{c}$. The second follows from $\Delta_x(\div_x \tilde{j})= \div_x (\Delta_x \tilde{j})$. Next transform $\div_x(\Curl_x(\tilde{c}))=0$ and $\div_x(\tilde{c}\tilde{j})= \tilde{c}\div_x(\tilde{j}) + \grad_x\tilde{c}\cdot\tilde{j} $. Finally observe that $\curl_x(\grad_x \tilde{c})=0$ as well as $\curl_x( \Delta_x \tilde{j} ) =  \Delta_x (\curl \tilde{j})$. 

If $\div_x(\tilde{j})=0$, a simple calculation together with the fact that in this case $\frac{\p \tilde{j}_1}{\p x_1}= - \frac{\p \tilde{j}_2}{\p x_2}$ shows that $\grad_x (\curl_x \tilde{j}) =\left[ \begin{smallmatrix}
                    0& -1 \\
                    1 & \phantom{-}0
                   \end{smallmatrix} \right] \Delta_x (\tilde{j}) $, which upon transformation yields the result.
For the last statement consider $\grad_x(\Delta_x \tilde{c})= \Delta_x( \grad_x \tilde{c})$.
\end{proof}

\begin{rem}
Let ${\nu}(x)$ be the unit normal vector at $x\in \p {\O}$. Then the corresponding transformed unit normal vector is given by
\[
\tilde{\nu}(x)=\norm{F^{-T}(x) {\nu}(x)}^{-1} \,  F^{-T}(x) {\nu}(x).
\]
If $n=2$, the unit tangential vector $\tilde{\tau}(x)$ has the direction $\left[\begin{smallmatrix}
                    0& -1 \\
                    1 & \phantom{-}0
                   \end{smallmatrix} \right] F^{-T}(x) {\nu}(x) = F(x) \left[\begin{smallmatrix}
                    0& -1 \\
                    1 & \phantom{-}0
                   \end{smallmatrix} \right ] {\nu}(x) = F(x) {\tau}(x)$, thus it holds
\[
\tilde{\tau}(x) = \norm{F(x){\tau}(x)}^{-1} \, F(x){\tau}(x).
\]
$\norm{\cdot}$ indicates the chosen norm in $\R^n$.
\end{rem}

\section{Very weak solutions of the transformed Stokes system}
\label{sec:vws}

In this appendix we develop the theory of very weak solutions for the transformed Stokes equations in $\O_1$, which has been suggested for $F=\Id$ by Conca in \cite{con_fl2}. We then derive estimates of the velocity and the pressure. We are looking for functions $(B,\beta) \in L^2(\O_1)^2\times H^{-1}(\O_1)$ as solution of
\begin{subequations}
\label{eq:vws}
\begin{empheq}[box=\widefbox]{alignat=2}
-\div(F^{-1}F^{-T} B) + F^{-T} \grad \beta &= G_1 + \div(F^{-1}G_2) & \quad&\text{ in } \O_1 \\
\div(F^{-1}B) &= 0 && \text{ in } \O_1 \\ 
B &= \xi & &\text{ on } \Sigma_T \\
B, \beta& \text{ are $L$-periodic in $x_1$.}
\end{empheq}
\end{subequations}
Here $\Sigma_T:= \Sigma \cup (0,L)\times \{ h\}$, and $\xi\in L^2(\Sigma_T)^2$, $G_1 \in L^2(\O_1)^2$ as well as $G_2 \in L^2(\O_1)^4$ are given functions. We require the compatibility condition
\[
-\int_\Sigma F^{-1}\xi \cdot e_2 \ud \sigma + \int_{\{ x_2=h\}} F^{-1} \xi \cdot e_2 \ud \sigma =0.
\]
Define the space $W_3=\{  g\in H^1(\O_1), g \text{ is $L$-periodic in $x_1$}, \int_{\O_1} g \ud x =0  \}$ and let $g\in L^2(\O_1)^2$, $u\in W_3$. The auxiliary problem
\begin{subequations}
\label{eq:vwsaux}
\begin{empheq}[box=\widebox]{alignat=2}
-\div(F^{-1}F^{-T} \phi) + F^{-T} \grad \pi &= g & \quad&\text{ in } \O_1 \\
\div(F^{-1}\phi) &= u && \text{ in } \O_1 \\ 
\phi &= 0 & &\text{ on } \Sigma_T \\
\phi, \pi \text{ are $L$-periodic} &\text{ in $x_1$.}
\end{empheq}
\end{subequations}
can be transformed to a Stokes system in $\tilde{\O}_1$ and solved on that domain. This yields the existence of a unique solution $(\phi, \pi) \in H^2(\O_1)^2 \times H^1(\O_1)/\R$ with
\begin{equation}
\label{eq:vwsestimhelp}
\nx{\phi}{H^2(\O_1)^2} + \nx{\grad \pi}{L^2(\O_1)^2} \leq C( \nx{g}{L^2(\O_1)^2} + \nx{u}{H^1(\O_1)}  ).
\end{equation}
Using integration by parts, one calculates that
\begin{gather*}
\int_{\O_1} B\cdot g \ud x - \langle \beta, u \rangle_{\text{\tiny $H^{-1},H^1$}} =  \int_{\O_1} ( G_1 + \div(F^{-1} G_2)  ) \phi \ud x - \int_{\Sigma_T} F^{-1}( F^{-T}\grad \phi - \pi I  )\nu \cdot \xi \ud \sigma
\end{gather*}
Define for $(g,u)\in L^2(\O_1)^2 \times W_3$
\[
l(g,u) = (G_1+\div(F^{-1}G_2),\phi)_\text{\tiny$\O_1$} - (F^{-1}( F^{-T}\grad \phi - \pi I  )\nu,\xi)_\text{\tiny$\Sigma_T$},
\]
where $\phi, \pi$ is a solution of \eqref{eq:vwsaux}.

\begin{defn}
$(B,\beta) \in L^2(\O_1)^2 \times H^{-1}(\O_1)$ is called a \emph{very weak solution} of the Stokes problem \eqref{eq:vws}, if for all $(g,u)\in L^2(\O_1)^2 \times W_3$ the identity
\[
(B,g)_\text{\tiny $\O_1$} - \langle \beta, u \rangle_{\text{\tiny $H^{-1},H^1$}}  = l(g,u)
\]
holds.
\end{defn}

\begin{lem}
The functional $l: L^2(\O_1)^2 \times W_3 \longrightarrow \R$ is linear and continuous.
\end{lem}
\begin{proof}
The solution operator to \eqref{eq:vwsaux} is linear, as is the gradient and the duality pairing. Thus, one obtains the linearity of $l$. For the continuity, we estimate %(where $\sup$ means $\sup_{\nx{g}{L^2(\O_1)^2} + \nx{u}{H^1(\O_1)} \leq 1  }$)
\begin{align*}
|l(g,u)| & \leq  |(G_1+\div(F^{-1}G_2),\phi)_\text{\tiny$\O_1$}| + |(F^{-1}( F^{-T}\grad \phi - \pi I  )\nu,\xi)_\text{\tiny$\Sigma_T$}| \\
&\leq \nx{G_1}{L^2(\O_1)^2} \nx{\phi}{L^2(\O_1)^2} + \nx{G_2}{L^2(\O_1)^4} \nx{\phi}{H^1(\O_1)^2} + C\nx{\grad \phi}{L^2(\Sigma_T)^4} \nx{\xi}{L^2(\Sigma_T)^2} + C\nx{\pi}{L^2(\Sigma_T)} \nx{\xi}{L^2(\Sigma_T)^2} \\
& \leq  ( \nx{G_1}{L^2(\O_1)^2}  +  \nx{G_2}{L^2(\O_1)^4} )\nx{\phi}{H^2(\O_1)^2} + C\nx{ \phi}{H^\frac{3}{2}(\O_1)^2} \nx{\xi}{L^2(\Sigma_T)^2} + C\nx{\grad\pi}{L^2(\O_1)^2} \nx{\xi}{L^2(\Sigma_T)^2} \\
& \leq C( \nx{G_1}{L^2(\O_1)^2}  +  \nx{G_2}{L^2(\O_1)^4}  + 2\nx{\xi}{L^2(\Sigma_T)^2}) ( \nx{\phi}{H^2(\O_1)^2} +   \nx{\grad\pi}{L^2(\O_1)} ) \\
& \leq C( \nx{G_1}{L^2(\O_1)^2}  +  \nx{G_2}{L^2(\O_1)^4}  + \nx{\xi}{L^2(\Sigma_T)^2})  ( \nx{g}{L^2(\O_1)^2} + \nx{u}{H^1(\O_1)} )
\end{align*}
by \eqref{eq:vwsestimhelp}. This shows that $l$ is bounded and thus continuous.
\end{proof}

\begin{prop}
There exists a unique very weak solution of \eqref{eq:vws}.
\end{prop}
\begin{proof}
The Lemma above shows that $l\in (L^2(\O_2)^2 \times W_3)^*$. Since $L^2(\O_2)^2 \times W_3$ is a Hilbert space, the Riesz representation theorem yields the existence of a unique $\tilde{B}\in L^2(\O_1)$ and a unique $\tilde{\beta}\in W_3$ such that
\[
(\tilde{B},g)_\text{\tiny $\O_1$} + \langle \tilde{\beta}, u \rangle_{\text{\tiny $H^{-1},H^1$}}  = l(g,u) \qquad \forall g,u \in L^2(\O_1)^2 \times W_3.
\]
By choosing $B=\tilde{B}$, $\beta = - \tilde{\beta}$, we see that $(B,\beta)$ is a very weak solution of \eqref{eq:vws}.
\end{proof}

\begin{lem}
Let $(B, \beta)$ be a very weak solution of \eqref{eq:vws}. There exists a constant $C>0$ such that the estimate
\begin{empheq}[box=\widefbox]{equation}
\label{eq:vwsestim}
\nx{B}{L^2(\O_1)^2} \leq C ( \nx{G_1}{L^2(\O_1)^2}  +  \nx{G_2}{L^2(\O_1)^4}  + \nx{\xi}{L^2(\Sigma_T)^2}  )
\end{empheq}
holds.
\end{lem}
\begin{proof}
Choose $g=B$ and $u=0$. Using the estimates from the proof of the previous lemma together with the scaled Young's inequality yields the result.
\end{proof}

\section{Various existence theorems}

\subsection{Transformed Stokes equation}
\label{sec:trst}

In this section we prove the existence and uniqueness of the solution of Problem  \eqref{eq:trsteps} for fixed $\eps$. Basically, the approach is the same as in the functional-analytic treatment of the Stokes equation (see e.g.\ \cite{sanpal}).

By multiplying \eqref{seq:trstepsmain} with $\phi\in H^1_\text{div}(\Oe)^2$ where
\[
H^1_\text{div}(\Oe)^2:=\{w\in H^1_{0,\#}(\Oe)^2\ |\ \div(F^{-1}w)=0\},
\] 
integrating by parts and noting that 
\[
\int_{\Oe} F^{-T}(x)\grad p(x) \cdot \phi(x) \ud x =-\int_{\Oe} p(x) \div(F^{-1}(x)\phi(x)) \ud x=0,
\]
we obtain the weak formulation of Problem \eqref{eq:trsteps} in the form
\begin{empheq}[box=\greybox]{equation}
\label{eq:wtrst1}
\int_{\Oe}^{} F^{-T}(x)\grad \ue(x) : F^{-T}(x)\grad \phi(x) \ud x= \int_{\Oe} f(x)\cdot \phi(x) \ud x\quad \forall \phi\in H^1_\text{div}(\Oe)^2
\end{empheq}
Note that $H^1_\text{div}(\Oe)$ is a Banach space with respect to the norm $\nx{\grad \cdot}{L^2(\Oe)^2}$.
We need the following lemma for the estimation of the left hand side:

\begin{lem}
\label{lem:constantskfKf}
There exist constants $0<k_F<K_F$ such that for the eigenvalues $\lambda(x)$ of $F^{-1}(x)F^{-T}(x)$ holds 
\[
k_F<\lambda(x)<K_F \quad \forall x\in\O,
\]
i.e. $F^{-1}(x)F^{-T}(x)$ is symmetric and positive definite.
\end{lem}
\begin{proof}
A calculation of the eigenvalues $\lambda_1(x), \lambda_2(x)$ of $F^{-1}(x)F^{-T}(x)$ yields
\begin{align*}
\lambda_1(x)&=1+\frac{g'(x_1)^2}{2}+ \sqrt{g'(x_1)^2+\frac{g'(x_1)^4}{4}}, \\
\lambda_2(x)&=1+\frac{g'(x_1)^2}{2}- \sqrt{g'(x_1)^2+\frac{g'(x_1)^4}{4}}.
\end{align*}
Because of the smoothness of $g$ there exists an $M>1$ such that $|g'(x)|<M$ for all $x\in\O$. Obviously $\lambda_i(x)\leq 1+2M^2=:K_F$, $i=1,2$.

Choose a $k_F$ small enough such that $M^2+2\leq \frac{1}{k_F}$. Another calculation shows that
\[
\lambda_2(x)\geq k_F \quad\Longleftrightarrow \quad g'(x_1)^2 \leq \frac{1}{k_F}+k_F-2,
\]
which gives the desired result since $\lambda_1(x)\geq\lambda_2(x)$.
\end{proof}

\begin{prop}
\label{prop:bilin}
Let $\eps>0$ be fixed and let $F$ be given by \eqref{eq:F}. For given $f\in L^2(\O)$, the Problem \eqref{eq:wtrst1} has a unique solution $\ue\in H^1_\mathrm{div}(\Oe)^2$.
\end{prop}
\begin{proof}
Define for $u,v\in H^1_\mathrm{div}(\Oe)^2$ the (bi-)linear forms
\[
a(u,v)= \int_{\Oe}  F^{-T}(x)\grad u(x) : F^{-T}(x)\grad v(x) \ud x
\]
and
\[
b(v)= \int_{\Oe} f(x)\cdot v(x) \ud x.
\]
The continuity of $b$ for $f\in L^2(\O)^2$ is standard. In order to apply the lemma of Lax-Milgram, we have to show that $a$ is continuous and coercive. 
First note that as a pointwise estimate we have
\begin{align*}
F^{-T}(x)\grad v(x)&:F^{-T}(x)\grad v(x)  =   \tr(\grad v(x)^TF^{-1}(x) F^{-T}(x)\grad v(x))   \\
&= \sum_{i=1}^2 e_i^T \grad v(x)^T F^{-1}(x) F^{-T}(x)\grad v(x) e_i  \\
 & \leq \sum_{i=1}^2 \nx{\grad v(x) e_i }{2} \nx{F^{-1}(x) F^{-T}(x)}{2} \nx{\grad v(x) e_i}{2} \\
& \leq K_F \sum_{i=1}^2 e_i^T \grad v(x)^T \grad v(x) e_i \\
& = K_F \grad v(x):\grad v(x) = K_F \nx{\grad v(x)}{2}^2,
\end{align*}
with $\nx{\cdot}{2}$ being the Euclidean vector- and matrixnorm.
This gives the continuity of $a$ due to
\begin{align*}
\int_{\Oe} | & F^{-T}(x)\grad u(x) : F^{-T}(x)\grad v(x) | \ud x  \\
 &= \int_{\Oe} | \sum_{i,j=1}^2 \bigl(F^{-T}(x)\grad u(x) \bigr)_{ij}  \bigl(F^{-T}(x)\grad v(x) \bigr)_{ij} | \ud x \\
& \leq  \int_{\Oe} \Bigl( \sum_{i,j=1}^2 \bigl(F^{-T}(x)\grad u(x))_{ij} \bigr)^2 \Bigr)^\frac{1}{2} \Bigl(\sum_{i,j=1}^2 \bigl(F^{-T}(x)\grad v(x))_{ij} \bigr)^2\Bigr)^\frac{1}{2} \ud x \\
&=  \int_{\Oe} \bigl(F^{-T}(x)\grad u(x):F^{-T}(x)\grad u(x)\bigr)^\frac{1}{2}  \bigl(F^{-T}(x)\grad v(x) : F^{-T}(x)\grad v(x) \bigr)^\frac{1}{2}  \ud x\\
& \leq K_F \int_{\Oe} \nx{\grad u(x)}{2} \nx{\grad v(x)}{2} \ud x\\
& \leq K_F \nx{\grad u}{L^2(\Oe)} \nx{\grad v}{L^2(\Oe)},
\end{align*}
where the Cauchy-Schwarz inequality in $L^2$ has been used in the last step. 
For the coercivity consider
\begin{align*}
k_F\nx{\grad v}{L^2(\Oe)} & \leq  \int_{\Oe} \lambda_2(x) \grad v(x):\grad v(x) \ud x\\
&\leq  \int_{\Oe}  \sum_{i=1}^2 \lambda_i(x) e_i^T \grad v(x)^T \grad v(x) e_i \ud x\\
& \leq \int_{\Oe}  F^{-T}(x)\grad v(x) :F^{-T}(x)\grad v(x) \ud x.
\end{align*}
Now the Lax-Milgram lemma implies the proposed result.
\end{proof}
Due to \eqref{eq:wtrst1}, the solution $\ue$ fullfills
\[
-\div(F^{-1}F^{-T}\grad \ue) - f \quad \in (H^1_\text{div}(\Oe)^2)^{\perp}.
\]
We will now characterize the orthogonal complement $(H^1_\text{div}(\Oe)^2)^{\perp}$ of $H^1_\text{div}(\Oe)^2$ in order to reintroduce the pressure. We remind the reader of the following results, the proofs of which can be found in \cite{wloka}:

\begin{thm}[Generalized Trace Theorem]
Let $k\in \R, k\geq 2$, and let $\Lambda $ be a bounded domain in $\R^n$, $n\in \N$ with boundary $\p\Lambda\in \mathcal{C}^{k+1}$. There exists a continuous linear operator $\mathcal{T}:H^k(\Lambda) \longrightarrow H^{k-\frac{1}{2}}(\p\Lambda) \times H^{k-1-\frac{1}{2}}(\p\Lambda)$ with
\[
\mathcal{T}(\phi) = (\phi|_{\p\Lambda}, \frac{\p \phi}{\p \nu}|_{\p\Lambda}  ) \qquad \text{for all $\phi\in \mathcal{C}^k(\bar{\Lambda})$.}
\]
\end{thm}

\begin{thm}[Generalized Inverse Trace Theorem]
\label{thm:invtrace}
Let $\Lambda $ be a bounded domain in $\R^n$, $n\in \N$, with boundary $\p\Lambda\in \mathcal{C}^{k+1}$ with a given $k\in \R, k\geq 2$. Let $\mathcal{T}$ be defined as above.

There exists a continuous linear extension operator $\mathcal{E}:H^{k-\frac{1}{2}}(\p\Lambda) \times H^{k-1-\frac{1}{2}}(\p\Lambda) \longrightarrow  H^k(\Lambda)$ such that
\[
\mathcal{T}\circ\mathcal{E}= \mathrm{Id}.
\]
\end{thm}

\begin{lem}
\label{lem:press}
It holds
\[
(H^1_\mathrm{div}(\Oe)^2)^{\perp} = \{ F^{-T}\grad p \ | \ p\in L^2(\Oe) \}.
\]
\end{lem}
\begin{proof}
Define $G:=\{ F^{-T}\grad p \ | \ p\in L^2(\Oe) \}$. Let $\phi\in G, u\in H^1_\text{div}(\Oe)^2$ with $\phi=F^{-T}\grad p$. Then
\[
 \langle\phi, u\rangle_{\text{\tiny $H^{-1}(\Oe)^2,H^1(\Oe)^2$}} = \langle F^{-T}\grad p , u\rangle_{\text{\tiny $H^{-1}(\Oe)^2,H^1(\Oe)^2$}}  = - \int_{\Oe} p \div(F^{-1}u) \ud x=0,
\]
such that $\phi\in (H^1_\text{div}(\Oe)^2)^{\perp}$. Therefore $G\subset (H^1_\text{div}(\Oe)^2)^{\perp}$.

For the other inclusion we will show that $\div(F^{-1}\cdot):H^1_0(\Oe)\longrightarrow L^2_0(\Oe)$ is surjective and that $-F^{-T}\grad\cdot$ is its adjoint operator, therefore being injective from $L^2_0(\Oe)$ to $\range(-F^{-T}\grad \cdot)$. Now if $\psi \in  (H^1_\text{div}(\Oe)^2)^{\perp}$, we consider $u \in H^1_0(\Oe)$ with $\div(F^{-1}u)=0$. It holds
\[
\langle\psi,u\rangle_{\text{\tiny $H^{-1}(\Oe)^2,H^1_0(\Oe)^2$}} = 0.
\]
Since $u$ is arbitrary, 
\begin{equation*}
\psi \perp \ker(\div(F^{-1}\cdot)),
\end{equation*}
and since $\ker(\div(F^{-1}\cdot))^{\perp}=\range (-F^{-T}\grad\cdot)$ there exists a $p\in L^2(\Oe)$ with
\[
\psi=F^{-T} \grad p.
\]

The surjectivity of $\div(F^{-1}\cdot)$ is a consequence of Lemma \ref{lem:divsurj} below; and the adjointness of the operators can easily be seen from the equation
\[
\langle F^{-T}\grad p , u\rangle_{\text{\tiny $H^{-1}(\Oe)^2,H^1(\Oe)^2$}}  = - \int_{\Oe} p \div(F^{-1}u) \ud x= ({p},{\div(F^{-1}u)})_\text{\tiny $L^2(\Oe)$}. \qedhere
\]
\end{proof}
Before proving some properties of the divergence operator, we need the following lemma.
\begin{lem}
Let $\theta \in H^1(\Oe)$. Then
\begin{align*}
\Curl(\theta)\cdot \nu &= -\grad \theta \cdot \tau \quad \text{on } \p \Oe\\
\Curl(\theta)\cdot \tau &= \phantom{-}\grad \theta \cdot \nu \quad \text{on } \p\Oe.
\end{align*}
\end{lem}
\begin{proof}
It holds
\begin{align*}
\Curl(\theta)\cdot \nu &= \begin{bmatrix}
                 0  & -1 \\
                 1 & \phantom{-}0
                \end{bmatrix}
                \grad \theta \cdot \nu 
= \grad \theta\cdot \Bigl(\begin{bmatrix}
                 0  & -1 \\
                 1 & \phantom{-}0
                \end{bmatrix}^T
                 \nu \Bigr)
= -\grad \theta\cdot \tau,
\end{align*}
since the matrix $\left[\begin{smallmatrix}
                 0  & -1 \\
                 1 & \phantom{-}0
                \end{smallmatrix}\right]^T$ corresponds to a rotation of $\frac{\pi}{2}$ 
and thus $\left[\begin{smallmatrix}
                 0  & -1 \\
                 1 & \phantom{-}0
                \end{smallmatrix}\right]^T \cdot \nu = -\tau$. The second equality follows along the same lines.
\end{proof}

%--Ende Hilfsprobleme
Now we are ready to prove the lemma used above:

\begin{lem}
\label{lem:divsurj}
Let $G\in L^2(\Oe)$ with $\int_{\Oe} G = 0$. There exists a $\phi\in H^1_0(\Oe)^2$ with
\begin{align*}
\div(F^{-1}(x)\phi(x))=G(x) & \text{ in $\Oe$} \\
\phi(x)=0 &\text{ on $\p\Oe$}
\end{align*}
such that
\[
\nx{\phi}{H^1(\Oe)^2} \leq C\nx{G}{L^2(\Oe)}.
\]
Thus $\div(F^{-1}\cdot):H^1_0(\Oe)^2\longrightarrow L^2_0(\Oe)$ is surjective.
\end{lem}
\begin{proof}
We look for $\phi$ in the form
\[
\phi=F\grad\eta +  F\Curl(\theta)
\]
with $\eta$ satisfying
\begin{align*}
\Delta \eta = G & \text{ in $\Oe$} \\
\grad \eta \cdot\nu=0 & \text{ on $\p\Oe$}.
\end{align*}
By considering the weak formulation of this problem
\[
-\int_{\Oe} \grad \eta: \grad \psi = \int_{\Oe} G\cdot \psi \quad\forall \psi\in H^1_0(\Oe)/\R
\]
and using estimates similar to those derived in Propositon \ref{prop:bilin} we see that a unique solution $\eta\in H^1_0(\Oe)/\R$ exists, satisfying the estimate
$\nx{\grad \eta}{L^2(\Oe)^2} \leq C\nx{G}{L^2(\Oe)}$.

By regularity arguments one can show that 
\[
\nx{\frac{\p^2\eta}{\p x_i \p x_j}}{L^2(\Oe)} \leq C\nx{F}{L^2(\Oe)}, \quad i,j\in\{1,2\}.
\]
As for $\theta$, it should hold
\begin{align*}
\Curl(\theta)\cdot \nu = -\grad \theta\cdot \tau= - \grad \eta \cdot \nu =0 & \text{ on $\p\Oe$}\\
\Curl(\theta)\cdot \tau = \grad \theta\cdot \nu= - \grad \eta \cdot \tau \in H^\frac{1}{2}(\Oe) & \text{ on $\p\Oe$}.
\end{align*}
%EXISTENCE???? (evtl. Temam, S. 37 und Theorem 1.2, S.9 konsultieren)
By the general inverse trace theorem \ref{thm:invtrace}, there exists a $\theta\in H^2(\Oe)$ with $\grad\theta\cdot \nu |_{\p\Oe} = - \grad \eta \cdot \tau $ and $\theta|_{\p\Oe}=0$ (thus especially $\grad \theta \cdot \tau =0$ on $\p\Oe$) and
%By the definition of $\Curlt$, this is a condition on the gradient of $\theta$ on $\p\Oe$. By the regularity of $\eta$ holds $\grad \theta |_{\p\Oe} \in H^\frac{1}{2}(\p\Oe)$, thus $\theta\in H\frac{3}{2}(\p\Oe)$; and by using the inverse trace theorem for Sobolev spaces one concludes that there exists an extension of $\theta$ to the whole of $\Oe$ (still denoted by the same symbol) such that
\[
\nx{\theta}{H^2(\Oe)}  %\leq C \nx{\theta}{H^\frac{3}{2}(\p\Oe)} 
\leq C \nx{\grad \eta}{H^1(\Oe)}.
\]
Now we have $\grad \eta + \Curl(\theta)=0$ on $\p\Oe$, therefore also $F(\grad \eta + \Curl(\theta))=0$ on the boundary of $\Oe.$
%Finally, using again estimates similar to \ref{prop:bilin}, one arrives at the desired inequality.
\end{proof}

To reintroduce the pressure, notice that by equation \eqref{eq:wtrst1}
\[
-\div(F^{-1}F^{-T}\grad\ue) - f \quad \in (H^1_\text{div}(\Oe)^2)^{\perp}.
\]
By Lemma \ref{lem:press} there exists a pressure $\pe\in L^2(\Oe)$, unique up to a constant, such that
\[
-\div(F^{-1}(x)F^{-T}(x) \grad{\ue}(x)) -f(x) = - F^{-T}(x) \grad{\pe}(x)
\] 
holds in $\Oe$. This finishes the considerations about the existence and uniqueness of the transformed Stokes equation.

We have the following regularity result:
\begin{prop}
If $f\in H^r(\O)^2$, $r\geq 0$, then $\ue\in H^{r+2}(\Oe)^2$ and $ \pe \in H^{r+1}(\Oe)$.
\end{prop}

We do not give a proof, which can be carried out by adapting the regularity arguments for the usual Stokes equation (see e.g. \cite{temam}). For the interior of the domain, one can use the following argument:

Applying $\div(F^{-1}\cdot)$ to Equation \eqref{seq:trstepsmain} gives (by the second formula of Lemma \ref{lem:transform2})
\[
\div(F^{-1}F^{-T}\grad \pe) = \div(F^{-1} f) \quad \in H^{r-1}({\Oe}') .
\]
Therefore $ \pe\in H^{r+1}({\Oe}')$, where ${\Oe}'$ is a strictly included subdomain of $\Oe$. Because of
\[
-\div(F^{-1}F^{-T}\grad\ue) = f - F^{-T}\grad\pe \quad \in H^r({\Oe}')^2,
\]
we conclude that $\ue\in H^{r+2}({\Oe}')^2$.

%\begin{rem}
%A careful investigation of the foregoing section shows that a solution of \eqref{eq:trsteps} exists even for $f\in H^{-1}(\Oe)^2$.
%\end{rem}

\subsection{Boundary Layer Functions}
\label{sec:boundary}

We define some function spaces that are used in the sequel: 
Let
\begin{align*}
V &= \Bigl\{  z\in L^2_\text{loc}(\ZBL)^2 \ |\ \grad z\in L^2(\ZBL)^4,\  z\in L^2(Z^-)^2, \\ & \qquad z=0  \text{ on } \bigcup_{k=1}^\infty \{  \p Y_S - \binom{0}{k} \}  , \text{ $z$ is $1$-periodic in $x_1$}  \Bigr\}
\end{align*}
and
\begin{align*}
V_{\div}&= \Bigl\{  z\in L^2_\text{loc}(\ZBL)^2 \ |\ \grad z\in L^2(\ZBL)^4,\  z\in L^2(Z^-)^2,\   z=0  \text{ on } \bigcup_{k=1}^\infty \{ \p Y_S - \binom{0}{k} \}  , \\ & \qquad \divy(F^{-1}(x)z(y))=0, \text{ $z$ is $1$-periodic in $x_1$}  \Bigr\}.
\end{align*}
Define $W$ as the completion of $V_{\div}$ with respect to the norm 
$\nx{z}{W}= \nx{\grad z}{L^2(\ZBL)^4}$. 
The Poincar\'e inequality in $Z^-$ reads 
$\nx{z}{L^2(Z^-)^2} \leq C \nx{\grad z}{L^2(Z^-)^4}$ for all $z\in V$.

\subsubsection{The Main Auxiliary Problem}
\label{sec:mauxprob}

For the development of a theory for the boundary layer functions, we start with a more general formulation:

Let $\gamma_1 >0$, $\sigma\in H^\frac{1}{2}(S)^2$, $\rho\in L^2(Z)^2$ and $\rho_1\in L^2(Z)^4$ be given. Assume that $e^{\gamma_1|y_2|}\rho \in L^2(\ZBL)^2 $ and $e^{\gamma_1|y_2|} \rho_1 \in L^2(\ZBL)^4$. Fix $x\in \O$ and consider the following parameter-dependent problem: Find $\zeta\in W$ such that
\begin{empheq}[box=\greybox]{equation}
\label{eq:maux}
\begin{aligned}
\int_{\ZBL} & F^{-T}(x)\grady \zeta(y):  F^{-T}(x) \grady \phi(y) \ud y= \int_{\ZBL} \rho(y) \cdot\phi(y) \ud y 
\\  &- \int_{\ZBL} \rho_1(y) : F^{-T}(x) \grady \phi(y) \ud y + \int_S F^{-1}(x) \sigma(y)\cdot\phi(y) \ud \sigma_y \quad \forall\  \phi\in W 
\end{aligned}
\end{empheq}

\begin{prop}
\label{prop:mauxexistence}
There exists a unique solution of Problem \eqref{eq:maux}.
\end{prop}
\begin{proof}
The result follows by application of the Lax-Milgram lemma:

Define for $u,\phi \in W$
\[
B(u,\phi)= \int_{\ZBL} F^{-T}(x) \grady \zeta(y):  F^{-T}(x) \grady \phi(y) \ud y,
\]
\[
b(\phi) = \int_{\ZBL} \rho(y)\cdot\phi(y)  \ud y - \int_{\ZBL} \rho_1(y) : F^{-T}(x) \grady \phi(y) \ud y   + \int_S F^{-1}(x) \sigma(y)\cdot \phi(y) \ud \sigma_y.
\]
The continuity and coercivity of the bilinear form $B$ in $W$ can be proved analogously to the case of the transformed Stokes equation, see Proposition \ref{prop:bilin}. 
To see that $b$ is bounded, note that
\begin{align*}
\biggl|\int_{\ZBL} \rho(y)\cdot\phi(y) \ud y \biggr|   & \leq \biggl|\int_{Z^+} \rho(y) \cdot\phi(y) \ud y \biggr| +\biggl|\int_{Z^-} \rho(y)\cdot\phi(y) \ud y\biggr| \\ & \leq\nx{(1+y_2)\rho}{L^2(Z^+)^2} \nx{(1+y_2)^{-1}\phi}{L^2(Z^+)^2} + \nx{\rho}{L^2(Z^-)^2} \nx{\phi}{L^2(Z^-)^2} \\ 
&\leq C \nx{\grad \rho}{L^2(\ZBL)^4},
\end{align*}
where we used the standard Poincar\'e inequality in $Z^-$, the fact that $|(1+y_2) \rho|\leq e^{|y_2|} |\rho| \leq \frac{1}{e^{\gamma_1}} e^{\gamma_1|y_2|} |\rho| \leq C e^{\gamma_1 |y_2|} |\rho| $ and $\nx{(1+y_2)^{-1}\phi}{L^2(Z^+)^2} \leq \nx{\grad \phi}{L^2(Z^+)}$, see \cite{jami_bc-fluidpor}.  The estimation of the remaining terms is standard.
\end{proof}

\begin{lem}
\label{lem:mauxreg}
Let $\divy(F^{-1}(x)\rho_1(y))\in L^2(\ZBL)^2$ and let $\rho, \rho_1, \sigma$ be $1$-periodic in $y_1$. Then the solution $\zeta$ of \eqref{eq:maux} is in $H^2_\mathrm{loc}(Z)^2$.
\end{lem}

\begin{prop}
\label{prop:pressloc}
Under the assumptions of Lemma \ref{lem:mauxreg}, there exists a pressure field $\kappa \in L^2_\mathrm{loc}(\ZBL)$ such that
\begin{align*}
-\divy(F^{-1}(x)F^{-T}(x)&\grady \zeta(y) ) +F^{-T}(x)\grady \kappa(y)  \\ &= \rho(y) \mathop{+} \divy(F^{-1}(x) \rho_1(y)) \quad \text{ in $W'$}.
\end{align*}
\end{prop}
\begin{proof}
We are going to use analogues of Lemmas \ref{lem:press} and \ref{lem:divsurj} for an increasing  sequence of sets in order to show that $W^\perp =  \{ F^{-T}(x)\grady p \ | \ p\in L^2_\mathrm{loc}(\ZBL) \}$.

Define for $l\in \N$ the sets $Z_l^*=[0,1] \times (  (0,l) \cup (  \bigcup_{k=1}^l \{   Y^* - \binom{0}{k} \}  )$ and the space
\begin{align*}
W_l &=\Bigl\{  z\in H^1(Z_l^*)^2 \ | \ z=0 \text{ for } y_2=\pm l \text{ and on } \bigcup_{k=1}^l \{  \p Y_S - \binom{0}{k} \} , \\ & \qquad z \text{ is $1$-periodic in $y_1$} \Bigr\}.
\end{align*}
It is clear that $Z_l^*\subset Z_{l+1}^*$ and that each $Z_l^*$ is a Lipschitz domain. 

$\div_l(F^{-1}(x)\cdot):W_l \longrightarrow L^2_0(Z_l^*)$, $\div_l(F^{-1}(x)\cdot):=\divy(F^{-1}(x) \cdot)$ is surjective by an analogue of Lemma \ref{lem:divsurj}, thus $F^{-T}(x)\grad_l( \cdot) :=F^{-T}(x)\grady(\cdot)$ is injective from $L^2_0(Z_l^*)$ to $W_l '$.

Now let $f\in V'$ such that $\langle f, \phi\rangle_{\text{\tiny $H^{-1}(\ZBL)^2,H^1(\ZBL)^2$}}=0$ for all $\phi\in W$. Let $u\in \ker( \div_l(F^{-1}(x)\cdot) )$ be given and denote by $\tilde{u}$ the extension by $0$ outside $Z_l^*$. Since then $\divy(F^{-1}(x)\tilde{u})=0$ in $\ZBL$ we have $\langle f, \tilde{u}\rangle_{\text{\tiny $H^{-1}(\ZBL)^2,H^1(\ZBL)^2$}}=0$. By duality of the extension operation we conclude that $f|_{Z_l^*}\perp \ker(\div_l(F^{-1}(x)\cdot))$. Therefore $f|_{Z_l^*}\in\range( F^{-T}(x)\grad_l \cdot )$, and there exists a $p_l\in L^2(Z_l^*)$, unique up to a constant with $f=F^{-T}(x)\grady p_l$ in $Z_l^*$.

Since $Z_l^*\subset Z_{l+1}^*$, the difference $p_{l+1}-p_l$ is constant in $Z_l^*$ and we can choose $p_{l+1}$ in such a way that $p_{l+1}=p_l$ in $Z_l^*$. Thus $f=F^{-T}(x)\grady p$ with $p\in L^2_\mathrm{loc}(\ZBL)$.

The pressure $\kappa$ can now be obtained by observing that -- via an integration by parts of \eqref{eq:maux} --,  $\divy(F^{-1}(x)F^{-T}(x)\grady \zeta(y) ) + \rho(y) \mathop{+} \divy(F^{-1}(x) \rho_1(y)) \in W^{\perp}$.
\end{proof}

\begin{lem}
Let $\zeta$ and $\kappa$ be defined as above. Under the assumptions of Lemma \ref{lem:mauxreg} we have $\zeta \in H^2_\mathrm{loc}(Z)^2$ and $\kappa\in H^1_\mathrm{loc}(Z)$.
\end{lem}
%\begin{proof}
%????????????
%\end{proof}
Finally, we obtain the following strong form of Problem \eqref{eq:maux}:
\begin{subequations}
\label{eq:stmaux}
\begin{empheq}[box=\widefbox]{align}
-\divy(F^{-1}(x)F^{-T}(x) \grady  \zeta(y)) + F^{-T}(x)\grady \kappa(y) &\notag \\ = \rho + \divy(F^{-1}(x)\rho_1(y)) & \text{ a.e. in $Z$} \label{subeq:stmaux1}\\
\divy(F^{-1}(x) \zeta(y))=0 & \text{ a.e. in $Z$} \\
\zeta(y_1,\pm 0) = \zeta^\pm_0 & \text{ on $S$} \\
\zeta = 0 & \text{ on } {\textstyle \bigcup_{k=1}^\infty \{  \p Y_S - {\binom{0}{k}} \} }\\
\zeta,\kappa  \text{ are $1$-periodic in $y_1$}
\end{empheq}
\end{subequations}
with known functions $\zeta^\pm_0 \in H^\frac{3}{2}(S)^2$.

\subsubsection{Exponential Decay}

Define for $k\in -\N$ the sets $Z_k=Z^-\cap ([0,1]\times [k,k+1])$ (these domains, as well as other auxiliary sets needed in the course of the derivation, are depicted in Figure \ref{fig:ZBLaux}). 

\begin{prop}
\label{prop:kappa}
Let $\bar{\rho}:=\rho+\divy(F^{-1}(x) \rho_1 ) \in L^2(Z^-)^2$ and let $\zeta$ and $\kappa$ be as above. Define
\[
r_k= \frac{1}{|Y^*|} \int_{Z_k} \kappa(y) \ud y.
\]
Then the following estimates hold:
\begin{gather*}
\nx{\kappa-r_k}{L^2(Z_k)} \leq C (  \nx{\grad \zeta}{L^2(Z_k)^4}   + \nx{\bar{\rho}}{L^2(Z_k)^2}) \\
| r_{k+1} -r_k | \leq C (  \nx{\grad \zeta}{L^2(Z_k\cup Z_{k+1}  )^4} + \nx{\bar{\rho}}{L^2(Z_k\cup Z_{k+1}  )^2}  )
\end{gather*}
\end{prop}
\begin{proof}
Define the space
\begin{align*}
V_k &=\Bigl\{  z\in H^1(Z_k)^2 \ |\ z=0 \text{ on } \p Z_k\backslash ( (\{0\} \cup  \{1 \})\times [k,k+1]  ),  z \text{ is $1$-periodic in $y_1$}  \Bigr\}.
\end{align*}
Consider Equation \eqref{subeq:stmaux1} on $Z_k$ with $\grady (\kappa - r_k)$ instead of $\grady \kappa$. By multiplication with a test function and integration by parts we obtain
\begin{align*}
\int_{Z_k} F^{-T}(x)\grady &\zeta(y): F^{-T}(x)\grady \phi(y) \ud y - \int_{Z_k} (\kappa - r_k) \divy(F^{-1}(x) \phi(y) ) \ud y \notag \\ &= \int_{Z_k} (\rho(y) + \divy(F^{-1}(x)\rho_1(y)) \cdot \phi(y) \ud y  \quad \forall \phi\in V_k. 
\end{align*}
Analogously to Lemma \ref{lem:divsurj} there exists $\phi_k \in V_k$, solution of
\[
\divy(F^{-1}(x) \phi_k(y)) = \kappa - r_k \quad \text{in $Z_k$}
\]
with
\[
\nx{\grady \phi_k}{L^2(Z_K)^4} \leq C \nx{\kappa-r_k}{L^2(Z_k)}.
\]
$C$ depends only on the geometry of $Y^*$ but not on $k$.

Inserting $\phi_k$ in the above equation and remarking that $\nx{F^{-T}(x)\grady z}{L^2} \leq C \nx{\grady z}{L^2}$ yields
\begin{align*}
\nx{\kappa-r_k}{L^2(Z_k)}^2 & \leq \nx{\bar{\rho}}{L^2(Z_k)^2}\nx{\grady \phi_k}{L^2(Z_k)^4} + C \nx{\grady \zeta}{L^2(Z_k)^4}\nx{\grady \phi_k}{L^2(Z_k)^4} \\
& \leq    C  \bigl(  \nx{\bar{\rho}}{L^2(Z_k)^2} + \nx{\grady \zeta}{L^2(Z_k)^4}  \bigr) \nx{\kappa-r_k}{L^2(Z_k)},
\end{align*}
thus the first assertion is proved.

Next, set $Z_{k,k+1}= Z_k \cup Z_{k+1} $ and consider $\phi_{k,k+1}$ satisfying
\begin{gather*}
\divy(F^{-1}(x) \phi_{k,k+1}(y)) =
\begin{cases}
1 & \text{ in $Z_k^0$} \\
-1 & \text{ in $Z_{k+1}^0$}
\end{cases}\\
\phi_{k,k+1}=0 \text{ on }(\p Z_k \cup \p Z_{k+1}) \backslash ( (\{ 0 \}\cup  \{1  \}) \times [k,k+2]  ) \\
\phi_{k,k+1} \text{ is $1$-periodic in $y_1$}
\end{gather*}
(the existence is assured since the right hand side of the first equation is in $L^2(Z_{k,k+1})$ and has mean value $0$).

Testing \eqref{subeq:stmaux1} with $\phi_{k,k+1}$ in $Z_{k,k+1}$ gives
\begin{align*}
-\int_{Z_k} \kappa(y) \ud y  &+ \int_{Z_{k+1}} \kappa(y) \ud y+ \int_{Z_{k,k+1}} F^{-T}(x)\grady \zeta(y): F^{-T}(x)\grady \phi_{k,k+1}(y) \ud y \\ & \qquad = \int_{Z_{k,k+1}} \bar{\rho}(y)\cdot\phi_{k,k+1}(y) \ud y.
\end{align*}
Note that $\nx{\phi_{k,k+1}}{L^2(Z_{k,k+1})^2} \leq C \nx{\grady \phi_{k,k+1}}{L^2(Z_{k,k+1})^4} \leq C |Z_{k,k+1}|$, thus dividing the equation by $|Y^*|$ gives the estimate
\[
| r_{k+1} -r_k | \leq C (  \nx{\grad \zeta}{L^2(Z_{k,k+1} )^4} + \nx{\bar{\rho}}{L^2(Z_{k,k+1}  )^2} ),
\]
which finishes the proof.
\end{proof}

%----
\begin{figure}[t]
\centering
\includegraphics[width=0.8\textwidth]{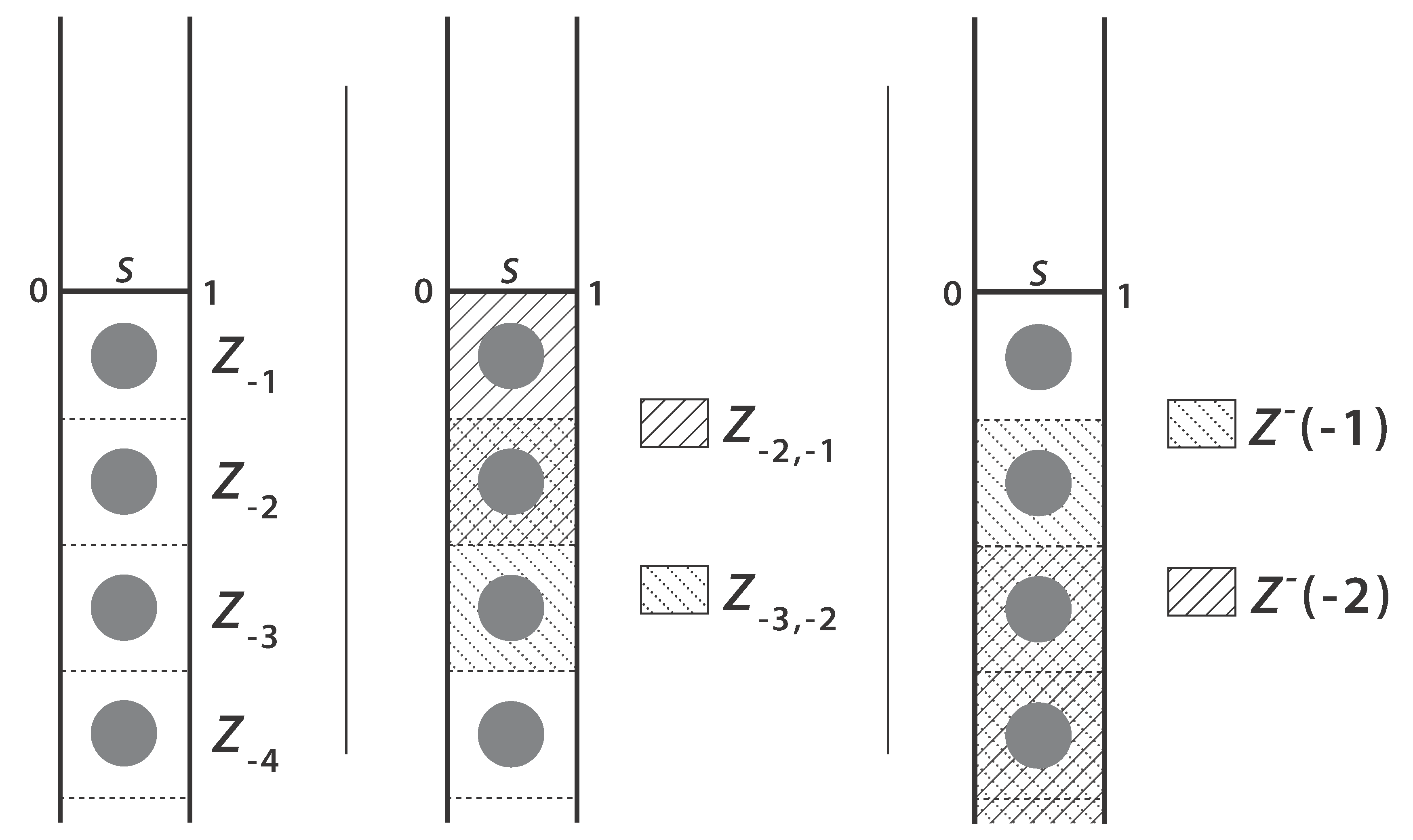}
\caption[Auxiliary domains based on the boundary layer cell $\ZBL$.]{Auxiliary domains based on the boundary layer cell $\ZBL$. {\itshape Left}: The translated reference cells $Z_k$, $k\in -\N$.      {\itshape Middle}:  The two-cell subsets $Z_{k,k+1}$, illustrated with the sets $Z_{-2,-1}$ (shaded with lines) and $Z_{-3,-2} $ (shaded with dots).    {\itshape Right}:  The unbounded strips $Z^-(k)$. In the figure the sets $Z^-(-1)$ (shaded with dots) and $Z^-(-2)$ (shaded with lines) are shown.  }
\label{fig:ZBLaux}
\end{figure}
%----

\begin{prop}
For $k\in -\N$ choose functions $\tilde{\sigma}_k \in \mathcal{C}^\infty(\R_{\leq 0})$, $0 \leq \tilde{\sigma}_k \leq 1$ with $\tilde{\sigma}_k(z)=0$ for $z\geq k+1$ and $\tilde{\sigma}_k(z)=1$ for $z \leq k$, $z\in \R_{\geq 0}$, such that $\tilde{\sigma}_k$ and the derivative $\tilde{\sigma}'_k$ are bounded uniformly in $k$. For $y= \binom{y_1}{y_2}\in [0,1]\times(-\infty,0]$ define $\sigma_k(y):=\tilde{\sigma}_k(y_2)$.

Let $\zeta,\kappa$ be a solution of Problem \eqref{eq:stmaux}. Then it holds
\begin{align*}
\int_{Z^-} | F^{-T}&(x)\grady \zeta(y) |^2  \sigma_k(y) \ud y= \int_{Z_k} (\kappa - r_k) \zeta \cdot F^{-T}(x)\grady \sigma_k(y) \ud y  \\ &+ \int_{Z^-} \bar{\rho}(y)\cdot\zeta(y)\sigma_k(y) \ud y  - \int_{Z^-} F^{-T}(x)\grady \zeta(y) : F^{-T}(x) (\zeta(y)\otimes \grady \sigma_{k}(y)) \ud y.
\end{align*}
\end{prop}
\begin{proof}
Testing \eqref{subeq:stmaux1} with $\phi\in \mathcal{C}^\infty_0(\ZBL)$, $\phi=0$ on $\bigcup_{k=1}^\infty \{  \p Y_S - {\binom{0}{k}} \}$ and $\phi$ $1$-periodic in $y_1$ yields
\begin{align*}
\int_{\ZBL} F^{-T}(x) \grady \zeta(y): F^{-T}(x) &\grady \phi(y) \ud y - \int_{\ZBL} \kappa \divy(F^{-1}(x)\phi(y)) \ud y %\\ & = 
\int_{\ZBL}\bar{\rho}(y) \cdot \phi(y) \ud y.
\end{align*}
Define for $l\leq k-1$ the functions $\sigma_{k,l}= \sigma_k (1-\sigma_l)$. Choosing $\phi=\zeta\sigma_{k,l}$ leads to
\begin{align*}
\int_{Z^-} & | F^{-T}(x)  \grady \zeta(y) |^2\sigma_{k,l}\ud y =  \int_{Z^-} \kappa(y)\zeta(y)\cdot F^{-T}(x)\grady\sigma_{k,l}(y) \ud y 
\\ & + \int_{Z^-} \bar{\rho}(y) \cdot \zeta(y)\sigma_{k,l}(y) \ud y - \int_{Z^-} F^{-T}(x)\grady \zeta(y) : F^{-T}(x) (\zeta(y)\otimes \grady \sigma_{k,l}(y)) \ud y,
\end{align*}
where we used the fact that
\begin{gather*}
F^{-T}(x)\grady ( \zeta \sigma_{k,l})= F^{-T}(x) (\grady \zeta) \sigma_{k,l} + \zeta \otimes F^{-T}(x)\grady \sigma_{k,l},\\
\divy(F^{-1}(x) \zeta\sigma_{k,l})= \zeta\cdot F^{-T}(x)\grady \sigma_{k,l} + \sigma_{k,l}\divy(F^{-1}(x) \zeta). 
\end{gather*}

We want to pass to the limit $l\longrightarrow -\infty$ for fixed $k$. First observe that $\sigma_{k,l}\longrightarrow \sigma_k$ as well as $\grad \sigma_{k,l} \longrightarrow \grad \sigma_k$ pointwise for $l\longrightarrow -\infty$. As $|\sigma_{k,l}| \leq C$ and $|\grad \sigma_{k,l}|= |(\grad \sigma_k) (1+\sigma_l)-\sigma_k \grad\sigma_l |\leq C$ a.e. with a constant $C$, we obtain that almost everywhere
\begin{align*}
\left| \  | F^{-T}(x)\grady \zeta(y) |^2  \sigma_k(y) \right|& \leq C | F^{-T}(x)\grady \zeta(y) |^2  \\
\left|\bar{\rho}(y)\cdot\zeta(y)\sigma_{k,l}(y)\right| &\leq C \left|\bar{\rho}(y)\cdot\zeta(y)\right| \\
\left| F^{-T}(x)\grady \zeta(y) : F^{-T}(x) (\zeta(y)\otimes \grady \sigma_{k,l}(y))\right| &\leq C \left| F^{-T}(x)\grady \zeta(y) : F^{-T}(x) (\zeta(y)\otimes I)\right|
\end{align*}
where $I$ denotes the identity matrix.
Since the right hand sides are integrable, application of Lebesgue's dominated convergence theorem yields for $l\longrightarrow -\infty$
\vspace{0.1cm}
\begin{gather*}
\int_{Z^-} | F^{-T}(x) \grady \zeta(y) |^2\sigma_{k,l} \ud y \quad \longrightarrow  \quad\int_{Z^-} | F^{-T}(x) \grady \zeta(y) |^2\sigma_{k} \ud y\\
 \int_{Z^-} \bar{\rho}(y)\cdot\zeta(y)\sigma_{k,l}(y) \ud y \quad\longrightarrow\quad  \int_{Z^-} \bar{\rho}(y)\cdot\zeta(y)\sigma_{k}(y) \ud y\\
\intertext{and}
\int_{Z^-} F^{-T}(x)\grady \zeta(y) : F^{-T}(x) (\zeta(y)\otimes \grady \sigma_{k,l}(y)) \ud y \qquad \qquad \qquad\qquad
\\  \qquad\qquad\qquad\qquad\longrightarrow\quad \int_{Z^-} F^{-T}(x)\grady \zeta(y) : F^{-T}(x) (\zeta(y)\otimes \grady \sigma_{k}(y)) \ud y.
\end{gather*}
Finally we have to consider the term $\int_{Z^-} \kappa(y)\zeta(y)\cdot F^{-T}(x)\grady\sigma_{k,l}(y)$. Because of $\grad \sigma_{k,l}(y)=0$ a.e. for $y\not\in Z_k \cup Z_l$ we have
\begin{align*}
\int_{Z^-} &\kappa(y) \zeta(y)\cdot F^{-T}(x)\grady\sigma_{k,l}(y) \ud y = \int_{Z_k \cup Z_l} (\kappa(y) -r_k) \zeta(y)\cdot F^{-T}(x)\grady\sigma_{k,l}(y)\ud y \\
= &   \int_{Z_k} (\kappa(y) -r_k) \zeta(y)\cdot F^{-T}(x)\grady\sigma_{k,l}(y) \ud y + \int_{Z_l} (\kappa(y) -r_l) \zeta(y)\cdot F^{-T}(x)\grady\sigma_{k,l}(y) \ud y\\& + (r_l -r_k) \int_{Z_l} \zeta(y)\cdot F^{-T}(x)\grady\sigma_{k,l}(y) \ud y.
\end{align*}
For $l\longrightarrow - \infty$ we obtain by using Poincar\'e's inequality 
\begin{gather*}
\biggl| (r_l -r_k) \int_{Z_l} \zeta(y)\cdot F^{-T}(x)\grady\sigma_{k,l}(y) \ud y \biggr| \leq C  \nx{\grad \zeta}{L^2(Z_l)^4}  \longrightarrow 0
\end{gather*}
and
\begin{align*}
\biggl|\int_{Z_l} (\kappa(y) -r_l) \zeta(y)\cdot F^{-T}(x)\grady\sigma_{k,l}(y) \ud y \biggr| & \leq 
C (  \nx{\grad \zeta}{L^2(Z_l)^4}   + \nx{\bar{\rho}}{L^2(Z_l)^2})\nx{\grad \zeta}{L^2(Z_l)^4} \\ & \qquad  \longrightarrow 0,
\end{align*}
where we also used the preceding lemma for the last estimate.
Thus arguing similarly with Lebesgue's theorem one arrives at
\[
\lim_{l\rightarrow - \infty} \int_{Z^-} \kappa(y) \zeta(y)\cdot F^{-T}(x)\grady\sigma_{k,l}(y) \ud y = \int_{Z_k} (\kappa(y) -r_k) \zeta(y)\cdot F^{-T}(x)\grady\sigma_{k}(y)  \ud y,
\]
and the proof is complete.
\end{proof}

Define for $k\in -\N$ the sets
\[
Z^-(k):=Z^- \cap ( [0,1]\times (-\infty,k]  ).
\]

\begin{prop}
\label{prop:zeta}
Let $\bar{\rho}\in L^2(Z^-)^2$ and let $\zeta,\kappa$ be a solution of problem \eqref{eq:stmaux}. There exists a constant $C_0$ independent of $k$ such that
\[
\nx{\grad \zeta}{L^2(Z^-(k))}^2 \leq C_0^2 \nx{\bar{\rho}}{L^2(Z^-(k))^2}^2.
\]
\end{prop}
\begin{proof}
We estimate the terms on the right hand side of the previous proposition separately: By the Poincar\'e inequality
\begin{align*}
\biggl| \int_{Z^-} F^{-T}(x) & \grady \zeta(y) : F^{-T}(x) (\zeta(y)\otimes \grady \sigma_{k,l}(y)) \ud y \biggr| \\ & =  \biggl| \int_{Z_k} F^{-T}(x)\grady \zeta(y) : F^{-T}(x) (\zeta(y)\otimes \grady \sigma_{k,l}(y)) \ud y \biggr| \\
& \leq C\nx{\grady \zeta}{L^2(Z_k)^4}^2 
\end{align*}
and
\begin{align*}
\biggl|\int_{Z^-} \bar{\rho}(y)\cdot\zeta(y)\sigma_k(y) \ud y \biggr| =\biggl| \int_{Z^-(k)} \bar{\rho}(y)\cdot\zeta(y)\sigma_k(y) \ud y\biggr| \leq C \nx{\grady \zeta}{L^2(Z^-(k))^4}\nx{\bar{\rho}}{L^2(Z^-(k))^2}.
\end{align*}
Using Proposition \ref{prop:kappa} gives
\begin{align*}
\biggl|\int_{Z_k} (\kappa - r_k) \zeta \cdot F^{-T}(x)\grady \sigma_k(y) \ud y \biggr| \leq C  \nx{\grady \zeta}{L^2(Z_k)^4}^2   + \nx{\bar{\rho}}{L^2(Z_k)^2}\nx{\grad \zeta}{L^2(Z_k)^4}.
\end{align*}
Because of $Z^-(k)\subset Z^-(k+1)$, $Z_k\subset Z^-(k+1)$ and Young's inequality we obtain
\begin{align*}
k_F \int_{Z^-} |\grady &\zeta(y)|^2 \sigma_k(y) \ud y  \leq \int_{Z^-} |F^{-T}(x)\grady \zeta(y)|^2 \sigma_k(y) \ud y \\ & \leq C^* \nx{\grady \zeta}{L^2(Z_k)^4}^2 + C\delta \int_{Z^-(k)} |\grady \zeta(y)|^2 \ud y + \frac{C}{\delta}\nx{\bar{\rho}}{L^2(Z^-(k+1))^2}^2
\end{align*}
for $\delta >0$. 
Next observe that 
\[
\int_{Z^-} |\grady \zeta(y)|^2 \sigma_k(y) \ud y = \int_{Z^-(k)} |\grady \zeta(y)|^2 \sigma_k(y) \ud y  + \int_{Z_k} |\grady \zeta(y)|^2 \sigma_k(y)  \ud y
\]
and 
\[
\nx{\grady \zeta}{L^2(Z_k)^4}^2  = \int_{Z^-(k+1)} |\grady \zeta(y)|^2 \ud y - \int_{Z^-(k)} |\grady \zeta(y)|^2  \ud y,
\]
thus leading to
\begin{align*}
(k_F-C\delta + C^* ) \int_{Z^-(k)}|\grady \zeta(y)|^2  \ud y \leq C^* \nx{\grady \zeta}{L^2(Z^-(k+1))^4}^2 + C_1\nx{\bar{\rho}}{L^2(Z^-(k+1))^2}^2.
\end{align*}
Choosing $\delta$ small enough such that $k_F-C\delta + C^* >0  $ and $k_F > C\delta$ gives the recursion
\[
a_k \leq \gamma a_{k+1} + F_k, \quad k\in -\N
\]
with
\begin{gather*}
a_k=\nx{\grad \zeta}{L^2(Z^-(k))}^2, \qquad \gamma= \frac{C^*}{k_F-C\delta + C^*}< 1, \\
F_k = \frac{C_1}{k_F-C\delta + C^*} \nx{\bar{\rho}}{L^2(Z^-(k+1))^2}^2.
\end{gather*}
Since $Z^-(k)\subset Z^-(k+1)$ we also have $F_k\leq F_{k+1}$.
This implies the claim as in \cite{jami_bc-fluidpor}. %(WIE DENN???)
%By observing that $a_k \leq a_{k+1}$ we get
%\begin{align*}
%
%\end{align*}
\end{proof}

\begin{cor}
Consider the situation as above. Then there exists a constant $\kappa_\infty$,
\[
\kappa_\infty = \lim_{k\rightarrow - \infty} \frac{1}{|Y^*|} \int_{Z_k} \kappa(y) \ud y
\]
and a constant $C^*$, independent of $k$, such that for $k\in -\N$ holds
\[
\nx{\kappa - \kappa_\infty}{L^2(Z^-(k))}^2 \leq C^*  \sum_{l=-\infty}^k \nx{\bar{\rho}}{L^2(Z^-(l+1))^2}^2.
\]
\end{cor}

\begin{proof}
Proposition \ref{prop:kappa} yields
\begin{gather*}
\nx{\kappa-r_k}{L^2(Z_k)} \leq C (  \nx{\grad \zeta}{L^2(Z_k)^4}   + \nx{\bar{\rho}}{L^2(Z_k)^2}) \\
| r_{k+1} -r_k | \leq C (  \nx{\grad \zeta}{L^2(Z_{k,k+1} )^4} + \nx{\bar{\rho}}{L^2(Z_{k,k+1} )^2}  ).
\end{gather*}
We show that $r_k$ is a Cauchy sequence in $\R$, thus providing the existence of $\kappa_\infty$: By the triangle inequality it holds for $k\in -\N$, $l \leq 0$
\begin{align*}
|r_{k+l} - r_k| & \leq \sum_{j=l}^{1} |r_{k+j}-r_{k+j-1}| \\
 & \leq \sum_{j=l}^{1} C  (  \nx{\grad \zeta}{L^2(Z_{k+j-1,k+j} )^4} + \nx{\bar{\rho}}{L^2(Z_{k+j-1,k+j} )^2}  ) \\
& \leq C  (  \nx{\grad \zeta}{L^2(Z^-(k-1) )^4} + \nx{\bar{\rho}}{L^2(Z^-(k-1) )^2}  ),
\end{align*} 
where the last constant is independent of $k$ and $l$. Since the last term converges to $0$ for $k\rightarrow -\infty$, we obtain the desired result.

%Since the right hand side of the last inequality converges to $0$ we get that $r_k$ is a Cauchy sequence in $\R$, giving the existence of $\kappa_\infty$.

Next observe that
\begin{align*}
 \nx{\kappa - \kappa_\infty}{L^2(Z_m)} 
& \leq     \nx{\kappa - r_m}{L^2(Z_m)} + \nx{r_m-\kappa_\infty}{L^2(Z_m)}   \\
& \leq C (  \nx{\grad \zeta}{L^2(Z_m)^4}   + \nx{\bar{\rho}}{L^2(Z_m)^2})  + |Y^*|\ |r_m+\kappa_\infty|  \\
& \leq  C (  \nx{\grad \zeta}{L^2(Z^-(m+1))^4}   + \nx{\bar{\rho}}{L^2(Z^-(m+1))^2}) \\ & \quad +  |Y^*|\lim_{j\rightarrow \infty}\biggl(  \sum_{l=0}^j  |r_{m-l} - r_{m-(l+1)}| + |r_{l-(j+1)}-\kappa_\infty| \biggr)    \\
&\leq   C  \nx{\bar{\rho}}{L^2(Z^-(m+1))^2}   + \sum_{l=0}^\infty  C ( \nx{\grad \zeta}{L^2(Z_{m-(l+1),m-l} )^4} + \nx{\bar{\rho}}{L^2(Z_{m-(l+1),m-l} )^2}  )\\
& \leq  C  \nx{\bar{\rho}}{L^2(Z^-(m+1))^2}   +2   C ( \nx{\grad \zeta}{L^2(Z^-(m+1) )^4} + \nx{\bar{\rho}}{L^2(Z^-(m+1) )^2}  )   \\
& \leq C \nx{\bar{\rho}}{L^2(Z^-(m+1))^2},
\end{align*}
where we used the above inequalites and Proposition \ref{prop:zeta}. Furthermore, note that $\lim_{j\rightarrow \infty}  |r_{l-(j+1)}-\kappa_\infty|  =0$. 
Thus by
\begin{align*}
\nx{\kappa - \kappa_\infty}{L^2(Z^-(k))}^2 & \leq \sum_{m=-\infty}^k \nx{\kappa - \kappa_\infty}{L^2(Z_m)} \\
& \leq C^* \sum_{m=-\infty}^k \nx{\bar{\rho}}{L^2(Z^-(m+1))^2}
\end{align*}
the second assertion holds. 
\end{proof}

Finally we are able to get a result on the decay of the solutions $\zeta,\kappa$ in the porous part $Z^-$ of $\ZBL$:

\begin{cor}
Assume that $e^{\gamma_1 |y_2|}\bar{\rho} \in L^2(\ZBL)^2$ for a $\gamma_1>0$. Then there exists a $\beta>0$ such that for the solution  $\zeta,\kappa$ of Problem \eqref{eq:stmaux} holds
\begin{empheq}[box=\widefbox]{align*}
\nx{\grad \zeta}{L^2(Z^-(k))^4} & \leq C e^{-\beta |k|}\\
\nx{\zeta}{L^2(Z^-(k))^2} & \leq C e^{-\beta |k|}\\
\nx{\kappa- \kappa_\infty}{L^2(Z^-(k))\phantom{2}} & \leq C e^{-\beta |k|}
\end{empheq}
\end{cor}
\begin{proof}
By the assumption on $\bar{\rho}$ note that $\nx{\bar{\rho}}{L^2(Z_k)^2}\leq C e^{-\gamma_1 |k|}$. Therefore
\begin{align*}
\nx{\bar{\rho}}{L^2(Z^-(l))^2} & \leq C \sum_{k=-\infty}^l e^{-\gamma_1 |k|} = C e^{-\gamma_1 |l|} \sum_{k=-\infty}^0 (e^{\gamma_1})^{-|k|} \\
& =  \frac{C}{1-e^{\gamma_1}} e^{-\gamma_1|l|},
\end{align*}
where we used the formula for the geometric series for $e^{\gamma_1} > 1$.
Using the same argument once again, one obtains
\[
\sum_{l=-\infty}^k \nx{\bar{\rho}}{L^2(Z^-(l))} \leq C e^{-\gamma_1 |k|} ,
\]
which gives the first and the last assertion. The second one follows due to Poincar\'e's inequality.
\end{proof}

In order to deal with the behavior of $\zeta$ and $\kappa$ in $Z^+$, we are going to use the theory for the exponential decay of solutions of elliptic problems, developed by Landis/Panasenko and Ole\u{\i}nik/Iosif'jan, see \cite{lapa_phragmlind} and \cite{olio_behavinf}.

\begin{thm}[Exponential Decay]
\label{thm:expdecay}
Let the geometry be given as above. In $Z^+$ consider the elliptic equation
\[
-\divy(F(y)\grady u(y)) = f(y)
\]
with a given matrix function $F\in L^\infty(Z^+)^4$ satisfying the following ellipticity condition: Let there exist constants $c_1, C_1>0$ and $M>0$ ($M= 1$ in case $F$ is symmetric) such that for all $\eta, \xi \in \R^2$
\begin{gather*}
c_1 |\xi|^2 \leq \xi^T F(y) \xi \leq C_1 |\xi|^2 \\
| \eta^T F(y) \xi | \leq M ( \eta^T F(y) \eta  )^\frac{1}{2}  (\xi^T F(y) \xi )^\frac{1}{2}.
\end{gather*}
Assume further periodic boundary conditions on $(\{0 \} \cup \{1 \} \times \R_{\geq 0})$ and Dirichlet and/or Neumann conditions on $S$ such that there exists a solution $u$ with $\grad u \in L^2(Z^+) $. Let there exist constants $q,Q >0$ such that $Qe^{q y_2} f \in L^2(Z^+)$.

Then there exist constants $q_1,Q_1>0$ and $C_u$ such that
\begin{align*}
\nx{\grad u}{L^2(Z^+ \cap ( [0,1] \times [k, \infty) ))^2} &\leq Q_1 e^{-q_1 k} \\
\nx{u-C_u}{L^2(Z^+ \cap ( [0,1] \times [k, \infty) ))} &\leq Q_1 e^{-q_1 k}.
\end{align*}
Furthermore, there exists $y^* >0$ with
\[
| u(y) - C_u  | \leq Q_1 e^{-q_1 y_2} \quad \text{for } y_2 > y^*.
\]
\end{thm}
\begin{proof}
Theorem 10 of \cite{olio_behavinf}  gives the first two estimates.

Due to the lifting property of elliptic operators, we obtain a solution $u\in H^2(Z^+)$; and because of the embedding $H^2(Z^+) \hookrightarrow \mathcal{C}^0(Z^+)$ there exists a continuous representative. Therefore we can apply Theorem 2 in \cite{lapa_phragmlind} in order to get the pointwise estimate.
\end{proof}

\begin{prop}
Assume that $(\zeta, \kappa)$ is a solution of Problem \eqref{eq:stmaux} with $e^{\gamma_1 y_2}\rho \in H^1(Z^+)^2$, $e^{\gamma_1 y_2}\rho_1\in H^1(Z^+)^4$ and $e^{\gamma_1 y_2}\divy(F^{-1}(x)\rho_1)\in H^1(Z^+)^2$.

There exist $\beta>0$, $y^*>0$, a vector $C_\zeta \in \R^2$ and a constant $C_\kappa$ such that
\begin{empheq}[box=\widefbox]{align*}
\nx{\grad \zeta}{L^2(Z^+ \cap ( [0,1] \times [k, \infty) ))^4}  & \leq C e^{-\beta k} \\
\nx{ \zeta - C_\zeta}{L^2(Z^+ \cap ( [0,1] \times [k, \infty) ))^2}  & \leq C e^{-\beta k} \\
\nx{\kappa - C_\kappa}{L^2(Z^+ \cap ( [0,1] \times [k, \infty)\phantom{^2} ))}  & \leq C e^{-\beta k} 
\end{empheq}
{and}
\begin{empheq}[box=\widefbox]{align*}
|\zeta(y) - C_\zeta | & \leq C e^{-\beta y_2} \quad\text{ for } y_2> y^*\\
|\kappa(y) - C_\kappa | & \leq C e^{-\beta y_2} \quad\text{ for } y_2> y^*
\end{empheq}
\end{prop}
\begin{proof}
Set $\xi=\curlt(\zeta)$. By taking the $\curlt$ of Equation \eqref{eq:stmaux} we obtain due to Lemma~\ref{lem:transform2}:
\begin{gather*}
\divy(F^{-1}(x)F^{-T}(x) \grady \xi(y))= - \curlt \Bigl(\rho(y) + \divy\bigl(F^{-1}(x)\rho_1(y)\bigr) \Bigr) \quad \text{ in } Z^+ .
\end{gather*}

The right hand side decays exponentially, thus by the preceding theorem we obtain
\begin{gather*}
\nx{\grad (\curlt \zeta)}{L^2(Z^+ \cap ( [0,1] \times [k, \infty) ))^2}   \leq C e^{-\beta k} \\
\nx{\curlt \zeta - C_C}{L^2(Z^+ \cap ( [0,1] \times [k, \infty) ))}   \leq C e^{-\beta k}  
\end{gather*}
with some constants $\beta>0$ and  $C_C $.

Using the 7th assertion of the transformation lemma \ref{lem:transform2} we see that
\begin{align*}
F^{-T}(x) \grady(\xi(y)) = F^{-T}(x) \grady\bigl(\curlt \zeta(y)\bigr) = \begin{bmatrix}
                    0& -1 \\
                    1 & \phantom{-}0
                   \end{bmatrix} \divy\bigl(F^{-1}(x) F^{-T}(x) \grady \zeta(y) \bigr).
\end{align*}
Therefore
\[
\divy\bigl(F^{-1}(x) F^{-T}(x) \grady \zeta(y) \bigr) = h(y) \quad \text{ in } Z^+,
\]
$h$ being a known function with $e^{\beta y_2} h\in L^2(Z^+)$. Theorem \ref{thm:expdecay} now shows that the first two asserted inequalities about the decay of $\zeta$ and $\grad\zeta$ hold.

By taking the transformed divergence of Equation \eqref{subeq:stmaux1} one obtains
\[
\divy\bigl(F^{-1}(x)F^{-T}(x) \grady \kappa(y)\bigr)= - \divy\Bigl(F^{-1}(x) \bigl[\rho(y) + \divy(F^{-1}(x)\rho_1(y) )\bigr]\Bigr) \quad \text{ in } Z^+.
\]
Again, the right hand side decays exponentially, and the estimate for $\kappa$ is proved. The remaining two inequalities follow easily.
\end{proof}

At the end of this section, we want to obtain some information about the constant $C_\zeta$:
\begin{lem}
\label{lem:blfunc}
For the solution $\zeta$ of Problem \eqref{eq:stmaux} it holds for all $z<0$
\[
\int_0^1 \zeta(y_1,z) \cdot F^{-T}(x)e_2 \ud y_1=\int_0^1 \zeta(y_1,0) \cdot F^{-T}(x)e_2\ud y_1 =0 ,
\]
and $\int_S \zeta \cdot F^{-T}(x)e_2 \ud \sigma_y =0 $.
\end{lem}
\begin{proof}
By density it is enough to show the claim for $\zeta\in W\cap \mathcal{C}^\infty_0(\ZBL)^2$.

Integration of the equation $\div_y(F^{-1}(x)\zeta(y))$ =0 over $[0,1]\times (z,0)$ and application of Stokes theorem yields due to the periodic boundary conditions
\[
\int_0^1 \zeta(y_1,0)\cdot F^{-T}(x)e_2\ud y_1 =  \int_0^1 \zeta(y_1,z)\cdot F^{-T}(x)e_2 \ud y_1 =:K_\zeta  \quad \forall z<0.
\]
Since $\zeta\in W$ it holds
\[
\int_{-\infty}^0 \biggl( \int_0^1 \zeta(y_1,t)\cdot F^{-T}(x)e_2 \ud y_1  \biggr)^2 \ud t \leq C \int_{-\infty}^0 \int_0^1 |\zeta(y_1,t)|^2  \ud y_1  \ud t  < \infty,
\]
and thus $K_\zeta=0$.
\end{proof}

\begin{lem}
\label{lem:constexpl}
For the constant $C_\zeta$ it holds $C_\zeta\cdot F^{-T}(x) e_2 =0$.
\end{lem}
\begin{proof}
Arguing as in the proof of the above lemma, we obtain
\[
\int_0^1 \zeta(y_1,k) \cdot F^{-T}(x)e_2 \ud y_1=\int_0^1 \zeta(y_1,-k) \cdot F^{-T}(x)e_2  \ud y_1
\]
for all $k>0$. Now the left hand side converges exponentially to $\int_0^1 C_\zeta \cdot F^{-T}(x)e_2 \ud y_1= C_\zeta \cdot F^{-T}(x)e_2$, whereas the right hand side converges to $0$. 
\end{proof}

%\formatwo{\enlargethispage{1cm}}

\subsubsection{Application to the Stokes Boundary Layer Problems}

We apply the results of the foregoing section to the problem: Find $(\bbl,\obl) \in V \times L^2_\mathrm{loc}(\ZBL)$ such that for fixed $x\in \O$ it holds 
\begin{subequations}
\label{eq:stbl}
\begin{empheq}[box=\widefbox]{align}
-\divy(F^{-1}(x)F^{-T}(x) \grady \bbl(y)) + F^{-T}(x)\grady\obl(y)=0 & \quad \text{ in $Z$} \label{subeq:wibl}\\
\divy(F^{-1}(x)\bbl(y))=0 &\quad \text{ in $Z$}  \label{subeq:wibldiv}\\
[ \bbl ]_S (y)= 0 & \quad\text{ on $S$} \\
[(F^{-1}(x)F^{-T}(x)\grady \bbl-F^{-1}(x)\obl)e_2]_S (y)  =K^\mathrm{bl}(x)& \quad\text{ on $S$} \\
\bbl(y)=0  \quad \text{ on }  {\textstyle \bigcup_{k=1}^\infty }&  {\textstyle \{  \p Y_S - \binom{0}{k} \} }\\
\bbl, \obl \text{ are $1$-periodic in $y_1$}
\end{empheq}
\end{subequations}
where $K^\mathrm{bl}(x)=F^{-1}(x)F^{-T}(x)\grad u^0(x) e_2$. This corresponds to the case $\rho = \rho_1=0$ and $\sigma=K^\mathrm{bl}(x)$.
Lemma \ref{lem:constexpl} shows that 
\begin{equation}
\label{eq:expstabvelocity}
\Cbl\cdot F^{-T}(x) e_2=0.
\end{equation}
Finally, we can obtain the complete information about the constants:
\begin{lem}
\label{lem:expstabpress}
For all $0<a<b$ it holds
\begin{align*}
\int_0^1 \obl(y_1,a) \ud y_1 = \int_0^1 \obl(y_1,b) \ud y_1.
\end{align*}
Thus the constant $\Cblo$ arising in the stabilization of the pressure is given by
\begin{equation}
\label{eq:Cblo}
\Cblo = \int_0^1 \obl(y_1,+0) \ud y_1.
\end{equation}
\end{lem}
\begin{proof}
Due to \eqref{subeq:wibldiv} and the actual entries of $F^{-1}(x)$ it holds 
\begin{equation}
\label{eq:wibldivfree}
\frac{\p}{\p y_1}\bbl_1 + \frac{\p}{\p y_2}\bbl_2 -g'(x_1) \frac{\p}{\p y_2} \bbl_1 =0.
\end{equation}
Note that $F^{-T}(x) \grady \obl(y) = \divy(F^{-1}(x)\obl(y))$ (cf.\ Lemma \ref{lem:transform2}), thus Equation~\eqref{subeq:wibl} reads column-wise
\begin{gather*}
\divy\Bigl(F^{-1}(x)\bigl(  F^{-T}(x)\grady \bbl_1(y) - \obl(y) e_1  \bigr)    \Bigr) =0 \\
 \intertext{and}
\divy\Bigl(F^{-1}(x)\bigl(  F^{-T}(x)\grady \bbl_2(y) - \obl(y) e_2  \bigr)    \Bigr) =0.
\end{gather*}
Let $0<a<b$. Now integration of the equation 
\begin{align*}
0 & =\divy\Bigl(F^{-1}(x)\bigl(  F^{-T}(x)\grady \bbl_2(y) - \obl(y)e_2  \bigr)    \Bigr) \\ &\qquad  - g'(x_1) \divy\Bigl(F^{-1}(x)\bigl(  F^{-T}(x)\grady \bbl_1(y) - \obl(y) e_1  \bigr)    \Bigr) 
\end{align*}
over the rectangle $[0,1]\times[a,b]$ yields due to Stokes' theorem and the periodicity of $\bbl$ and $\obl$ in $y_1$-direction
\begin{align*}
0& =\int_0^1 \bigl(F^{-T}(x)\grady \bbl_2 - \obl e_2 \bigr)(y_1,b) \cdot F^{-T}(x)e_2 \ud y_1 \\ & \qquad- \int_0^1 \bigl(F^{-T}(x)\grady \bbl_2 - \obl e_2 \bigr)(y_1,a)  \cdot F^{-T}(x)e_2\ud y_1 \\
&\qquad -g'(x_1) \int_0^1 \bigl(F^{-T}(x)\grady \bbl_1 - \obl e_1 \bigr) (y_1,b)\cdot F^{-T}(x)e_2 \ud y_1 \\
& \qquad +g'(x_1) \int_0^1 \bigl(F^{-T}(x)\grady \bbl_1 - \obl e_1 \bigr)(y_1,a) \cdot F^{-T}(x)e_2 \ud y_1.
\end{align*}

Now we have (by using the actual form of $F^{-T}(x)$)
\begin{align*}
\obl(y)e_2 \cdot F^{-T}(x) e_2 = \obl(y) \quad \text{and}\quad \obl(y)e_1 \cdot F^{-T}(x)e_2 = -g'(x_1) \obl(y)
\end{align*}
as well as
\begin{align*}
F^{-T}(x) \grady \bbl_1(y) \cdot F^{-T}(x)e_2& = (1+g'(x_1)^2) \frac{\p}{\p y_2} \bbl_1(y) -g'(x_1) \frac{\p}{\p y_1} \bbl_1(y)\\
F^{-T}(x) \grady \bbl_2(y) \cdot F^{-T}(x)e_2&= (1+g'(x_1)^2) \frac{\p}{\p y_2} \bbl_2(y) -g'(x_1) \frac{\p}{\p y_1} \bbl_2(y)\\
&=(1+g'(x_1)^2)\Bigl(   g'(x_1) \frac{\p}{\p y_2} \bbl_1(y) - \frac{\p}{\p y_1} \bbl_1  \Bigr) \\ & \qquad - g'(x_1) \frac{\p}{\p y_1} \bbl_2,
\end{align*}
where the identity \eqref{eq:wibldivfree} was used in the last equation.
Substituting these results in the above equation, one obtains 
\begin{align*}
0&= \int_0^1 \Bigl(  -g'(x_1) \frac{\p}{\p y_1} \bbl_2 - (1+g'(x_1)^2)\frac{\p}{\p y_1} \bbl_1 +g'(x_1)\frac{\p}{\p y_1}\bbl_1  \Bigr)(y_1,b) \ud y_1 \\
& \qquad - \int_0^1 \Bigl(  -g'(x_1) \frac{\p}{\p y_1} \bbl_2 - (1+g'(x_1)^2)\frac{\p}{\p y_1} \bbl_1 +g'(x_1)\frac{\p}{\p y_1}\bbl_1  \Bigr)(y_1,a) \ud y_1 \\
& \qquad - (1+g'(x_1)^2) \int_0^1 \obl(y_1,b) \ud y_1% \\
%& \qquad
 + (1+g'(x_1)^2) \int_0^1 \obl(y_1,a) \ud y_1.
\end{align*}
The first two integrals vanish due to the fundamental theorem of calculus and the periodic boundary conditions. We divide by $(1+g'(x_1)^2)$ to obtain
\[
\int_0^1 \obl(y_1,a) \ud y_1 = \int_0^1 \obl(y_1,b) \ud y_1 \qquad \forall\  0<a<b.
\]
This proves the first statement.

To obtain the second one, notice that for $k>0$ due to Jensen's inequality
\begin{align*}
\int_k^\infty \Bigl|  \int_0^1 \obl(y_1,y_2) \ud y_1 - \Cblo \Bigr|^2 \ud y_2 & \leq \int_{[0,1]\times (k,\infty)} |\obl(y_1,y_2)  - \Cblo|^2 \ud y \\
& \qquad \longrightarrow 0 \quad \text{for } k\rightarrow \infty
\end{align*}
because of the exponential stabilization of $\obl$; therefore $\int_0^1 \obl(y_1,b) \ud y_1$ converges to $\Cblo$ for $b \rightarrow \infty$. % in $L^2((k,\infty))$-sense.
Now letting $a\rightarrow +0$ yields the result. %(NOCH ETWAS UNKLAR)
\end{proof}

\begin{lem}
\label{lem:expstabveloc}
%Assume that $|g'(z)|<1$ for all $ z\in [0,L]$. %Then for $b>0$ it holds
%\[
%\int_0^1 \wibl(y_1,0)\cdot F(x)e_1 \ud y_1 = \int_0^1 \wibl(y_1,b)\cdot F(x) e_1 \ud y_1
%\]
%Hence 
For the constant $\Cbl$ appearing in the exponential stabilization of the velocity $\bbl$ it holds
\begin{equation}
\Cbl \cdot F(x) e_1= \int_0^1 \bbl(y_1,+0) \cdot F(x) e_1\ud y_1.
\label{eq:expstabveloc}
\end{equation}
\end{lem}
\begin{proof} %\formatwo{\enlargethispage{2cm}}
Let $b>0$. Similarly to the above lemma, we multiply the equation 
\begin{align*}
0 & =\divy\Bigl(F^{-1}(x)\bigl(  F^{-T}(x)\grady \bbl_1(y) - \obl e_1  \bigr)    \Bigr) \\ &\qquad  + g'(x_1) \divy\Bigl(F^{-1}(x)\bigl(  F^{-T}(x)\grady \bbl_2(y) - \obl e_2  \bigr)    \Bigr) 
\end{align*}
by $y_2$ and integrate over $[0,1]\times [0,b]$. Integration by parts then yields
\begin{align*}
0& =  - \int_{[0,1]\times [0,b]} F^{-1}(x) \bigl(  F^{-T}(x)\grady \bbl_1(y) - \obl(y) e_1  \bigr)\cdot e_2 \ud y \\ & \qquad  -g'(x_1)  \int_{[0,1]\times [0,b]} F^{-1}(x) \bigl(  F^{-T}(x)\grady \bbl_2(y) - \obl(y) e_2  \bigr) \cdot e_2 \ud y \\
& \qquad + \int_0^1 b\bigl(F^{-T}(x) \grady \bbl_1 - \obl e_1\bigr)(y_1,b) \cdot F^{-T}(x)e_2 \ud y_1 \\
& \qquad + g'(x_1) \int_0^1 b\bigl(F^{-T}(x) \grady \bbl_2 - \obl e_2\bigr)(y_1,b) \cdot F^{-T}(x)e_2 \ud y_1.
\end{align*}
As in the proof of the preceeding lemma we have 
$F^{-1}(x)\obl(y)e_2 \cdot e_2 = \obl(y)e_2 \cdot F^{-T}(x)e_2 = \obl(y)$
and
$\obl(y)e_1 \cdot F^{-T}(x)e_2 = -g'(x_1) \obl(y)$, 
thus the terms containing the pressure $\obl$ cancel out and we have 
\begin{align*}
0 &= - \int_{[0,1]\times [0,b]} F^{-1}(x)   F^{-T}(x)\grady \bbl_1(y) \cdot e_2 \ud y \\ & \qquad  -g'(x_1)  \int_{[0,1]\times [0,b]} F^{-1}(x)   F^{-T}(x)\grady \bbl_2(y) \cdot e_2 \ud y \\
& \qquad + \int_0^1 bF^{-T}(x) \grady \bbl_1(y_1,b) \cdot F^{-T}(x)e_2 \ud y_1 \\
& \qquad + g'(x_1) \int_0^1 bF^{-T}(x) \grady \bbl_2(y_1,b) \cdot F^{-T}(x)e_2 \ud y_1.
\end{align*}

Another integration by parts of the volume terms now yields %\formatwo{\enlargethispage{2cm}}
\begin{align*}
0 &= - \int_0^1 \bbl_1(y_1,b) F^{-1}(x)F^{-T}(x) e_2 \cdot e_2 \ud y_1  \\ 
& \qquad + \int_0^1  \bbl_1(y_1,+0) F^{-1}(x)F^{-T}(x) e_2 \cdot e_2 \ud y_1 \\
&\qquad -g'(x_1) \int_0^1 \bbl_2(y_1,b)  F^{-1}(x)F^{-T}(x) e_2 \cdot e_2 \ud y_1 \displaybreak[0] \\
& \qquad +g'(x_1) \int_0^1 \bbl_2(y_1,+0)  F^{-1}(x)F^{-T}(x) e_2 \cdot e_2 \ud y_1 \displaybreak[0]\\
& \qquad + \int_0^1 bF^{-T}(x) \grady \bbl_1(y_1,b) \cdot F^{-T}(x)e_2 \ud y_1 \\
& \qquad - g'(x_1) \int_0^1 bF^{-T}(x) \grady \bbl_2(y_1,b) \cdot F^{-T}(x)e_2 \ud y_1.
\end{align*}
When passing to the limit $b\rightarrow \infty$,  the last two integrals vanish since $\grady \bbl$ decays exponentially to $0$. The terms $\bbl_1$ and $\bbl_2$ converge to $\Cbl_1$ and $\Cbl_2$, repectively. 
Thus
\begin{gather*}
\int_0^1 \Bigl( \Cbl_1 + g'(x_1) \Cbl_2\Bigr) F^{-1}(x)F^{-T}(x) e_2 \cdot e_2 \ud y_1 \\ =   \int_0^1 \Bigl( \bbl_1 + g'(x_1)\bbl_2 \Bigr)(y_1,+0)F^{-1}(x)F^{-T}(x) e_2 \cdot e_2 \ud y_1 .
\end{gather*}
Since $F^{-1}(x)F^{-T}(x) e_2 \cdot e_2 =1+g'(x_1)^2$ we can divide the above equation by $1+g'(x_1)^2$, leading to
\begin{align*}
\Cbl_1 + g'(x_1) \Cbl_2 = \int_0^1 (\bbl_1 + g'(x_1)\bbl_2)(y_1,+0) \ud y_1.
\end{align*}
This is equation \eqref{eq:expstabveloc}.
\end{proof}

\subsubsection{Dependence on the parameter $x$}

We summarize the results about the decay of the boundary layer function in the following proposition. Here, we take the dependence on the parameter $x$ explicitly into account.
\begin{prop}
\label{prop:const}
\label{rem:constants}
For the decay function it holds $\Cbl(x)= \int_0^1 \bbl(x,y_1,+0) \ud y_1$.
\end{prop}
\begin{proof}
With the help of \eqref{eq:expstabvelocity}, \eqref{eq:expstabveloc} and Lemma~\ref{lem:blfunc} it is possible to obtain the value of $\Cbl$: Using the exact form of $F(x)$ and $F^{-T}(x)$, the above conditions read
\begin{equation*}
\begin{matrix}
\Cbl_1 &+& g'(x_1)\Cbl_2 &=&   \int_0^1 \bbl_1(y_1,+0) \ud y_1 + g'(x_1)\int_0^1 \bbl_2(y_1,+0) \ud y_1  \\
-g'(x_1) \Cbl_1 & +& \Cbl_2 &=&  -g'(x_1 )\int_0^1 \bbl_1(y_1,+0) \ud y_1 + \int_0^1 \bbl_2(y_1,+0) \ud y_1 =0.
\end{matrix}
\end{equation*}
Since the determinant of the coefficient matrix fulfills
\[
\det \begin{bmatrix}
1 & g'(x_1) \\
-g'(x_1) & 1
\end{bmatrix}=(1+g'(x_1)^2) \not=0,
\]
we can multiply both sides from the left by $\begin{bmatrix}
1 & g'(x_1) \\
-g'(x_1) & 1
\end{bmatrix}^{-1}$ to obtain the result.
\end{proof}

By using the implicit function theorem for Banach spaces, one can show that the boundary layer function $\bbl$ is continuously differentiable in the parameter $x$. By the Proposition above, this also leads to $\Cbl$ having the same property. Denote by $\beta^{\mathrm{bl},i}$ and $\omega^{\mathrm{bl},i}$ the derivatives $\frac{\p}{\p x_i} \bbl$ and $\frac{\p}{\p x_i} \obl$ for $i=1,2$. They fulfill the equations
\begin{subequations}
\begin{empheq}[box=\greybox]{align}
-\divy(F^{-1}(x)F^{-T}(x) \grady \bblo(x,y)) + F^{-T}(x)\grady\oblo(x,y)\qquad \qquad \notag\\
= \divy(\frac{\p}{\p x_1}[F^{-1}(x)F^{-T}(x)] \grady \bbl(x,y))  -  \frac{\p}{\p x_1}F^{-T}(x)\grady\obl(x,y) & \quad \text{ in $Z$} \\
\divy(F^{-1}(x)\bblo(x,y))= - \divy( \frac{\p}{\p x_1}F^{-1}(x) \bbl(x,y)  )&\quad \text{ in $Z$} \label{seq:auxdiffdiv} \\
[ \bble(x) ]_S (y)= 0 & \quad\text{ on $S$} \\
[(F^{-1}(x)F^{-T}(x)\grady \bblo(x)-F^{-1}(x)\oblo(x))e_2]_S (y)= \frac{\p}{\p x_1} K^\mathrm{bl}(x) \notag\\
    -[  \frac{\p}{\p x_1}[F^{-1}(x)F^{-T}(x)] \grady \bbl(x) -  \frac{\p}{\p x_1}F^{-1}(x) \obl(x) ]_S(y) e_2  & \quad\text{ on $S$} \\
\bblo(x,y)=0  \quad \text{ on }  {\textstyle \bigcup_{k=1}^\infty }&  {\textstyle \{  \p Y_S - \binom{0}{k} \} }\\
\bblo, \oblo \text{ are $1$-periodic in $y_1$}
\end{empheq}
\end{subequations}
as well as
\begin{subequations}
\begin{empheq}[box=\greybox]{align}
-\divy(F^{-1}(x)F^{-T}(x) \grady \bblz(x,y)) + F^{-T}(x)\grady\oblz(x,y)= 0& \quad \text{ in $Z$} \\
\divy(F^{-1}(x)\bblz(x,y))= 0 &\quad \text{ in $Z$} \\
[ \bblz(x) ]_S (y)= 0 & \quad\text{ on $S$} \\
[(F^{-1}(x)F^{-T}(x)\grady \bblz(x)-F^{-1}(x)\oblz(x))e_2]_S (y)  = \frac{\p}{\p x_2} K^\mathrm{bl}(x)    & \quad\text{ on $S$} \\
\bblz(x,y)=0  \quad \text{ on }  {\textstyle \bigcup_{k=1}^\infty }&  {\textstyle \{  \p Y_S - \binom{0}{k} \} }\\
\bblz, \oblz \text{ are $1$-periodic in $y_1$}
\end{empheq}
\end{subequations}
due to $\frac{\p}{\p x_2}F^{-1}(x)=0$.
Since $\frac{\p}{\p x_1} F^{-1}(x)= \begin{pmatrix} 0 &-g''(x_1)\\0&0 \end{pmatrix}$, we have that
\begin{gather*}
\int_{\ZBL} \divy( \frac{\p}{\p x_1}F^{-1}(x) \bbl   ) = -g''(x_1) \int_{\ZBL} \frac{\p}{\p y_1} \bbl_2 \ud y =0
\end{gather*}
due to the periodic boundary conditions for $\bbl$. Thus a variant of Proposition~\ref{prop:thetai} below shows that there exists a function $\theta^\beta \in H^2_\mathrm{loc}(\ZBL)$, $\grady \theta^\beta \in L^2(\ZBL)$ such that
\begin{empheq}[]{alignat=2}
\divy(F^{-1}(x)\theta^\beta(y)) &= - \divy( \frac{\p}{\p x_1}F^{-1}(x) \bbl(x,y)  )&\quad&  \text{ in $Z$} \\
\theta^\beta(y)&=0 &&  \text{ on }   {\textstyle \bigcup_{k=1}^\infty \{  \p Y_S - \binom{0}{k} \} }\\
[\theta^\beta]_S (y)& = 0  && \text{ on $ S$} \\
\theta^\beta & \text{ is $1$-periodic in $y_1$.}
\end{empheq}
Using $\theta^\beta$ to correct the divergence in equation \eqref{seq:auxdiffdiv}, one can use the results for the boundary layer functions to obtain an exponential decay towards constants $\Cbli$ (for $\beta^{\mathrm{bl},i}$) and $\Cbloi$ (for $\omega^{\mathrm{bl},i}$), and the identities
\begin{gather*}
\Cbli(x) = \int_0^1  \beta^{\mathrm{bl},i} (x,y_1,+0) \ud y_1 = %\int_0^1  \frac{\p}{\p x_i}\beta^{\mathrm{bl}} (x,y_1,+0) \ud y_1 =
\frac{\p}{\p x_i} \int_0^1  \beta^{\mathrm{bl}} (x,y_1,+0) \ud y_1  = \frac{\p}{\p x_i} \Cbl(x) \\
\Cbloi(x) = \int_0^1  \omega^{\mathrm{bl},i} (x,y_1,+0) \ud y_1 =% \int_0^1  \frac{\p}{\p x_i}\omega^{\mathrm{bl}} (x,y_1,+0) \ud y_1 =
\frac{\p}{\p x_i} \int_0^1  \omega^{\mathrm{bl} }(x,y_1,+0) \ud y_1  = \frac{\p}{\p x_i} \Cblo(x) 
\end{gather*}
hold for $i=1,2$. This shows that the derivative in $x$-direction of $\bbl(x)$, $\obl(x)$ decays to the corresponding derivative of the decay function $\Cbl(x)$, $\Cblo(x)$, and terms like $\gradx(  \bbl(x,y) - H(y_2) \Cbl(x) )$ show the same decay behavior as $ \bbl(x,y) - H(y_2) \Cbl(x)$, leading to similar estimates.

%
%We finish this subsection by showing the relation between the permeability tensor $A$ and the term $\hat{A}_i$:
%\begin{lem}
%\label{lem:continnormal}
%Let $A_i=(\int_{Y^*} w^i \ud y)$ be the $i$-th column of the permeability tensor and set $\hat{A}_{i} = ( \int_S w^i \cdot F^{-T}(x)e_2 \ud \sigma_y)$ as above. Then it holds
%\[
%\hat{A}_i = A_i \cdot F^{-T}(x)e_2 = \Cbl\cdot F^{-T}(x) e_2. 
%\]
%\end{lem}
%
%\begin{proof}
%We argue as above: Integration of the condition $\divy(F^{-1}(x) w^i)=0$ over the set $(0,1 )\times (0,b)$ for $b\in (0,1)$ and application of Stokes' theorem yields due to the periodic boundary conditions
%\[
%\int_0^1 w^i (y_1,0) \cdot F^{-T}(x) e_2 \ud y_1 = \int_0^1 w^i (y_1,b) \cdot F^{-T}(x) e_2 \ud y_1.
%\]
%Now integrate this equation over the interval $(0,1)$ with respect to $b$ to obtain
%\begin{align*}
%\hat{A}_i&= \int_0^1 \int_0^1 w^i (y_1,0) \cdot F^{-T}(x) e_2 \ud y_1  \ud b = \int_0^1 \int_0^1 w^i (y_1,b) \cdot F^{-T}(x) e_2 \ud y_1\ud b \\
%& = \Bigl(\int_{Y^*} w^i \ud y \Bigr) \cdot F^{-T}(x) e_2 = A_i \cdot F^{-T}(x) e_2
%\end{align*}
%The second equality follows due to Lemma \ref{lem:ciblexpl}.
%\end{proof}
%%\formatwo{\vspace{0.0cm}}
%
%

\subsection{Functions for the Correction of the Divergence}
\label{sec:divcorr}

In this section we consider the auxiliary problems associated with the correction of the transformed divergence of $\tilde{\mathcal{U}}^\eps$. Fix $x\in \O$ and define
\[
\Cqb(x)=F(x) \Bigl( \int_{\ZBL} \divx\Bigl(F^{-1}(x)\bigl[ \bbl(x,y)-H(y_2)\Cbl(x) \bigr] \Bigr)\ud y \Bigr) e_2.
\]

\begin{prop}
\label{prop:thetai}
The problem: Find $\theta$ such that
\begin{empheq}[box=\widefbox]{align*}
\divy(F^{-1}(x)\theta(y)) &= \divx\Bigl(F^{-1}(x)\bigl[ \bbl(x,y) -H(y_2)\Cbl(x) \bigr] \Bigr) \quad  \text{ in $Z$} \\
\theta(y)&=0 \qquad\qquad\qquad\qquad\qquad\qquad\quad\quad\: \,  \text{ on }   {\textstyle \bigcup_{k=1}^\infty \{  \p Y_S - \binom{0}{k} \} }\\
[\theta]_S (y)& = \Cqb(x)  \qquad\qquad \qquad\qquad\qquad\qquad\qquad\qquad\quad\quad\,   \text{ on $ S$} \\
\theta & \text{ is $1$-periodic in $y_1$}
\end{empheq}
has at least one solution $\theta\in H^1(Z)^2 \cap \mathcal{C}^\infty_\mathrm{loc}(Z) $.
\end{prop}

\begin{proof}
We argue similarly to Lemma \ref{lem:divsurj} and carry out the following ansatz:
\[
\theta(y)=F(x)\grady \eta(y) + F(x)\Curl_y \xi(y),
\]
where for $\eta$ it holds
\begin{align*}
\Delta_y \eta(y) = \divx\Bigl(F^{-1}(x)\bigl[ \bbl(x,y) -H(y_2)\Cbl(x) \bigr] \Bigr) & \quad \text{ in } Z \\
\grady\eta(y)\cdot \nu= 0  \quad\text{ on } {\textstyle\bigcup_{k=1}^\infty \{  } &{\textstyle \p Y_S - \binom{0}{k} \} }\\
[\grady \eta(y) \cdot e_2]_S(y) =  \int_{\ZBL} \divx(F^{-1}(x)[ \bbl(x,y)-H(y_2)\Cbl(x) ] ) \ud y& \quad\text{ on } S\\
[\eta]_S(y)=0 &\quad \text{ on } S \\
\eta \text{ is $1$-periodic in $y_1$.}
\end{align*}
We investigate solvability in the space $W_D/\R$, with
\[
W_D= \bigl\{  z\in L^2_\mathrm{loc}(\ZBL)\ | \ \grad z \in L^2(\ZBL),\  z \text{ is $1$-periodic in $y_1$}  \bigr\}.
\] 
Define the linear functional
\begin{align*}
\mathcal{L}(\phi) &= \int_{\ZBL} \Bigl[\divx\bigl(F^{-1}(x) [ \bbl(x,y) -H(y_2)\Cbl(x) ] \bigr) \Bigr] \phi(y)\ud y \\ & \qquad - \int_0^1 \Bigl( \int_{\ZBL} \divx\bigl(F^{-1}(x) [ \bbl(x,y) -H(y_2)\Cbl(x) ] \bigr) \ud y  \Bigr)\phi(y_1,0) \ud y_1.
\end{align*}
Since $\mathcal{L}(1)=0$ the linear functional is well defined on $W_D/\R$, and by the properties of $\bbl(x,y) -H(y_2)\Cbl(x) $ it is continuous.
An integration by parts shows that the weak formulation of the above equation reads
\[
\int_{\ZBL} \grady \eta \cdot \grady \phi \ud y = \mathcal{L}(\phi).
\]
Thus we get a solution $\eta$, unique up to a constant.

Next, we search for $\xi$ satisfying
\begin{align*}
\Curl(\xi)\cdot \nu  = -\frac{\p \xi}{\p y_1}=0 & \quad\text{ on }  {\textstyle \bigcup_{k=1}^\infty \{  \p Y_S - \binom{0}{k} \} }\\
\Curl(\xi)\cdot \tau =\frac{\p \xi}{\p y_2}= - \grad \eta \cdot \tau &\quad \text{ on } {\textstyle \bigcup_{k=1}^\infty \{  \p Y_S - \binom{0}{k} \}}.
\end{align*}
Application of the inverse trace theorem \ref{thm:invtrace} to each cell $Y-{\binom{0}{k}}$ in $Z^-$ and setting $\xi=0$ in $Z^+$  yields the existence of $\xi$ similar to the proof of Lemma \ref{lem:divsurj}.
\end{proof}

Since the right hand side of the equation for $\eta$ decays exponentially, we can apply Theorem \ref{thm:expdecay} and obtain an exponential stabilization of $\eta$ towards some constant and a stabilization of $\grad \eta$ towards $0$. As the construction of $\xi$ is local, the decay carries over to this auxiliary function as well, and we obtain an exponential stabilization of $\theta$ to $0$ in $y$ for $|y_2|\longrightarrow \infty$.

Therefore we obtain
\begin{prop}
The above problem has at least one solution $\theta\in V$ such that there exists a $\gamma_0>0$ with
\[
e^{\gamma_0 |y_2|} \theta \in H^1(Z).
\]
\end{prop}

%DECAY IN $x$?

% Bibliographie
\newpage
\bibliography{literatur}

\begin{thebibliography}{OTW95b}

\bibitem[All92]{al_2s2}
Gr{\'{e}}goire Allaire.
\newblock Homogenization and two-scale convergence.
\newblock {\em SIAM Journal for Mathematical Analysis}, 23(6):1482--1518, 1992.

\bibitem[Amo97]{amo_osc}
A.A. Amosov.
\newblock Weak convergence for a class of rapidly oscillating functions.
\newblock {\em Mathematical Notes}, 62(1):122--126, 1997.

\bibitem[BJ67]{bejo_bc}
Gordon~S. Beavers and Daniel~D. Joseph.
\newblock Boundary conditions at a naturally permeable wall.
\newblock {\em Journal of Fluid Mechanic}, 30:197--207, 1967.

\bibitem[Con87]{con_fl2}
Carlos Conca.
\newblock {\'{E}}tude d'un fluide traversant une paroi perfor{\'{e}}e. {II.}
  {C}omportement limite loin de la paroi.
\newblock {\em J Math. pures et appl.}, 66(45-69), 1987.

\bibitem[DB10]{do_trapr}
S{\"o}ren Dobbersch{\"u}tz and Michael B{\"o}hm.
\newblock A transformation approach for the derivation of boundary conditions
  between a curved porous medium and a free fluid.
\newblock {\em Comptes Rendus M{\'{e}}canique}, 338(2):71--77, 2010.

\bibitem[DL90]{dautlions}
Robert Dautray and Jacques-Louis Lions.
\newblock {\em Mathematical Analysis and Numerical Methods for Science and
  Technology: Volume 3: Spectral Theory and Applications}.
\newblock Springer, Berlin, 1990.

\bibitem[Dob09]{do_dipl}
S{\"o}ren Dobbersch{\"u}tz.
\newblock Derivation of boundary conditions at a curved contact interface
  between a free fluid and a porous medium via homogenisation theory, 2009.
\newblock Diplomarbeit, Universit{\"a}t Bremen.

\bibitem[Dob14]{do_stda}
S{\"o}ren Dobbersch{\"u}tz.
\newblock Stokes-{D}arcy coupling for periodically curved interfaces.
\newblock {\em Comptes Rendus M{\'{e}}canique}, 342(2):73--78, 2014.

\bibitem[Dob15]{sd_effflow}
S{\"o}ren Dobbersch{\"u}tz.
\newblock Effective behaviour of a free fluid in contact with a flow in a
  curved porous medium.
\newblock {\em SIAM Journal on Applied Mathematics}, (accepted), 2015.

\bibitem[EFM00]{esfami_filtr}
Magne Espedal, Antonio Fasano, and Andro Mikeli{\'{c}}.
\newblock {\em Filtration in Porous Media and Industrial Applications}, volume
  1734 of {\em Lecture Notes in Mathematics}.
\newblock Springer, Heidelberg, 2000.

\bibitem[Hor97]{hor_homog}
Ulrich Hornung.
\newblock {\em Homogenization and Porous Media}.
\newblock Springer, New York, 1997.

\bibitem[JM96]{jami_bc-fluidpor}
Willie J{\"a}ger and Andro Mikeli{\'{c}}.
\newblock On the boundary condition at the contact interface between a porous
  medium and a free fluid.
\newblock {\em Ann. Sc. Norm. Sup. Pisa, Classe Fis. Mat. Ser. IV},
  23(3):403--465, 1996.

\bibitem[JM00]{jami_ibc-bjs}
Willie J{\"a}ger and Andro Mikeli{\'{c}}.
\newblock On the interface boundary condition of {Beavers}, {Joseph}, and
  {Saffman}.
\newblock {\em SIAM Journal on Applied Mathematics}, 60(4):1111--1127, 2000.

\bibitem[JMN01]{jamine_lamvisc}
Willie J{\"a}ger, Andro Mikeli{\'{c}}, and Nicolas Neuss.
\newblock Asymptotic analysis of the laminar viscous flow over a porous bed.
\newblock {\em SIAM Journal on Scientific Computing}, 22(6):2006--2028, 2001.

\bibitem[LP85]{lapa_phragmlind}
Evgenii~Mikhailovich Landis and Grigori~P. Panasenko.
\newblock A variant of a {Phragm{\'{e}}n}-{Lindel{\"{o}}f} theorem for elliptic
  equations with coefficients that are periodic functions of all variables
  except one.
\newblock In {\em Topics in Modern Mathematics: Petrovskii Seminar No. 5},
  pages 133--172, New York, 1985. Consultant Bureau.

\bibitem[MCM12]{mik_effpress}
Anna Marciniak-Czochra and Andro Mikeli{\'{c}}.
\newblock Effective pressure interface law for transport phenomena between an
  unconfined fluid and a porous medium using homogenization.
\newblock {\em Multiscale Modeling {\&} Simulation}, 10(2):285--305, 2012.

\bibitem[Mik91]{mi_nsgrain}
Andro Mikeli{\'{c}}.
\newblock Homogenization of nonstationary {Navier-Stokes} equations in a domain
  with a grained boundary.
\newblock {\em Annali Matematica pura ed applicata}, 158:167--179, 1991.

\bibitem[Ngu89]{ng_2s1}
Gabriel Nguetseng.
\newblock A general convergence result for a functional related to the theory
  of homogenization.
\newblock {\em SIAM Journal for Mathematical Analysis}, 20(3):608--623, 1989.

\bibitem[NR00]{radu_decay}
Maria Neuss-Radu.
\newblock A result on the decay of the boundary layers in the homogenization
  theory.
\newblock {\em Asymptotic Analysis}, 23:313--328, 2000.

\bibitem[NR01]{radu_bbehav}
Maria Neuss-Radu.
\newblock The boundary behaviour of a composite material.
\newblock {\em Mathematical Modelling and Numerical Analysis}, 35(3):407--435,
  2001.

\bibitem[OI81]{olio_behavinf}
Olga~Arsen'evna Ole{\u{\i}}nik and G.A. Iosif'jan.
\newblock On the behaviour at infinity of solutions of second order elliptic
  equations in domains with noncompact boundary.
\newblock {\em Math. USSR Sbornik}, 40(4):527--548, 1981.

\bibitem[OTW95a]{ochtapwhi_momtransf1}
J.~Alberto Ochoa-Tapia and Stephen Whitaker.
\newblock Momentum transfer at the boundary between a porous medium and a
  homogeneous fluid -- {I. Theoretical} development.
\newblock {\em International Journal of Heat and Mass Transfer},
  38(14):2635--2646, 1995.

\bibitem[OTW95b]{ochtapwhi_momtransf2}
J.~Alberto Ochoa-Tapia and Stephen Whitaker.
\newblock Momentum transfer at the boundary between a porous medium and a
  homogeneous fluid -- {II. Comparison} with experiment.
\newblock {\em International Journal of Heat and Mass Transfer},
  38(14):2647--2655, 1995.

\bibitem[Saf71]{sa_bcpor}
Philip~Geoffrey Saffman.
\newblock On the boundary condition at the interface of a porous medium.
\newblock {\em Stud. Appl. Math.}, 1:93--101, 1971.

\bibitem[SP80]{sanpal}
Enrique Sanchez-Palencia.
\newblock {\em Non-Homogeneous Media and Vibration Theory}, volume 127 of {\em
  Lecture Notes in Physics}.
\newblock Springer, Berlin, Heidelberg, New York, 1980.

\bibitem[Tem77]{temam}
Roger Temam.
\newblock {\em Navier-Stokes Equations. Theory and Numerical Analysis}.
\newblock North-Holland, Amsterdam, 1977.

\bibitem[Ver07]{verfuerth}
R{\"u}diger Verf{\"u}hrt.
\newblock Computational fluid dynamics.
\newblock Ruhr-Universit{\"a}t Bochum, 2007.

\bibitem[Wlo92]{wloka}
Joseph Wloka.
\newblock {\em Partial Differential Equations}.
\newblock Cambridge University Press, Cambridge, 1992.

\bibitem[Zei88]{zeidler4}
Eberhard Zeidler.
\newblock {\em Nonlinear Functional Analysis and its Applications
  IV.~Application to Mathematical Physics}.
\newblock Springer, New York, 1988.

\end{thebibliography}

\end{document}